\documentclass{lmcs}
\pdfoutput=1
\usepackage[utf8]{inputenc}

\usepackage{lastpage}
\lmcsdoi{20}{2}{8}
\lmcsheading{}{\pageref{LastPage}}{}{}%
{Oct.~25,~2022}{May~16,~2024}{}

\usepackage{hyperref}
\usepackage{cleveref}
\usepackage{graphicx}
\usepackage{proof-at-the-end}
\pgfkeys{/prAtEnd/default/.style={restate, proof at the end, stared, no link to proof}}
\pgfkeys{/prAtEnd/allend/.style={all end, no link to theorem}}
\usepackage{amssymb}
\usepackage{amsmath}
\usepackage{proof}
\usepackage{amsthm}
\usepackage{arydshln}
\usepackage{stmaryrd}
\usepackage{mathpartir}
\usepackage{upgreek}
\usepackage{xcolor}

\makeatletter
\let\c@lem\undefined%
\let\c@cor\undefined%
\let\c@prop\undefined%
\let\c@fact\undefined%
\let\c@lemC\undefined%
\let\c@exa\undefined%
\let\c@defi\undefined%
\let\c@thm\undefined%
\let\c@rem\undefined%
\let\c@thmC\undefined%
\makeatother

\theoremstyle{plain}
\declaretheorem[name=Theorem,numberwithin=section]{thm}
\newtheorem{lem}[thm]{Lemma}
\newtheorem{cor}[thm]{Corollary}
\newtheorem{prop}[thm]{Proposition}
\newtheorem{fact}[thm]{Fact}
\newtheorem{rem}[thm]{Remark}

\theoremstyle{thmC}
\newtheorem{lemC}[thm]{Lemma}
\newtheorem{thmC}[thm]{Theorem}

\theoremstyle{definition}
\newtheorem{exa}[thm]{Example}
\newtheorem{defi}[thm]{Definition}

\crefname{prop}{Proposition}{Propositions}
\crefname{thm}{Theorem}{Theorems}
\crefname{clm}{Claim}{Claims}
\crefname{lem}{Lemma}{Lemmas}
\crefname{lemC}{Lemma}{Lemmas}
\crefname{thmC}{Theorem}{Theorems}
\crefname{rem}{Remark}{Remarks}
\crefname{cor}{Corollary}{Corollaries}
\crefname{exa}{Example}{Examples}
\crefname{defi}{Definition}{Definitions}
\crefname{fact}{Fact}{Facts}

\crefname{section}{Section}{Sections}
\crefname{figure}{Figure}{Figures}

\usepackage{xspace}

\definecolor{citecolor}{rgb}{0.0,0.4,0.0}
\definecolor{urlcolor}{rgb}{0.0,0.0,0.4}
\definecolor{linkcolor}{rgb}{0.0,0.0,0.4}
\hypersetup{
  pdftex,
  plainpages=false,
  pdfpagelabels,
  linktoc=all,
  bookmarks=true,
  bookmarksopen=true,
  breaklinks=true,
  colorlinks,
  linkcolor=linkcolor,
  urlcolor=urlcolor,
  citecolor=citecolor}

\usepackage{tikz}
\usepackage[all]{xy}

\newcommand\eqdef\triangleq

\newcommand\pspace{\textsc{PSpace}}
\newcommand\expspace{\textsc{ExpSpace}}

\makeatletter
\providecommand{\leftsquigarrow}{%
  \mathrel{\mathpalette\reflect@squig\relax}%
}
\newcommand{\reflect@squig}[2]{%
  \reflectbox{$\m@th#1\rightsquigarrow$}%
}
\makeatother

\newcommand{\KA}{\ensuremath{\mathsf{KA}}\xspace}
\newcommand{\BA}{\ensuremath{\mathsf{BA}}\xspace}
\newcommand{\DL}{\ensuremath{\mathsf{DL}}\xspace}

\newcommand{\KAT}{\ensuremath{\mathsf{KAT}}\xspace}

\newcommand{\KATF}{\ensuremath{\mathsf{KATF}}\xspace}

\newcommand\lang[1]{\left\llbracket#1\right\rrbracket}

\newcommand\ttop{\mathbb T}
\newcommand\full{\mathbb F}

\newcommand\cl[1]{{#1}^\star}
\newcommand\clH[1]{H^\star_{#1}}
\newcommand\Hlt[1]{H_{{{<}#1}}}
\newcommand\clHlt[1]{\clH{{{<}#1}}}
\newcommand\sem[2]{\cl{#1}\!\lang{#2}}
\newcommand\semH[2]{\sem{H_{#1}}{#2}}

\newcommand\proves[1][]{#1 \vdash}         % abbreviated version
\newcommand\KAproves[1][]{\KA_{#1} \vdash} % explicit version
\newcommand\BAproves{\vdash_\BA}
\newcommand\DLproves{\vdash_\DL}
\newcommand\KATproves{\KAT\vdash}

\newcommand\A{\mathcal{A}}
\newcommand\e{\mathbf{e}}
\newcommand\br{\mathbf{r}}
\newcommand\f{\mathbf{f}}

\newcommand\G{\mathcal{G}}
\newcommand\GS{\mathcal{GS}}
\newcommand\aalpha{\boldsymbol\alpha}

\newcommand\pphi{\boldsymbol\phi}
\newcommand\phis{\boldsymbol\phi}
\newcommand\psis{\boldsymbol\psi}

\newcommand{\termst}[1]{\ensuremath{{\mathcal{T}_{\scriptscriptstyle#1}}}}
\newcommand{\termska}{\termst{}}
\newcommand{\termskat}{\termst{\KAT}}
\newcommand{\termsba}{\termst{\BA}}
\newcommand{\termsdl}{\termst{\DL}}

\newcommand\hset[1]{\ensuremath{\mathsf{#1}}\xspace}
\newcommand{\hkat}{\hset{kat}}
\newcommand{\hnetkat}{\hset{netkat}}
\newcommand{\hbool}{\hset{bool}}
\newcommand{\hkapt}{\hset{kapt}}
\newcommand{\hdl}{\hset{dl}}
\newcommand{\hgluedl}{\hset{glue''}}
\newcommand{\hatomdl}{\hset{atm'}}
\renewcommand{\hglue}{\hset{glue}}
\newcommand{\hatom}{\hset{atm}}
\newcommand{\hkao}{\hset{kao}}
\newcommand{\hkabo}{\hset{kabo}}
\newcommand{\hgluekao}{\hset{glue'}}
\newcommand{\hcontr}{\hset{ctr}}

\newcommand{\At}{\mathsf{At}}

\newcommand{\ltr}[1]{\mathtt{#1}}

\newcommand{\id}{\mathrm{id}}
\newcommand{\pipe}{\;\;|\;\;}

\newcommand{\pow}{\mathcal{P}}

\newcommand\tothinklater[1]{}

\newcommand{\hkatf}{\hset{katf}}
\newcommand{\hkatc}{\hset{katc}}
\newcommand{\termskatf}{\termst{\KATF}}
\newcommand{\htop}{\hset{t}}
\newcommand{\hful}{\hset{f}}
\newcommand{\hcnv}{\hset{cnv}}
\newcommand{\hc}{\hset{c}}
\newcommand{\hi}{\hset{i}}
\newcommand{\hs}{\hset{s}}
\newcommand{\hfull}{\hset{full}}

\newcommand{\inl}{\mathsf{in}_l}
\newcommand{\inr}{\mathsf{in}_r}

\newcommand\NN{\mathbb{N}}
\newcommand\set[1]{\left\{#1\right\}}
\newcommand\paren[1]{\left(#1\right)}
\newcommand\tuple[1]{\left\langle #1\right\rangle}

\newcommand\dup{\mathsf{dup}}

\newcommand\closeoverlap[5]{
      \begin{align*}
        \xymatrix@C=3em{
        #3 & \ar@{~>}[l]_{#1}
        #4 & \ar@{~>}[l]_{#2}
        #5 \ar@{=}@/^1em/[ll]}\\
      \end{align*}}
\newcommand\closeoverlapi[6]{
      \begin{align*}
        \xymatrix@C=3em@R=0em{
        #4 & \ar@{~>}[l]_{#1}
        #5 & \ar@{~>}[l]_{#2}
        #6 \ar@{~>}@/^2em/[ll]_{#3}}\\
      \end{align*}}
\newcommand\closeoverlapii[8]{
    \begin{align*}
      \xymatrix@R=.6em{
       #5 & \ar@{~>}[l]_{#1}
       #6 & \ar@{~>}[l]_{#2}
       #7 \ar@{~>}@/^.8em/[ld]^{#4}\\
      &#8 \ar@{~>}@/^.8em/[lu]^{#3}\\}
    \end{align*}}

\newcommand{\neworrenewcommand}[1]{\providecommand{#1}{}\renewcommand{#1}}
\newcommand\closeoverlapiii[9]{
    \neworrenewcommand{\fcloseoverlapiii}[1]{
    \begin{align*}
      \xymatrix@R=1.5em{
      #6 & \ar@{~>}[l]_{#1}
      #7 & \ar@{~>}[l]_{#2}
      #8 \ar@{~>}[d]^{#5}\\
      #9 \ar@{~>}[u]^{#3}
      &&##1\ar@{~>}[ll]^{#4}\\}
    \end{align*}}
    \fcloseoverlapiii
    }
\newcommand\closeoverlapiiii[9]{
    \neworrenewcommand{\fcloseoverlapiiii}[3]{
    \begin{align*}
      \xymatrix@R=1.5em{
      #7 & \ar@{~>}[l]_{#1}
      #8 & \ar@{~>}[l]_{#2}
      #9 \ar@{~>}[d]^{#6}\\
      ##1 \ar@{~>}[u]^{#3}
      &##2\ar@{~>}[l]^{#4}
      &##3\ar@{~>}[l]^{#5}\\}
    \end{align*}}
    \fcloseoverlapiiii
    }

\begin{document}

\title[On Tools for Completeness of Kleene Algebra with Hypotheses]%
{On Tools for Completeness \\ of Kleene Algebra with Hypotheses}%
  \thanks{%
    This paper is an extended version of the paper with the same title which appeared in Proc. RAMiCS 2021~\cite{prw:ramics21:mkah}; we summarise the additions at the end of the Introduction.
    This work has been supported by the ERC (CoVeCe, grant No 678157), by the LABEX MILYON (ANR-10-LABX-0070), within the program ANR-11-IDEX-0007, and by the project ‘Mode(l)s of Verification and Monitorability’ (MoVeMent) (grant No 217987) of the Icelandic Research Fund.%
}

\author[D.~Pous]{Damien Pous\lmcsorcid{0000-0002-1220-4399}}[a]
\author[J.~Rot]{Jurriaan Rot\lmcsorcid{0000-0002-1404-6232}}[b]
\author[J.~Wagemaker]{Jana Wagemaker\lmcsorcid{0000-0002-8616-3905}}[c]

% affiliation 1 (automatically numbered a)
\address{CNRS, LIP, ENS de Lyon}
\email{Damien.Pous@ens-lyon.fr}

% affiliation 2 (automatically numbered b)
\address{Radboud University, Nijmegen}
\email{Jurriaan.Rot@ru.nl}

% affiliation 3 (automatically numbered c)
\address{Reykjavik University, Reykjavik}
\email{janaw@ru.is}

\begin{abstract}
  In the literature on Kleene algebra, a number of variants have been proposed which impose additional structure specified by a theory, such as Kleene algebra with tests (KAT) and the recent Kleene algebra with observations (KAO), or make specific assumptions about certain constants, as for instance in NetKAT.
  Many of these variants fit within the unifying perspective offered by \emph{Kleene algebra with hypotheses}, which comes with a canonical language model constructed from a given set of hypotheses. For the case of KAT, this model corresponds to the familiar interpretation of expressions as languages of guarded strings.

  A relevant question therefore is whether Kleene algebra together with a given set of hypotheses is complete with respect to its canonical language model.
  In this paper, we revisit, combine and extend existing results on this question to obtain tools for proving completeness in a modular way.
        
  We showcase these tools by giving new and modular proofs of completeness for KAT, KAO and NetKAT,
  and we prove completeness for new variants of KAT: KAT extended with a constant for the full relation, KAT extended with a converse operation, and a version of KAT where the collection of tests only forms a distributive lattice.

\end{abstract}

\maketitle

\section{Introduction}
\label{sec:intro}
Kleene algebras (KA)~\cite{Kle56,Conway71} are algebraic structures involving an iteration operation, Kleene star, corresponding to reflexive-transitive closure in relational models and to language iteration in language models.
Its axioms are complete w.r.t. relational models and language models~\cite{Kozen91,Krob91a,Boffa90}, and the resulting equational theory is decidable via automata algorithms (in fact, \pspace-complete~\cite{MS72,HUNT1976}).

These structures were later extended in order to deal with common programming constructs. For instance, Kleene algebras with tests (KAT)~\cite{kozen:97:kat}, which combine Kleene algebra and Boolean algebra, make it possible to represent the control flow of while programs. Kleene star is used for while loops, and Boolean tests are used for the conditions of such loops, as well as the conditions in if-then-else statements.
Again, the axioms of KAT are complete w.r.t. appropriate classes of models, and its equational theory remains in \pspace. Proving so is non-trivial: Kozen's proof reduces completeness of KAT to completeness of KA, via a direct syntactic transformation on terms.

Another extension is Concurrent Kleene algebra (CKA)~\cite{HoareMSW11}, where a binary operator for parallelism is added. The resulting theory is characterised by languages of pomsets rather than languages of words, and is \expspace-complete~\cite{bps:concur17:cka}. Trying to have both tests and concurrency turned out to be non-trivial, and called for yet another notion: Kleene algebras with observations (KAO)~\cite{kao-2019}, which are again complete w.r.t.\ appropriate models, and decidable.

When used in the context of program verification, e.g., in a proof assistant, such structures make it possible to write algebraic proofs of correctness, and to mechanise some of the steps: when two expressions $e$ and $f$ representing two programs happen to be provably equivalent in KA, KAT, or KAO, one does not need to provide a proof, one can simply call a certified decision procedure~\cite{bp:itp10:kacoq,KraussN12,pous:itp13:ra}.
However, this is often not enough~\cite{kozenp00:kat:compiler:opts,angusk01:kat:schemato,hardink02:kat:hypotheses}: most of the time, the expressions $e$ and $f$ are provably equal only under certain assumptions on their constituents. For instance, to prove that $(a+b)^*$ and $a^*b^*$ are equal, one may have to use the additional assumption $ba= ab$.
In other words, one would like to prove equations under some assumptions, to have algorithms for the Horn theory of Kleene algebra and its extensions rather than just their equational theories.

Unfortunately, those Horn theories are typically undecidable~\cite{kozen96:kat:commutativitycond,kozen02:ka:complexity}, even with rather restricted forms of hypotheses (e.g., commutation of specific pairs of letters, as in the above example).
Nevertheless, important and useful classes of hypotheses can be `eliminated', by reducing to the plain and decidable case of the equational theory. This is, for instance, the case for \emph{Hoare hypotheses}~\cite{kozen00:kat:hoare}, of the shape $e=0$, which allow to encode Hoare triples for partial correctness in KAT.

In some cases, one wants to exploit hypotheses about specific constituents (e.g, $a$ and $b$ in the above example). In other situations, one wants to exploit assumptions on the whole structure. For instance, in commutative Kleene algebra~\cite{redko1964algebra,Conway71,Brunet19}, one assumes that the product is commutative everywhere.

Many of these extensions of Kleene algebra (KAT, KAO, commutative KA, specific hypotheses) fit into the generic framework of Kleene algebra with hypotheses~\cite{dkpp:fossacs19:kah}, providing in each case a canonical model in terms of closed languages.

We show that we recover standard models in this way, and we provide tools to establish completeness and decidability of such extensions, in a modular way. The key notion is that of \emph{reduction} from one set of hypotheses to another. We summarise existing reductions, provide new ones, and design a toolbox for combining those reductions together. We use this toolbox in order to obtain new and modular proofs of completeness for KAT, KAO and NetKAT~\cite{AndersonFGJKSW14}; to prove completeness of KAT with a full element and KAT with a converse operation, two extensions that were formerly proposed for KA alone, respectively in~\cite{pw:concur22:katop} and~\cite{BES95,EB95}; and to prove completeness for the variant of KAT where tests are only assumed to form a distributive lattice.

\subsection*{Outline}

We first recall Kleene algebra, Kleene algebra with hypotheses, and the key concept of reduction (\cref{sec:prelim}).
Then we design a first toolbox for reductions (\cref{sec:tools}): a collection of primitive reductions, together with lemmas allowing one to combine them together in a modular fashion.
We exemplify those tools by using them to give a new proof of completeness for KAT (\cref{sec:kat}).

After this first example, we develop more systematic tools for reductions (\cref{sec:moretools}): lemmas to combine more than two reductions at a time, and lemmas to prove the resulting side-conditions more efficiently.
We finally demonstrate the flexibility of our approach in a series of examples of increasing complexity:
KAO (\cref{sec:kao}),
KAT with a full element (\cref{sec:katf}),
KAT with converse (\cref{sec:katc}),
KA with positive tests (KAPT---\cref{sec:kapt}), and
NetKAT (\cref{sec:netkat}).

The first appendix contains an abstract theory of least closures in complete lattices, which we base our  toolbox on (\cref{app:closures}). Two subsequent appendices contain proofs omitted from the main text (a direct soundness proof for Kleene algebra with hypotheses---\cref{app:soundness}, and proofs related to guarded strings for KAT---\cref{app:kat:gs}).

\subsection*{Differences with the conference version}
This paper is a revised and extended version of an earlier paper in the proceedings of RaMiCS 2021~\cite{prw:ramics21:mkah}.
Besides various improvements, and inclusion of most proofs, here are the main differences w.r.t.\ the conference version.

\Cref{ssec:automata:red} is new, it provides a general technique to construct reductions via finite automata. We use it to establish some of the basic reductions from \cref{ssec:basic:red}.

With the exception of \cref{prop:cup}, \cref{sec:moretools} on advanced tools is new (e.g., \cref{prop:gsos} or development on overlaps). This section makes it possible to present most examples using `tables' (\cref{table:kabo,table:katf,table:katc,table:kapt,table:netkat}), which we consider as a new contribution concerning methodology.
\cref{lem:basic:e1}, which makes it possible to obtain reductions for hypotheses of the shape $e\leq 1$ is also new; we need it to deal with NetKAT.

The examples of KAT with a full element, KAT with converse, and NetKAT, are new (\cref{sec:katf,sec:katc,sec:netkat}); the example of KAPT (\cref{sec:kapt}) is presented more uniformly, thanks to the new tools.
\cref{app:closures} on the abstract theory of least closures is entirely new; it makes it possible to get rid of all proofs by transfinite induction that we used in the conference version. \cref{app:kat:gs} was publicly available via the author version of~\cite{prw:ramics21:mkah}, but not officially published.

\subsection*{Preliminaries, notation}

We write $1$ for the empty word and $uv$ for the concatenation of two words $u,v$.
Given a set $X$, subsets of $X$ form a complete lattice $\tuple{\pow(X),\subseteq,\bigcup}$.
Given two functions $f,g$, we write $f\circ g$ for their composition: $(f\circ g)(x)=f(g(x))$. We write $\id$ for the identity function.
Functions into a complete lattice, ordered pointwise, form a complete lattice where suprema are computed pointwise. We reuse the same notations on such lattices: for instance, given two functions $f,g\colon Y\to\pow(X)$, we write $f\subseteq g$ when $\forall y\in Y, f(y)\subseteq g(y)$, or $f\cup g$ for the function mapping $y\in Y$ to $f(y)\cup g(y)$.

A function $f\colon \pow(X)\to\pow(X)$ is \emph{monotone} if for all $S,S'\in \pow(X)$, $S\subseteq S'$ implies $f(S)\subseteq f(S')$; it is a \emph{closure} if it is monotone and satisfies $\id\subseteq f$ and $f\circ f \subseteq f$. When $f$ is a closure, we have $f\circ f=f$, and for all $S,S'$, $S\subseteq f(S')$ iff $f(S)\subseteq f(S')$. These notions are generalised to arbitrary complete lattices in \cref{app:closures}.

\section{Kleene algebra with hypotheses, closures, reductions}
\label{sec:prelim}

A \emph{Kleene algebra}~\cite{Conway71,Kozen91} is a tuple $(K,+,\cdot,^*,0,1)$ such that $(K,+,\cdot,0,1)$ is an idempotent semiring, and $^*$ is a unary operator on $K$ such that for all $x,y\in K$ the following axioms are satisfied:
\begin{mathpar}
1 + x \cdot x^* \leq x^*
\and
x + y \cdot z \leq z \Rightarrow y^* \cdot x \leq z
\and
x + y \cdot z \leq y \Rightarrow x \cdot z^* \leq y
\end{mathpar}
Here, as later in the paper, we write $x \leq y$ as a shorthand for $x + y = y$. Given the idempotent semiring axioms, $\leq$ is a partial order in every Kleene algebra, and all operations are monotone w.r.t. that order. The seemingly missing law $1 + x^* \cdot x \leq x^*$ is derivable from these axioms.

Languages over some alphabet, as well as binary relations on a given set, are standard examples of Kleene algebras.

\medskip

We let $e,f$ range over regular expressions over an alphabet $\Sigma$, defined by:
\begin{align*}
  e,f ::= e + f \pipe e \cdot f \pipe e^* \pipe 0 \pipe 1 \pipe a \in \Sigma
\end{align*}
We write $\termska(\Sigma)$ for the set of such expressions, or simply $\termska$ when the alphabet is clear from the context.
We usually write $ef$ for $e\cdot f$, and $e^+$ for $e\cdot e^*$.
Given alphabets $\Sigma$ and $\Gamma$, a function $h\colon\Sigma\to\termska(\Gamma)$ extends
uniquely to a homomorphism $h \colon \termska(\Sigma)\to\termska(\Gamma)$, referred to as the \emph{homomorphism generated by $h$}.
As usual, every regular expression $e$ gives rise to a language $\lang{e}\in \pow(\Sigma^*)$.
Given regular expressions $e,f$, we write $\KAproves e=f$ when $e=f$ is derivable from the axioms of Kleene algebra. (Equivalently, when the equation $e=f$ holds universally, in all Kleene algebras.)

The central theorem of Kleene algebra is the following:
\begin{thm}[Soundness and Completeness of KA~\cite{Kozen91,Krob91a,Boffa90}]%
  \label{thm:ka}
  For all regular expressions $e,f\in\termska$, we have:
  \begin{align*}
    \KAproves e=f \quad\text{if and only if}\quad \lang{e}=\lang{f}\,.
  \end{align*}
\end{thm}
As a consequence, the equational theory of Kleene algebras is decidable.

\subsection{Hypotheses}
\label{ssec:hyps}

Our goal is to extend this result to the case where we have additional hypotheses on some of the letters of the alphabet, or axioms restricting the behaviour of certain operations.
Those are represented by sets of \emph{inequations}, i.e., pairs $(e,f)$ of regular expressions written $e\leq f$ for the sake of clarity.
Given a set $H$ of such inequations, we write $\KAproves[H] e\leq f$ when the inequation $e\leq f$ is derivable from the axioms of Kleene algebra and the hypotheses in $H$ (similarly for equations).
By extension, we write $\KAproves[H]H'$ when $\KAproves[H] e\leq f$ for all $e\leq f$ in $H'$.
Since the ambient theory will systematically be KA, we will sometimes abbreviate $\KAproves[H]H'$ as $\proves[H]H'$, to alleviate notation.

Note that we consider letters of the alphabet as constants rather than variables.
In particular, while we have $\KAproves[ba\leq ab] (a+b)^*\leq a^*b^*$, we do not have $\KAproves[ba\leq ab] (a+c)^*\leq a^*c^*$.
Formally, we use a notion of derivation where there is no substitution rule, and where we have all instances of Kleene algebra axioms as axioms.
When we want to consider hypotheses that are universally valid, it suffices to use all their instances. For example, to define commutative Kleene algebra, we simply use the infinite set $\set{ef\leq fe\mid e,f\in\termska}$.

\subsection{Closures associated to hypotheses}
\label{ssec:clo}

We associate a canonical language model to Kleene algebra with a set of hypotheses $H$,
defined by \emph{closure} under $H$~\cite{dkpp:fossacs19:kah}.

For $u,v\in\Sigma^*$ and $L\subseteq \Sigma^*$, let $uLv\eqdef\set{uxv\mid x\in L}$.

\begin{defi}[$H$-closure]
  \label{def:clo}
  Let $H$ be a set of hypotheses and let $L\subseteq \Sigma^*$ be a language. The \emph{$H$-closure of $L$}, denoted as
  $\cl H(L)$, is the smallest language containing $L$ such that for all $e\leq f\in H$ and $u,v\in\Sigma^*$,
  if $u\lang{f}v\subseteq \cl H(L)$, then $u\lang{e}v\subseteq \cl H(L)$.
\end{defi}
For every set of hypotheses $H$, $\cl H$ is indeed a closure (also see \cref{ssec:clo:props}).

In many cases, $H$ will consist of inequations whose members are words rather than arbitrary expressions.
In that case, there is an intuitive string rewriting characterisation of the function $\cl H$.
Given words $u,v$, we write $u\leftsquigarrow_H v$ when $u$ can be obtained from $v$ by replacing an occurrence as a subword of the right-hand side of an inequation of $H$ by its left-hand side. Further write $\leftsquigarrow^*_H$ for the reflexive-transitive closure of $\leftsquigarrow_H$.
Then we have $u\in\cl H(L)$ iff $u\leftsquigarrow^*_H v$ for some word $v\in L$.
\begin{exa}
  Let $H=\set{ab\leq bc}$.
  We have $baab\leftsquigarrow_H babc \leftsquigarrow_H bbcc$, whence $baab \in \cl H(\set{bbcc})$.
  For $L=\lang{bc^*}$, we have $\cl H(L)=\lang {a^*bc^*}$.
  \qed
\end{exa}

This notion of closure gives a \emph{closed interpretation} of regular expressions, $\sem H -$, for which $\KA_H$ is sound:
\begin{theoremEnd}[default,category=snd]{thm}
  \label{thm:soundness}
  $\KAproves[H] e = f$ implies $\sem H e = \sem H f$.
\end{theoremEnd}
\begin{proofEnd}
  By induction on the derivation.
  \begin{itemize}
  \item For KA equational axioms, we actually have $\lang e=\lang f$.
  \item For an inequation $e\leq f\in H$, we have $\lang e\subseteq H\lang f$ by definition, whence $\lang e\subseteq \sem H f$.
  \item The reflexivity, symmetry, and transitivity rules are easy to handle.
  \item The contextuality rules are dealt with \cref{lem:presoundness}. For instance, for $+$,
    suppose $\KAproves[H]e+f=e'+f'$ is obtained from $\KAproves[H]e=e'$ and $\KAproves[H]f=f'$, so that $\sem H e = \sem H {e'}$ and $\sem H f = \sem H{f'}$ by induction. We derive:
    \begin{align*}
      \sem H {e+f}
      \tag{\cref{lem:presoundness}(1)}
      &= \cl H (\sem H e+\sem H f)\\
      \tag{IH}
      &= \cl H (\sem H{e'}+\sem H{f'})\\
      \tag{\cref{lem:presoundness}(1)}
      &= \sem H {e'+f'}
    \end{align*}
  \item It remains to deal with the two implications of KA.
    Suppose $\KAproves[H]f^*e\leq g$ is obtained from $\KAproves[H]e+fg\leq g$, so that
    $\lang{e+fg}\subseteq \sem H g$ by induction. We have to show $\lang{f^*e}=\lang{f}^*\lang e\subseteq \sem H g$, which we obtain from
    \begin{align*}
      \lang e+\lang{f}\sem H{g}
      &\subseteq \lang e+\sem H{f}\sem H {g}\\
      \tag{\cref{lem:con:clo:app}}
      &\subseteq \lang e+\sem H{fg}\\
      &\subseteq \sem H{e+fg}
      \tag{IH}
      \subseteq \sem H{g}
    \end{align*}
    The other implication is dealt with symmetrically.\qedhere
  \end{itemize}
\end{proofEnd}
\begin{proof}
  This is basically \cite[Theorem 2]{dkpp:fossacs19:kah}, which is proved there by constructing a model of $\KA_H$. We provide a direct proof in \cref{app:soundness}, by induction on the derivation.
\end{proof}
In the sequel, we shall prove the converse implication, completeness, for specific choices of $H$:
we say that $\KA_H$ is \emph{complete} if for all expressions $e,f$:
\begin{align*}
  \sem H e = \sem H f \quad\text{implies}\quad \KAproves[H] e = f\,.
\end{align*}

\begin{rem}
  Thanks to the presence of sum in the syntax of regular expressions (so that $\KAproves[H] e \leq f$ iff $\KAproves[H] e+f = f$), and to the fact that $\cl H$ is a closure, we have the following counterpart to \cref{thm:soundness} for soundness w.r.t inequations:
  \begin{align*}
    \KAproves[H] e \leq f \quad\text{implies}\quad \lang e \subseteq \sem H f\,.
  \end{align*}
  Similarly, $\KA_H$ is complete if and only if for all expressions $e,f$,
  \begin{align*}
    \tag*\qed
    \lang e \subseteq \sem H f \quad\text{implies}\quad \KAproves[H] e \leq f\,.
  \end{align*}
\end{rem}

We could hope that completeness always holds, notably because the notion of closure is invariant under inter-derivability of the considered hypotheses, as a consequence of the following lemma:
\begin{lemC}[{\cite[Lemma 4.10]{KappeB0WZ20}}]
  \label{lem:clo:incl}
  Let $H$ and $H'$ be sets of hypotheses. If $\proves[H]H'$ then $\cl{H'} \subseteq \cl H$.
\end{lemC}
\noindent
(There, $\cl{H'} \subseteq \cl H$ means pointwise inclusion of functions, i.e.,\ for all $L$, $\cl{H'}(L) \subseteq \cl H(L)$.)

Another property witnessing the canonicity of the $H$-closed language interpretation is that it is sound and complete w.r.t.\ $*$-\emph{continuous} models~\cite[Theorem~2]{dkpp:fossacs19:kah}.

Unfortunately, there are concrete instances for which $\KA_H$ is known to be incomplete.
For instance, there is a finitely presented monoid (thus a finite set $H_0$ of equations) such that $\set{(e,f)\mid\sem{H_0}e=\sem{H_0}f}$ is not recursively enumerable~\cite[Theorem~1]{KozenM14}.
Since derivability in $\KA_H$ is recursively enumerable as soon as $H$ is, $\KA_{H_0}$ cannot be complete.

For a more concrete counter-example, consider $H_1\eqdef\set{apb\leq p,~u\leq a,u\leq p,u\leq b}$, where $a,p,b,u$ are four distinct letters. The closed language $\sem{H_1}p$ contains all words of the shape $a^npb^n$ for some $n\in\NN$ thanks to the first inequation in $H_1$, and thus all words of the shape $u^{2n+1}$ for some $n\in\NN$ thanks to the other three inequations. It follows that $(uu)^*u\subseteq \sem{H_1}p$. However, $(uu)^*u\leq p$ is not derivable in $\KA_{H_1}$. The intuition is that the star induction rules in \KA do not make it possible to reach the middle of a word. Formally, the action lattice provided by Kuznetsov in~\cite[Section~3]{Kuznetsov18} gives a counter-model: consider his Lemma~3.6 and interpret $p$ as $p$, $a$ as $p/q$, $b$ as $p\backslash q$, and $u$ as $p\land q\land a\land b$. 

\medskip

Before turning to techniques for proving completeness, let us describe the closed interpretation of regular expressions for two specific choices of hypotheses.

\begin{exa}
  \label{ex:comm:ka}
Let us first consider \emph{commutative Kleene algebra}, obtained using the set $\set{ef \leq fe \mid e,f \in \termska(\Sigma)}$,
as also mentioned in the introduction. Under Kleene algebra axioms, this set is equiderivable with its restriction to letters, $C = \set{ab \leq ba \mid a,b \in \Sigma}$ (a consequence of~\cite[Lemma~4.4]{angusk01:kat:schemato}).
The associated closure can be characterised as follows:
\begin{align*}
  \cl C(L) = \set{w \in \Sigma^* \mid \exists v \in L, \, |w|_x = |v|_x \text{ for all }x \in \Sigma }
\end{align*}
where $|w|_x$ denotes the number of occurrences of $x$ in $w$. Thus, $w \in \cl C(L)$
if it is a permutation of some word in $L$.

This semantics matches the one used
in~\cite{Conway71} for commutative Kleene algebra: there, a
function
$\lang{-}_c \colon \termska(\Sigma) \rightarrow
\pow(\mathbb{N}^\Sigma)$ interprets regular expressions as subsets of
$\mathbb{N}^\Sigma$, whose elements are thought of as ``commutative
words'': these assign to each letter the number of occurrences, but
there is no order of letters.  Let
$q \colon \pow(\Sigma^*) \rightarrow \pow(\mathbb{N}^\Sigma)$,
$q(L) = \set{\lambda x. |w|_x \mid w \in L}$; this map computes the
\emph{Parikh image} of a given language $L$, that is, the set of
multisets representing occurrences of letters in words in $L$.  Then
the semantics is characterised by $\lang{-}_c = q \circ \lang{-}$.

One may observe that $\lang{-}_c = q\circ\sem C -$, since
$\cl C$ only adds words to a language which have the same number of
occurrences of each letter as some word which is already
there. Conversely, we have $\sem C -=q'\circ\lang{-}_c$, where
$q' \colon \pow(\mathbb{N}^\Sigma) \rightarrow \pow(\Sigma^*) $,
$q'(L) = \set{w\mid p\in L,~\forall x\in\Sigma, |w|_x = p(x) }$.  As a
consequence, we have $\lang{e}_c=\lang{f}_c$ iff
$\sem C e=\sem C f$.

From there, we can easily deduce from the completeness result
in~\cite[Chapter~11, Theorem~4]{Conway71}, attributed to
Pilling (see also~\cite{Brunet19}, which mentions the history of this result including the earlier proof by Redko~\cite{redko1964algebra}), that $\KA_C$ is complete.

\qedhere
\end{exa}

\begin{exa}
  \label{ex:ab0}
Let us now consider a single hypothesis: $D=\set{ab\leq 0}$ for some
letters $a$ and $b$. The $D$-closure of a language $L$ consists of those words that either belong to $L$, or contain $ab$ as a subword. As a consequence, we have $\sem D e=\sem D f$ if and only if $\lang{e}$ and $\lang{f}$ agree on all words not containing the pattern $ab$.

In this example, we can easily obtain decidability and completeness of $\KA_D$. Indeed, consider the function $r\colon\termska(\Sigma)\to\termska(\Sigma)$, $r(e)\eqdef e+\Sigma^*ab\Sigma^*$. For all $e$, we have $\KAproves[D] e=r(e)$, and $\sem D e=\lang{r(e)}$. As a consequence, we have
\begin{align*}
  &&\sem D e&=\sem D f \\&\Leftrightarrow
  &\lang{r(e)}&=\lang{r(f)} \\&\Leftrightarrow\tag{\cref{thm:ka}}
  &\KAproves r(e)&=r(f) \\&\Rightarrow
  &\KAproves[D] e&=f 
\end{align*}
The first step above establishes decidability of the closed semantics; the following ones reduce the problem of completeness for $\KA_D$ to that for $\KA$ alone (which holds). By soundness (\cref{thm:soundness}), the last line implies the first one, so that these conditions are all equivalent.
\end{exa}

This second example exploits and illustrate a simple instance of the framework we design in the sequel to prove completeness of Kleene algebra with various sets of hypotheses.

\subsection{Reductions}
\label{ssec:red}

As illustrated above, the overall strategy is to reduce completeness of $\KA_H$, for a given set of hypotheses $H$,
to completeness of Kleene algebra without hypotheses. The core idea is to provide a map $r$ from expressions to expressions,
which incorporates the hypotheses $H$ in the sense that $\lang{r(e)} = \sem H e$,
and such that $r(e)$ is provably equivalent to $e$ under the hypotheses $H$.
This idea leads to the unifying notion of \emph{reduction}, developed in~\cite{KozenM14,dkpp:fossacs19:kah,KappeB0WZ20}.

\begin{defi}[Reduction]\label{def:red}
  Assume $\Gamma\subseteq\Sigma$ and let $H$, $H'$ be sets of
  hypotheses over $\Sigma$ and $\Gamma$ respectively. We say that $H$
  \emph{reduces to} $H'$ if
  \begin{enumerate}
  \item\label{red:proves} $\KAproves[H] H'$, and
  \end{enumerate}
  there exists a map
  $r \colon \termska(\Sigma) \rightarrow \termska(\Gamma)$ such that for all
  $e\in\termska(\Sigma)$,
  \begin{enumerate}
    \setcounter{enumi}{1}
  \item\label{red:proves:r} $\KAproves[H] e = r(e)$, and
  \item\label{red:lang:r} $\sem H e\cap \Gamma^* \subseteq \sem{H'}{r(e)}$.
  \end{enumerate}
\end{defi}
We often refer to the witnessing map $r$ itself as a reduction.
Note that condition~\eqref{red:proves} vanishes when $H'=\emptyset$ is empty.
Thanks to the first two conditions, the third condition can be strengthened into an equality:
\begin{lem}
  \label{lem:red:opt}
  If $r$ is a reduction from $H$ to $H'$ as in the above definition,
  then for all expressions $e$ we have $\sem{H'}{r(e)} = \sem H e\cap \Gamma^*$.
\end{lem}
\begin{proof}
  By item~\ref{red:lang:r}, it suffices to show $\sem{H'}{r(e)}\subseteq \sem H e$. We have
  \begin{align*}
    \sem{H'}{r(e)}&\subseteq\sem{H'}{\sem{H}{r(e)}}\\
    \tag{item~\ref{red:proves:r} with \cref{thm:soundness}}
                  &=\sem{H'}{\sem{H}{e}}\\
    \tag*{(item~\ref{red:proves} with \cref{lem:clo:incl})~\qedhere}
                  &\subseteq \sem{H}{\sem{H}{e}} = \sem{H}{e} 
  \end{align*}
\end{proof}
Generalising \cref{ex:ab0}, we get the following key property of reductions:
\begin{thm}
  \label{thm:red}
  Suppose $H$ reduces to $H'$. If $\KA_{H'}$ is complete, then so is $\KA_H$.
\end{thm}
\begin{proof}
Let $r$ be the map for the reduction from $H$ to $H'$.
For all $e,f \in \termska(\Sigma)$,
\begin{align*}
  &&\sem H e&=\sem H f \\&\Rightarrow \tag{\cref{lem:red:opt}}
  &\sem{H'}{r(e)}&=\sem{H'}{r(f)} \\&\Rightarrow\tag{completeness of $\KA_{H'}$}
  &\KAproves[H'] r(e)&=r(f) \\&\Rightarrow \tag{$r$ a reduction, item~\ref{red:proves}}
  &\KAproves[H] r(e) &= r(f) \\&\Rightarrow \tag*{($r$ a reduction, item~\ref{red:proves:r})\qedhere}
  &\KAproves[H] e &= f
\end{align*}
\end{proof}
An important case is when $H'=\emptyset$: given a reduction from $H$ to $\emptyset$,
\cref{thm:red} gives completeness
of $\KA_H$, by completeness of $\KA$. Such reductions to the empty set are what we ultimately aim for.
However, in order to decompose reductions into small and elementary pieces, we need
the extra generality allowed by our definition (e.g., so that reductions can be composed---\cref{lem:red:comp}).

While we focus on completeness in this paper, note that reductions can also be used to
prove decidability:
\begin{thm}
  \label{thm:dec}
  If $\KA_{H'}$ is complete and decidable, and $H$ reduces
  to $H'$ via a computable function $r$, then $\KA_H$ is decidable.
\end{thm}
\begin{proof}
  By soundness of $\KA_H$ (\cref{thm:soundness}) and by inspecting the above proof of \cref{thm:red}, we see that $\KAproves[H] e=f$ iff $\KAproves[H'] r(e)=r(f)$.
\end{proof}

\begin{exa}
  \label{ex:katop}
  We consider KA together with a global ``top element'' $\ttop$ and the axiom $e \leq \ttop$. To make this precise in Kleene algebra with hypotheses, we assume an alphabet $\Sigma$ with $\ttop \in \Sigma$, and take the set of hypotheses $H_\ttop = \set{e \leq \ttop \mid e \in \termska(\Sigma)}$.

  We claim that $H$ reduces to $\emptyset$. The first condition is void since we target the empty set.
  For the two other conditions, define the unique homomorphism
  $r \colon \termska(\Sigma) \rightarrow \termska(\Sigma)$
  such that $r(\ttop) = \Sigma^*$ (where we view $\Sigma$ as the sum $a_1+...+a_n$ of all the letters in $\Sigma$) and $r(a) = a$ for $a \in\Sigma\setminus\set\ttop$.  

  Thanks to the hypotheses in $H$, we have $\KAproves[H] \ttop=\Sigma^*$, from which we easily deduce  the second condition ($\KAproves[H] e=r(e)$ for all expressions $e$), by induction.
  Concerning the third condition, notice that the restriction to words over the smaller alphabet vanishes since the alphabet remains the same in that example.
  Then it suffices to observe that $\clH\ttop(L)$ contains those words obtained from a word $w \in L$ by replacing occurrences of $\ttop$ in $w$ by arbitrary words in $\Sigma^*$.

  This gives a first proof that $\KA_{H_\ttop}$ is complete, which we will revisit in \cref{ex:katop:revisited}.
  Note that this implies completeness w.r.t.\ validity of equations in all language models,
  where $\ttop$ is interpreted as the largest language: indeed, the
  closed semantics $\sem{H_\ttop}-$ is generated by such a model.
  \qedhere
\end{exa}

\begin{rem}
  Reductions do not always exist.
  In \cref{ex:comm:ka}, we discussed commutative KA as an instance of Kleene algebra with hypotheses $H$.
  While $\KA_C$ is complete in that case, there is no reduction from $C$ to $\emptyset$, as
  $\cl C$ does not preserve regularity. Indeed, $\sem C {(ab)^*} = \set{w \mid |w|_a = |w|_b}$ which is
  not regular. The completeness proof in~\cite{Conway71,Brunet19} is self-contained, and does not rely
  on completeness of KA.
\end{rem}

\section{Tools for analysing closures and constructing reductions}
\label{sec:tools}

To deal with concrete examples we still need tools to construct reductions.
We provide such tools in this section.
First, we give simpler conditions to obtain reductions in the special case where the underlying function is a homomorphism (\cref{ssec:hred}).
Second, we give a general technique to construct reductions via finite automata (\cref{ssec:automata:red}).
Third, we establish a few basic reductions for commonly encountered sets of hypotheses (\cref{ssec:basic:red}).
And fourth, we provide lemmas to compose reductions (\cref{ssec:compo}).

Before doing so, we mention elementary properties of the functions $\cl H$ associated to sets $H$ of hypotheses. Those properties will be useful first to obtain the aforementioned tools, and then later to deal with concrete examples.

\subsection{On closures w.r.t. sets of hypotheses}
\label{ssec:clo:props}

Given a set $H$ of hypotheses, the function $\cl H$ (\cref{def:clo}) actually is the least closure containing the following function:
\begin{align*}
  H\colon \pow(\Sigma^*) &\to \pow(\Sigma^*)\\
  L &\mapsto \bigcup\set{ u\lang e v \mid e\leq f\in H,~u\lang f v\subseteq L}
\end{align*}
This function intuitively computes the words which are reachable from $L$ in exactly one step. In particular, when $H$ consists of inequations between words, we have $u\in H(L)$ iff $u\leftsquigarrow_H v$ for some $v\in L$. (Note that we slightly abuse notation by writing $H$ both for a set of hypotheses and for the associated function.)

Thanks to the above characterisation, many properties of $\cl H$ can be obtained by analysing the simpler function $H$. To this end, we actually develop a general theory of least closures in \cref{app:closures}. We will state the relevant lemmas in the main text, delegating the corresponding proofs to that appendix.

\medskip

Let us call a function between complete lattices \emph{linear} when it preserves all joins, and \emph{affine} when it preserves all non-empty joins (see \cref{app:closures}).
Not all sets of hypotheses yield linear (or even affine) functions: consider for instance $\set{1\leq a+b}$.
However, an important class of hypotheses do:
\begin{lem}
  \label{lem:hyp:add}
  If $H$ is a set of hypotheses whose right-hand sides are words, then the functions $H$ and $\cl H$ are linear.
\end{lem}
\begin{proof}
  When $H$ consists of hypotheses of the shape $e\leq w$ with $w$ being a word, the definition of the function $H$ simplifies to $H(L)=\bigcup\set{u\lang e v \mid e\leq w\in H,~uwv\in L}$.
  We deduce
  \begin{align*}
    H\paren{\bigcup_{i\in I}L_i}
    &= \bigcup\{u\lang e v \mid e\leq w\in H,~uwv\in \bigcup_{i\in I}L_i\}\\
    &= \bigcup\set{u\lang e v \mid e\leq w\in H,~i\in I,~uwv\in L_i}\\
    &= \bigcup_{i\in I}\bigcup\set{u\lang e v \mid e\leq w\in H,~uwv\in L_i}
     = \bigcup_{i\in I}H(L_i)
  \end{align*}
  Thus $H$ is linear, and we deduce that so is $\cl H$ by \cref{lem:clo:add}.
\end{proof}

\subsection{Homomorphic reductions}
\label{ssec:hred}

The following result from~\cite{KappeB0WZ20} (cf. Remark~\ref{rem:red:ckah})
simplifies the conditions of a reduction, assuming the underlying map $r$
is a homomorphism.
\begin{prop}
  \label{prop:hred}
  Assume $\Gamma\subseteq\Sigma$ and let $H$, $H'$ be sets of
  hypotheses over $\Sigma$ and $\Gamma$ respectively.
  If
  \begin{enumerate}
  \item \label{hred:proves} $\KAproves[H] H'$,
  \end{enumerate}
  and there exists a homomorphism
  $r \colon \termska(\Sigma) \rightarrow \termska(\Gamma)$ such that:
  \begin{enumerate}
    \setcounter{enumi}{1}
  \item \label{hred:proves:r}
    for all $a \in \Sigma$, we have $\KAproves[H] a = r(a)$,
  \item \label{hred:idem}
    for all $a \in \Gamma$, we have $\KAproves a \leq r(a)$, and
  \item \label{hred:hyps}
    for all $e\leq f \in H$, we have $\KAproves[H'] r(e) \leq r(f)$,
  \end{enumerate}
  then $H$ \emph{reduces to} $H'$, with $r$ being the witness of the reduction.
\end{prop}
\noindent
Conditions \eqref{hred:idem} and \eqref{hred:hyps} are stated as derivability conditions, but this is only a matter of convenience: all we need are the corresponding inclusions of languages---closed by $H'$ for~\eqref{hred:hyps}.

We prove this proposition below; first we illustrate its usage by revisiting \cref{ex:katop}.
\begin{exa}
  \label{ex:katop:revisited}
  The reduction we defined to deal with \cref{ex:katop} is a homomorphism.
  Accordingly, we can use \cref{prop:hred} to prove more easily that it is indeed a reduction from $H_\ttop$ to the empty set, without ever having to understand the closure $\clH\ttop$:
  \begin{itemize}
  \item the first condition is void since $H'=\emptyset$;
  \item the second and third condition hold trivially for all letters different from $\ttop$, on which $r$ is the identity. For the letter $\ttop$, we have
    \begin{align*}
      \tag{because $\ttop\in\Sigma$}
      &\KAproves \ttop \leq \Sigma^*=r(\ttop)\\
      \tag{by a single application of $H_\ttop$}
      &\KAproves[H_\ttop] r(\ttop) \leq \ttop
    \end{align*}
  \item the fourth condition holds by completeness of KA:
    $\KAproves r(e)\leq \Sigma^*=r(\ttop)$ for all expressions $e$ since any language is contained in $\Sigma^*$, by definition.
    (Note that using completeness of KA is an overkill here.)\qed
  \end{itemize}
\end{exa}

We now turn to proving \cref{prop:hred}. We need the following auxiliary definition:
for a map $r \colon\termska(\Sigma)\rightarrow\termska(\Gamma)$, define
\begin{align*}
  \dot r\colon  \pow(\Sigma^*)&\to\pow(\Gamma^*)\\
  L&\mapsto \bigcup_{w\in L} \lang{r(w)}\,.
\end{align*}
(Here words $w$ are seen as regular expressions when fed to $r$.)
The function $\dot r$ is linear. When $r$ is a homomorphism, we have:
\begin{fact}
  \label{fact:dotr}
  If $r$ is a homomorphism then $\dot r$ is a homomorphism and for all expressions $e\in\termska(\Sigma)$ we have $\dot r(\lang{e})=\lang{r(e)}$.
\end{fact}

\begin{proof}[Proof of \cref{prop:hred}]
  First, using that $r$ is a homomorphism, we extend by induction the assumptions (2) and (3) into:
  \begin{enumerate}[align=left]
  \item[(2')] for all $e\in\termska(\Sigma)$, $\KAproves[H] e = r(e)$
  \item[(3')] for all $e\in\termska(\Gamma)$, $\KAproves e \leq r(e)$
  \end{enumerate}
  From (3') we deduce
  \begin{enumerate}[align=left]
  \item[(3'')] for all $L\subseteq \Gamma^*$, $L\subseteq \dot r(L)$
  \end{enumerate}
  (Note the restriction to words in $\Gamma^*$.) Indeed, by soundness of KA and (3'), we have
  $\lang e\subseteq \lang{r(e)}$ for all $e\in\termska(\Gamma)$, and in particular
  $w\in \lang{r(w)}$ for all words $w\in\Gamma^*$.

  \medskip

  At this point we only have the third condition of a reduction left to show, i.e., that for all $e\in\termska(\Sigma)$, $\sem H e\cap \Gamma^* \subseteq \sem{H'}{r(e)}$.
  We proceed as follows:
  \begin{align*}
    \tag{by (3'')}
    \sem H e\cap \Gamma^*
    &\subseteq \dot r(\sem H e\cap \Gamma^*)\\
    \tag{$\dot r$ monotone}
    &\subseteq \dot r(\sem H e)\\
    \tag{$\dagger$}
    &\subseteq \cl{H'}(\dot r(\lang e))\\
    \tag{by \cref{fact:dotr}}
    &= \sem{H'}{r(e)}\,.
  \end{align*}
  It remains to show $(\dagger)~\dot r\circ \cl H \subseteq \cl{H'}\circ\dot r$.
  Since $\dot r$ is linear, \cref{prop:clo:iter:r} applies and it suffices to prove
  $(\ddagger)~\dot r\circ H \subseteq \cl{H'}\circ\dot r$.
  Let $L\subseteq \Sigma^*$. Again since $\dot r$ is linear, we have
  \begin{align*}
  \dot r(H(L))=\bigcup_{
    \begin{array}{c}
      e\leq f\in H\\u\lang f v\subseteq L
    \end{array}
    }\dot r(u\lang e v)\,.
  \end{align*}
  Accordingly, let $e\leq f\in H$ and let $u,v$ be words such that
  $u\lang f v\subseteq L$.
  By assumption~\eqref{hred:hyps}, we have a derivation $\KAproves[H'] r(e) \leq r(f)$, whence, by monotonicity, another derivation
  \begin{equation}\label{eq:rurerv}
  \KAproves[H'] r(u)r(e)r(v) \leq r(u)r(f)r(v) \, .
  \end{equation}
  We deduce:
  \begin{align*}
    \dot r (u\lang e v)
    &=\dot r (\lang{u e v})\\
    \tag{by \cref{fact:dotr}}
    &=\lang{r(uev)}\\
    \tag{$r$ homomorphism}
    &=\lang{r(u)r(e)r(v))}\\
    \tag{by \cref{thm:soundness} applied to \eqref{eq:rurerv}}
    &\subseteq\sem{H'}{r(u)r(f)r(v))}\\
    \tag{$r$ homomorphism}
    &=\sem{H'}{r(ufv)}\\
    \tag{by \cref{fact:dotr}}
    &=\cl{H'}(\dot r(\lang{u f v}))\\
    \tag{by monotonicity}
    &=\cl{H'}(\dot r(u\lang f v))\subseteq\cl{H'}(\dot r(L))\,.
  \end{align*}
  Therefore, we get $\dot r(H(L))\subseteq \cl{H'}(\dot r(L))$: we have proved $(\ddagger)$, which concludes the proof.
\end{proof}

\begin{rem}
  \label{rem:red:ckah}
  The idea to use two sets of hypotheses in Definition~\ref{def:red}
  is from~\cite{KappeB0WZ20}, where reductions are defined slightly differently: the alphabet is fixed (that is, $\Sigma = \Gamma$),
  and the last condition is instead defined as $\sem H{e}=\sem H{f}\Rightarrow \sem{H'}{r(e)}=\sem{H'}{r(f)}$. An extra notion of \emph{strong} reduction is then introduced,
  which coincides with our definition if $\Sigma=\Gamma$. By allowing a change of alphabet, we do not need
  to distinguish reductions and strong reductions.
  \cref{prop:hred} is in~\cite[Lemma 4.23]{KappeB0WZ20}, adapted here
  to the case with two alphabets (this is taken care of in \emph{loc.\ cit.} by assuming
  $\cl{H'}$ preserves languages over $\Gamma$).
  \qedhere
\end{rem}

\subsection{Automata-based reductions}
\label{ssec:automata:red}

Another recipe to construct reductions consists in using finite automata (see, e.g.,\ \cite{EB95,KozenM14,dkpp:fossacs19:kah}).

Suppose, for instance, that we want to build a reduction from $H\eqdef\set{a\leq bc}$ to the empty set.
Given an automaton for the language $L$, an automaton for the language $\cl{H}(L)$ is easily obtained by adding ``shortcuts'': we add an $a$-transition from state $i$ to state $j$ whenever there is a $bc$-labelled path from $i$ to $j$.
Using Kleene's theorem---the languages denoted by regular expressions are the languages accepted by finite automata---we define a function $r$ on regular expressions accordingly: given an expression $e$, compute an automaton for its language, add shortcuts as above, and define $r(e)$ to be any expression for the resulting automaton.
The first requirement of a reduction holds trivially since we target the empty set; the third one holds by construction; and since $\lang e\subseteq \lang{r(e)}$ (also by construction), we have $\KAproves e \leq r(e)$ by completeness of KA (\cref{thm:ka}). Therefore, it only remains to prove $\KAproves[H] r(e) \leq e$. This last step is the most delicate one; it can be obtained using the following technique, inspired from~\cite[Lemma~32]{ddp:lpar18:lefthanded}.

\begin{defi}
  \label{def:compat}
  Let $H$ be a set of hypotheses.
  Given a (non-deterministic, finite) automaton $\A$, a labelling function $\f$ from states to regular expressions is \emph{compatible (with $\A$)} if $\KAproves[H] 1\leq \f(i)$ whenever $i$ is an accepting state, and $\KAproves[H] a\cdot \f(j)\leq \f(i)$ whenever there is an $a$-transition from $i$ to $j$ (for $a$ a letter, or $1$ in case of an epsilon transition).
\end{defi}
Any function associating to each state a regular expression for its language in the automaton is compatible (by completeness of KA); we call such labellings \emph{canonical}.
\begin{prop}
  \label{prop:compat}
  Let $H$ be a set of hypotheses, let $\A$ be an automaton, let $\br$ be a canonical labelling and let $\f$ be a compatible labelling. For all states $i$, we have $\KAproves[H] \br(i)\leq \f(i)$.
\end{prop}
\begin{proof}
  Thanks to completeness of KA, this is a direct consequence of~\cite[Lemma~32]{ddp:lpar18:lefthanded} (this lemma only gives the statement for a specific canonical labelling).
  
  A more direct proof can be obtained as follows, using the matricial representation of automata~\cite{Kozen91}.
  Write the automaton $\A$ matricially as $\tuple{u,M,v}$, where $u$ is the row vector of initial states, $M$ is the square matrix of transitions, and $v$ is the column vector of accepting states. Accordingly, see the labellings $\br$ and $\f$ as column vectors. 
  Recall that an expression for the language of $\A$ can be obtained by computing the product $u\cdot M^*\cdot v$ in the Kleene algebra of matrices.
  We have $\KAproves \br=M^*\cdot v$ by completeness of KA, and the compatibility of $\f$ can be restated as $\KAproves[H] v\leq \f$ and $\KAproves[H] M\cdot \f\leq \f$.
  By star induction on the left, we deduce $\KAproves[H] \br=M^*\cdot v\leq \f$.  
\end{proof}

This proposition makes it possible to finish our example with $H\eqdef\set{a\leq bc}$.
Let $e$ be a regular expression, let $\A$ be an automaton for its language, and let $\A'$ be the automaton obtained after adding shortcuts. Let $\e$ and $\br$ be canonical labellings of $\A$ and $\A'$, respectively. We have that $\e$ is compatible with $\A$, but also with $\A'$. Indeed, each new transition $i\xrightarrow{a}j$ in $\A'$ comes from two transitions $i\xrightarrow{b}k\xrightarrow{c}j$ in $\A$, and we have
\begin{align*}
  \tag{hypothesis in $H$}
  \KAproves[H] a\cdot \e(j)&\leq b\cdot c\cdot \e(j)\\ 
  \tag{$\e$ is compatible with $\A$}
  &\leq b\cdot \e(k)\\ 
  \tag{$\e$ is compatible with $\A$}
  &\leq \e(i)
\end{align*}
Therefore, by \cref{prop:compat}, $\KAproves[H] \br\leq \e$. It follows that
$\KAproves[H] r(e)=u\cdot \br\leq u\cdot \e = e$ (using the row vector $u$ of initial states of $\A$).

\subsection{Basic reductions}
\label{ssec:basic:red}

The following result collects several sets of hypotheses for which we have reductions to the empty set.
These mostly come from the literature; we establish them using the tools from the previous two sections. They form basic building blocks used in the more complex reductions that we present in the examples below.

\begin{lem}
  \label{lem:reds}
  The following sets of hypotheses over $\Sigma$ reduce to the empty set.
  \begin{enumerate}[(i)]
  \item\label{it:red:aw} $\set{u_i \leq w_i \mid i \in I} $
    with $u_i,w_i\in \Sigma^*$ and $|u_i|\leq 1$ for all $i\in I$
  \item\label{it:red:1S}
    $\set{1 = \sum_{a \in S_i}a\mid i\in I}$ with each $S_i \subseteq \Sigma$ finite
  \item\label{it:red:e0} $\set{e \leq 0}$ for $e \in \termska(\Sigma)$
  \item\label{it:red:ea} $\set{ea \leq a}$ and $\set{ae \leq a}$ for $a \in \Sigma$, $e \in \termska(\Sigma \setminus \set a)$
  \item\label{it:red:aa} $\set{aa \leq a}$ for $a \in \Sigma$
  \end{enumerate}
\end{lem}
\begin{proof}\hfill
  \begin{enumerate}[(i)]
  \item This is~\cite[Theorem~2]{KozenM14}. (The result mentions equations, but in the proof
    only the relevant inequations are used.) 
    Intuitively, the example from Section~\ref{ssec:automata:red} can easily be generalised: starting from an automaton $\A$, add $u_i$-labelled shortcuts for all $w_i$-labelled paths in $\A$ to obtain an automaton $\A'$. There is a catch, however: if some of the $u_i$ are letters appearing in some of the $w_i$, then $\A'$ is not necessarily closed: we may have to add shortcuts to $\A'$ and get a third automaton $\A''$, and so on. Fortunately, this process terminates since we never add new states to the automaton.
    
  \item This is \cite[Theorem~4]{dkpp:fossacs19:kah}.
    In the case of a single hypothesis $H=\set{1=a+b}$, the proof given there can be rephrased as follows.
    Split $H$ into two sets $H_1=\set{a\leq 1,~b\leq 1}$ and $H_2=\set{1\leq a+b}$.
    One can easily show that $\clH{}=\clH2\circ\clH1$ (in fact this is an immediate consequence of \cref{cor:nooverlap} and \cref{lem:dl:2,lem:gsos:2} below).
    By the previous item \ref{it:red:aw}, $H_1$ reduces to the empty set, say via $s$. 
    Let us try to construct a reduction $r$ from $H_2$ to the empty set.
    We use for that the automata-based technique from Section~\ref{ssec:automata:red}. A difficulty is that we need more room in the starting automaton to add shortcuts: since the right hand-side of the inequation in $H_2$ is a sum, we have to detect conjunctions.

    Given an expression $e$ and a deterministic finite automaton $\A_0$ for its language $L$, we first extend $\A_0$ into a (deterministic, finite) automaton $\A$ whose states recognise all intersections of residuals of $L$: if $\A_0$ has state-space $X$, initial state $i\in X$, transition function $\delta_0$, and accepting states $F_0\subseteq X$, then $\A$ has state-space $\pow(X)$, initial state $\set i$, transition function $\delta(a,S)\eqdef\set{\delta_0(a,x)\mid x\in S}$, and accepting states $\pow(F_0)\subseteq \pow(X)$.
    By construction, the language of a state $S$ in $\A$ is the intersection of the languages of the elements of $S$ in $\A_0$, and $\A$ contains $\A_0$ via the singleton function.

    Then we add shortcuts to $\A$ to obtain an automaton $\A'$: from every state $S$, we add an epsilon transition to the state
    \begin{align*}
      S'\eqdef \delta(a,S)\cup\delta(b,S)
    \end{align*}
    (Recall that here, the union semantically yields a conjunction.) 
    We let $r(e)$ be a regular expression for the language of $\A'$.

    Let us prove the second condition of a reduction for $r$.
    We have $\KAproves e \leq r(e)$ by completeness and construction, since the language of $\A$ is contained in that of $\A'$. For the other inequation, we use \cref{prop:compat}.
    Let $\e$ and $\br$ be canonical labellings of $\A$ and $\A'$, respectively. Since we have $\KAproves e=\e(\set i)$ and $\KAproves r(e)=\br(\set i)$, it suffices to prove that
    $\e$, which is compatible with $\A$, is also compatible with $\A'$. 
    Indeed, for each added epsilon transition from $S$ to $S'$, we have:
    \begin{align*}
      \tag{hypothesis in $H$}
      \KAproves[H] 1\cdot \e(S') &\leq (a+b)\cdot \e(S')\\
      \tag{distributivity}
      &= a\cdot \e(S') + b\cdot \e(S')\\
      \tag{KA completeness with $\lang{S'}=\lang{\delta(a,S)}\cap \lang{\delta(b,S)}$}
      &\leq a\cdot \e(\delta(a,S)) + b\cdot \e(\delta(b,S))\\
      \tag{compatibility of $\e$, along $a$ and $b$ transitions}
      &\leq \e(S)
    \end{align*}
    Now consider the third condition for a reduction and define the following function on languages:
    \begin{align*}
      P_2\colon L\mapsto\set{u_0u_1\dots u_n\mid u_0\set{a,b}u_1\dots\set{a,b}u_n\subseteq L}
    \end{align*}
    Intuitively, $P_2(L)$ is the ``one-parallel step'' closure of $L$ w.r.t. $H_2$, and the expression $r(e)$ we constructed via the automaton $\A'$ is such that:
    \begin{align*}
      \lang{r(e)}=P_2\lang e
    \end{align*}
    Unfortunately, in general, we only have $P_2(L)\subseteq \clH2(L)$.
    For instance, if $L=\set{aa,ab,bb}$, then the empty word does not belong to $P_2(L)$, although it belongs to $\clH2(L)$. Therefore, $r$ is not a reduction from $H_2$ to the empty set.
    Still, if $L$ is closed w.r.t. $H_1$ (i.e., $H_1(L)\subseteq L$), then $\clH2(L)=P_2(L)$.
    This makes it possible to conclude: $r\circ s$ is a reduction from $H$ to the empty set: for all expressions $e$, we have
    \begin{gather*}
      \KAproves[H] r(s(e))=s(e)=e,\text{ and }\\
      \sem H e=\clH2(\sem{H_1}e)=P_2(\sem{H_1}e)=P_2\lang{s(e)}=\lang{r(s(e))}\,.
    \end{gather*}
    
  \item This is basically due to~\cite{cohen94:ka:hypotheses}, but since it is phrased differently
    there we include a proof.

    Define $r \colon \termska(\Sigma) \rightarrow \termska(\Sigma)$
    by $r(f)=f+\Sigma^*\cdot e \cdot \Sigma^*$ for $f\in\termska(\Sigma)$,
    where $\Sigma^*$ is seen as the expression $(\sum_{a \in \Sigma} a)^*$.
    We claim $r$ witnesses a reduction.
    First, we have $\KAproves[H] f\leq r(f)$ trivially, and $\KAproves[H] r(f)\leq f$ follows
    from $e \leq 0 \in H$.

    Second, we have to prove that $\cl H(\lang{f})\subseteq \lang{r(f)}$.
    In fact, since the right-hand side of the considered hypothesis is $0$, the one-step function associated to $H$ is constant, and we have that for all languages $L$,
    \begin{align*}
      \cl H(L)
      = L \cup H(L)
      = L \cup \set{u \lang{e} v \mid u,v \in \Sigma^*}
      = L \cup \Sigma^*\cdot \lang{e} \cdot \Sigma^*\,.
    \end{align*}
    In particular, for $L=\lang f$, we get $\sem H f=\lang{f + \Sigma^* e \Sigma^*}=\lang{r(f)}$.
  \item Hypotheses of a similar form are studied in the setting of Kleene algebra with tests in~\cite{hardink02:kat:hypotheses}.
    Consider the case $\set{ea \leq a}$.
    Define $r$ as the unique homomorphism satisfying
$r(a) = e^* a$, and $r(b) = b$ for all $b \in \Sigma\setminus{\{a\}}$.
    We use \cref{prop:hred} to show that this witnesses a reduction.

    The first two conditions are trivial for letters $b \in \Sigma$ with $a \neq b$, since
    $r$ is the identity on those letters.
    For $a$, we have $\KAproves a \leq e^* a = r(a)$ easily from the KA axioms.
    Further, we get $\KAproves[H] r(a)=e^*a \leq a$ by $ea \leq a$ and the (left) induction
    axiom for Kleene star in $\KA$.
    Finally, we have to prove that $\KAproves r(ea) \leq r(a)$.
    We have $r(a)=e^*a$, and
    since $e$ does not contain $a$, we have $r(e) = e$. Therefore, $r(ea) = r(e)r(a) = ee^* a$,
    and the required inequality follows since $\KAproves e e^* \leq e^*$.
  \item This is a variant of the previous item, but not an instance due to the condition that $a$ cannot appear in $e$.
    The same homomorphism $r$ (satisfying $r(a) = a^+$) nevertheless gives a reduction: the first three conditions of \cref{prop:hred} are satisfied like in the previous item; for the fourth one, we have
    $\KAproves r(aa)=a^+a^+\leq a^+=r(a)$.
    \qedhere
  \end{enumerate}
\end{proof}

Note that \cref{it:red:e0} above covers finite sets of hypotheses of the form
$\set{e_i \leq 0}_{i \in I}$, as these can be encoded as the single hypothesis $\sum_{i \in I} e_i \leq 0$.

\begin{rem}
  \label{rem:wrong:red}
  We had wrongly announced in~\cite[Lemma~3.6(ii)]{prw:ramics21:mkah} a reduction to the empty set for $\set{1 \leq \sum_{a \in S_i}a\mid i\in I}$ (inequations rather than equations), attributing a proof to \cite[Proposition~6]{dkpp:fossacs19:kah}. We had forgotten the side condition in this proposition, which precisely prevents the function $r$ given in the proof sketch of item~\ref{it:red:1S} above to be a reduction. Whether there exists a reduction to the empty set for such sets of hypotheses remains open.
\end{rem}

\medskip

We conclude this subsection with a lemma allowing one to construct reductions for hypotheses of the shape $f\leq 1$ for $f$ a regular expression.
This is not always possible: consider for instance the hypothesis $P\eqdef\set{ab\leq 1}$ for two distinct letters $a,b$; then $\sem P 1$ consists of all well-parenthesised words (where $a$ means opening a parenthesis, and $b$ means closing one), which is not a regular language.
In its special case where $H$ is the empty set, \cref{lem:basic:e1} below shows a solution when we have a regular expression $f'$ above $f$, whose language is closed under $f\leq 1$, and such that $\proves[f\leq 1]f'\leq 1$. These requirements are often too strong (as in the above example), which is why we allow for a generic
target $H$. We use this lemma only for our most advanced example (NetKAT, in \cref{sec:netkat}).

\begin{lem}
  \label{lem:basic:e1}
  Let $f$ be a regular expression and let $H$ be a set of hypotheses.
  Set $H_f\eqdef \set{f\leq 1}$.
  If
  \begin{enumerate}
  \item $\cl{(H\cup H_f)}=\cl{H}\circ\clH f$
  \end{enumerate}
  and there is a regular expression $f'$ such that:
  \begin{enumerate}
    \setcounter{enumi}{1}
  \item $\proves[H,f\leq 1] f' \leq 1$
  \item $\lang f \cup H_f\lang{f'}\subseteq \lang{f'}$
  \end{enumerate}
  then $H\cup H_f$ reduces to $H$.
\end{lem}
\begin{proof}
  Let $H_1\eqdef H\cup H_f$. We trivially have $\proves[H_1]H$.
  We use the automata technique from \cref{ssec:automata:red} to build the reduction: 
  given a regular expression $e$, compute an automaton for its language, add $f'$-labelled self-loops to all states, and read back a regular expression $r(e)$.

  By assumption~2, and using \cref{prop:compat} like in \cref{ssec:automata:red}, we get $\proves[H_1] r(e)\leq e$. It follows that $\proves[H_1] r(e)=e$ by KA completeness, since $\lang e\subseteq\lang{r(e)}$.
  Therefore, it only remains to show the third condition of a reduction.
  By construction, we have
  \begin{align*}
    \lang{r(e)}=\set{u_0v_0\cdots v_{n-1}u_n\mid u_0\cdots u_n\in \lang e\text{ and }\forall i, v_i\in \lang{f'}}
  \end{align*}
  We prove that this language is $H_f$-closed $(H_f\lang{r(e)}\subseteq \lang{r(e)})$. Indeed, if we insert a word of $\lang{f}$ into a word of $\lang{r(e)}$, then it has to be either into one of the $u_i$ in the above formula, or into one of the $v_i$. In both cases, we deduce by assumption~3 that the resulting word belongs to $\lang{r(e)}$.

  The third condition of a reduction is $\semH1 e \subseteq \sem H{r(e)}$.
  By general closure properties, it suffices to show that $\sem H{r(e)}$ is $H_f$-closed $(H_f(\sem H{r(e)})\subseteq \sem H{r(e)})$. In turn, since $H_f\circ \cl{H}\subseteq\cl{H}\circ\cl{H_f}$ by assumption~1, this follows from the fact that $\lang{r(e)}$ is $H_f$-closed, which we have just proved.
\end{proof}

\subsection{Composing reductions}
\label{ssec:compo}

The previous subsection gives reductions for uniform sets of equations. However,
in the examples we often start with a \emph{collection} of hypotheses of different shapes,
which we wish to reduce to the empty set.
Therefore, we now discuss a few techniques for combining reductions.

First, reductions can be composed:
\begin{lem}
  \label{lem:red:comp}
  Let $H_1, H_2$ and $H_3$ be sets of hypotheses.
  If $H_1$ reduces to $H_2$ and $H_2$ reduces to $H_3$ then $H_1$ reduces to $H_3$.
\end{lem}
\begin{proof}
  Let $r_1 \colon \termska(\Sigma_1) \rightarrow \termska(\Sigma_2)$ and $r_2 \colon \termska(\Sigma_2) \rightarrow \termska(\Sigma_3)$ be reductions from $H_1$ to $H_2$ and from $H_2$ to $H_3$, respectively.
  We show that $r_2 \circ r_1$ is a reduction from $H_1$ to $H_3$.
  \begin{enumerate}
  \item We have $\proves[H_1]H_2$ and $\proves[H_2]H_3$, and therefore $\proves[H_1]H_3$.
  \item For $e\in\termska(\Sigma_1)$, we have $\proves[H_1] e = r_1(e)$.
    Since $r_2(e)\in\termska(\Sigma_2)$, we also have  $\proves[H_2] r_1(e) = r_2(r_1(e))$.
    Since $\proves[H_1] H_2$, we deduce $\proves[H_1] e = r_1(e) = r_2(r_1(e))$.
  \item For $e\in\termska(\Sigma_1)$, we have:
    \begin{align*}
    \sem{H_1}e \cap \Sigma_3^*
    = (\sem{H_1}e \cap \Sigma_2^*) \cap \Sigma_3^*
    \subseteq \sem{H_2}{r_1(e)} \cap \Sigma_3^*
    \subseteq \sem{H_3}{r_2(r_1(e))}
    \end{align*}
  \end{enumerate}
  using that $\Sigma_3 \subseteq \Sigma_2$ for the first equality, and then that $r_1$ and $r_2$ are reductions.
\end{proof}

Second, the identity function gives non-trivial reductions under certain assumptions:
\begin{lem}
  \label{lem:red:id}
  Let $H, H'$ be sets of hypotheses over a common alphabet.
  If $\proves[H] H'$ and $H \subseteq \cl{H'}$, then $H$ reduces to $H'$.
\end{lem}
\begin{proof}
  We have $\cl{H'}\subseteq \cl H$ by \cref{lem:clo:incl}, whence $\cl H=\cl{H'}$ via the second assumption.
  Therefore, the identity map on terms fulfils the requirements of \cref{def:red}.
\end{proof}
In particular, equiderivable sets of hypotheses reduce to each other:
\begin{cor}
  \label{cor:red:id}
  Let $H, H'$ be sets of hypotheses over a common alphabet.
  If $\proves[H] H'$ and $\proves[H'] H$ then $H$ and $H'$ reduce to each other.
\end{cor}
\begin{proof}
  Thanks to \cref{lem:clo:incl}, \cref{lem:red:id} applies in both directions.
\end{proof}

Third, the following lemma makes it possible to obtain reductions for unions of hypotheses by considering those hypotheses separately.

\begin{lem}
  \label{lem:cup:2}
  Let $H_1,H_2,H'_1,H'_2$ be sets of hypotheses over a common alphabet $\Sigma$.\\
  If\quad
  $\begin{cases}
   \text{$H_1$ reduces to $H'_1$ and $H_2$ reduces to $H'_2$}\\
   \text{$\cl{(H_1\cup H_2)} \subseteq\clH2 \circ \clH1$} \\
   \text{$\proves[H'_2] H'_1$}
  \end{cases}$
  then $H_1\cup H_2$ reduces to $H'_1\cup H'_2$\,.
\end{lem}
\begin{proof}
  Let $H=H_1\cup H_2$ and $H'=H'_1\cup H'_2$.
  For $i=1,2$, let $r_i$ be the reduction from $H_i$ to $H'_i$.
  We show that $r_2 \circ r_1$ is a reduction from $H$ to $H'$.
  \begin{enumerate}
  \item Since $\proves[H_i] H'_i$ for $i=1,2$, we have $\proves[H] H'$.
  \item For all expressions $e$, we have
    $\proves[H]\proves[H_i] e = r_i(e)$ for $i=1,2$.
    We deduce that for all expressions $e$, $\proves[H] e = r_1(e)=r_2(r_1(e))$.
  \item For the third requirement, note that the restriction to the codomain alphabet is void since there is a single alphabet (so that we simply have $\sem{H_i}e \subseteq \sem{H'_i}{r_i(e)}$ for all $i,e$).
    
    Also observe that by \cref{lem:clo:incl} and the assumptions $\proves[H_2]\proves[H'_2] H'_1$, we have
    $\cl{H'_1}\subseteq \cl{H'_2} \subseteq \cl{H_2}$. From the first inclusion, we deduce $(\dagger)~\cl{H'_2}=\cl{H'}$; from the composed one, we obtain $(\ddagger)~\cl{H_2}\cl{H'_1}=\cl{H_2}$.
    We deduce
    \begin{align*}
      \tag{by assumption}
      \sem H e &\subseteq \clH2(\semH1 e)\\
      \tag{$r_1$ a reduction}
               &\subseteq \clH2(\sem{H'_1}{r_1(e)})\\
      \tag{$\ddagger$}
               &= \semH2{r_1(e)}\\
      \tag{$r_2$ a reduction}
               &\subseteq \sem{H'_2}{r_2(r_1(e))}\\
      \tag{$\dagger$}
               &= \sem{H'}{r_2(r_1(e))}
    \end{align*}
    which concludes the proof.
    \qedhere
  \end{enumerate}
\end{proof}

In most cases, we use the above lemma with $H'_1=H'_2$, and the third condition is trivially satisfied.
We now state a few results that make it possible to fulfil its second requirement.
Those results hold in a much more general setting, they are proved in \cref{app:closures}.

\begin{lem}
  \label{lem:dl:2}
  Let $H_1,H_2$ be two sets of hypotheses over a common alphabet.\\
  We have $\cl{(H_1\cup H_2)}=\clH2\circ\clH1$ if and only if $\clH1\circ \clH2\subseteq \clH2 \circ \clH1$.
\end{lem}

\begin{lem}
  \label{lem:gsos:pre}
  Let $H_1,H_2$ be two sets of hypotheses over a common alphabet.\\
  If $H_1\circ \clH2\subseteq \clH2 \circ \clH1$, then $\clH1\circ\clH2\subseteq\clH2\circ\clH1$.
\end{lem}

\noindent
Given a function $f$, we write $f^=$ for the function $f\cup\id$ (i.e., $f^=(L)=f(L)\cup L$).
\begin{lem}
  \label{lem:gsos:2}
  Let $H_1,H_2$ be two sets of hypotheses over a common alphabet.\\
  If $H_1$ is affine and either
  \begin{enumerate}
  \item $H_1\circ H_2\subseteq \clH2 \circ H_1^=$, or
  \item $H_1\circ H_2\subseteq H_2^= \circ \clH1$,
  \end{enumerate}
  then $\clH1\circ\clH2\subseteq\clH2\circ\clH1$.
\end{lem}

We conclude this section by returning to hypotheses of the form $e \leq 0$. We can indeed strengthen \cref{lem:reds}\ref{it:red:e0} to \cref{lem:red:0} below, which gives a general treatment of hypotheses of the from $e \leq 0$: we can always get rid of finite sets of hypotheses of this form.
A similar result, in terms of Horn formulas and in the context of KAT, is shown in~\cite{Hardin05}.
\begin{lem}
  \label{lem:dl:0}
  Let $H_0 = \set{e \leq 0}$ for some term $e$. The set $H_0$ is affine, and for all sets $H$ of hypotheses we have $H_0 \circ H \subseteq H_0$ and $\cl{(H_0\cup H)} = \cl H \circ \clH0$.
\end{lem}
\begin{proof}
  In fact the function $H_0$ is constant: $H_0(L)=\Sigma^*\lang e\Sigma^*$.
  Whence affineness and the first containment.
  The remaining equality follows from \cref{lem:dl:2,lem:gsos:2}.
\end{proof}

\begin{lem}
  \label{lem:red:0}
  For any set of hypotheses $H$ and any term $e$, $H \cup \set{e \leq 0}$ reduces to $H$.
\end{lem}
\begin{proof}
  This is a direct application of \cref{lem:cup:2}:
  we combine the reduction from $\set{e \leq 0}$ to the empty set provided by  \cref{lem:reds}\ref{it:red:e0} with the identity reduction from $H$ to itself.
  The second requirement holds by \cref{lem:dl:0}; the third requirement is void, since $H'_1$ is chosen to be empty.
\end{proof}

\section{Kleene algebra with tests}
\label{sec:kat}

In this section we apply the machinery from the previous sections to obtain a modular completeness proof for Kleene algebra with tests~\cite{kozens96:kat:completeness:decidability}.

A \emph{Kleene algebra with tests (KAT)} is a
Kleene algebra $K$ containing a Boolean algebra $L$ such that the meet of $L$
coincides with the product of $K$, the join of $L$ coincides with the
sum of $K$, the top element of $L$ is the multiplicative identity of $K$, and the bottom elements of $K$ and $L$ coincide.

Syntactically, we fix two finite sets $\Sigma$ and $\Omega$ of primitive
actions and primitive tests.
We denote the set of Boolean expressions over alphabet $\Omega$ by $\termsba$:
\[
  \phi,\psi ::= \phi \vee \psi \pipe \phi \wedge \psi \pipe \lnot\phi \pipe \bot \pipe \top \pipe o \in \Omega
\]
We write $\BAproves\phi=\phi'$ when this equation is derivable from Boolean algebra axioms~\cite{birkhoff-bartee-1970,DaveyPriestley90}, and similarly for inequations.

We let $\alpha,\beta$ range over \emph{atoms}: elements of the finite set $\At\eqdef2^\Omega$. Those may be seen as valuations for Boolean expressions, or as complete conjunctions of literals: $\alpha$ is implicitly seen as the Boolean formula $\bigwedge_{\alpha(o)=1}o\wedge\bigwedge_{\alpha(o)=0}\lnot o$. They form the atoms of the Boolean algebra generated by $\Omega$.
We write $\alpha\models\phi$ when $\phi$ holds under the valuation $\alpha$. A key property of Boolean algebras is that for all atoms $\alpha$ and formulas $\phi$, we have
\begin{align*}
  \alpha\models\phi ~\Leftrightarrow~ \BAproves\alpha\leq\phi \qquad\text{and}\qquad
  \BAproves\phi = \bigvee_{\alpha\models \phi}\alpha
\end{align*}

The KAT terms over alphabets $\Sigma$ and $\Omega$ are the regular expressions over the alphabet $\Sigma+\termsba$: $\termskat\eqdef\termska(\Sigma+\termsba)$.
We write $\KATproves e=f$ when this equation is derivable from the axioms of KAT, and similarly for inequations.

The standard interpretation of KAT associates to each term a language of guarded strings.
A \emph{guarded string} is a sequence of the form $\alpha_0 a_0 \alpha_1 a_1 \ldots a_{n-1} \alpha_n$
with $a_i \in \Sigma$ for all $i<n$, and $\alpha_i \in \At$ for all $i\leq n$.
We write $\GS$ for the set $\At\times(\Sigma\times \At)^*$ of such guarded strings.
Now, the interpretation $\G \colon \termska(\Sigma+\termsba) \rightarrow 2^\GS$
is defined as the homomorphic extension of
the assignment $\G(a)=\set{\alpha a \beta \mid \alpha,\beta\in\At}$ for $a\in\Sigma$ and $\G(\phi)=\set{\alpha\mid \alpha\models\phi}$ for $\phi\in\termsba$, where for sequential composition of guarded strings the \emph{coalesced product} is used. The coalesced product of
guarded strings $u\alpha $ and $\beta v$ is defined as $u\alpha v$ if $\alpha=\beta$ and undefined otherwise.

\begin{thmC}[{\cite[Theorem 8]{kozens96:kat:completeness:decidability}}]
  \label{thm:kat:classic}
  For all $e,f \in \termskat$, we have $\KATproves e = f$ iff $\G(e) = \G(f)$.
\end{thmC}

We now reprove this result using Kleene algebra with hypotheses.
We start by defining the additional axioms of KAT as hypotheses.

\begin{defi}
  We write \hbool for the set of all instances of Boolean algebra axioms
  over $\termsba$ and $\hglue$ for the following set
  of hypotheses relating the Boolean algebra connectives to the Kleene algebra ones:
  \begin{align*}
    \hglue &\eqdef
             \set{\phi\wedge\psi = \phi\cdot\psi,~
             \phi\vee\psi = \phi+\psi \mid \phi,\psi\in\termsba }
             \cup
             \set{\bot=0,~\top=1}
  \end{align*}
  We then define $\hkat \eqdef \hbool \cup \hglue$.
\end{defi}
\noindent
(Note that all these equations are actually understood as double inequations.)

We prove completeness of $\KA_{\hkat}$ in Section~\ref{sec:kat:complete} below, by constructing a suitable reduction. Recall that this means completeness w.r.t.\ the interpretation $\sem\hkat-$ in terms of closed languages.
Before proving completeness of $\KA_{\hkat}$, we compare it to the
classical completeness (Theorem~\ref{thm:kat:classic}).
First note that $\KA_{\hkat}$ contains the same axioms as Kleene algebra with tests, so that provability in $\KA_{\hkat}$ and KAT coincide: $\KAproves[\hkat] e=f$ iff $\KATproves e=f$.
Comparing the interpretation $\sem\hkat-$ to the guarded string interpretation $\G$
is slightly more subtle, and is the focus of the next subsection.

\subsection{Relation to guarded string interpretation}
\label{ssec:kat:gs}

We prove here that the (standard) guarded string interpretation of KAT and the $\hkat$-closed language interpretation we use in the present paper both yield the same semantics.
We only state the main lemmas we use to establish the correspondence, delegating detailed proofs to \cref{app:kat:gs}.

\begin{textAtEnd}[allend, category=kat]
We give a detailed proof of \cref{cor:kat:gs} in this appendix, following the path sketched in \cref{ssec:kat:gs}. This result relates the (standard) guarded string interpretation of KAT to the $\hkat$-closed language interpretation we use in the present paper.
\end{textAtEnd}

The key step consists in characterising the words that are present in the closure of a language of guarded strings (\cref{lem:gs} below). First observe that a guarded string may always be seen as a word over the alphabet $\Sigma+\termsba$. Conversely, a word over the alphabet $\Sigma+\termsba$ can always be uniquely decomposed as a sequence $\phis_0 a_0 \cdots \phis_{n-1} a_{n-1} \phis_n$ where $a_i \in \Sigma$ for all $i<n$ and each $\phis_i$ is a possibly empty sequence of Boolean expressions. We let $\phis$ range over such sequences, and we write $\overline\phis$ for the conjunction of the elements of $\phis$.

\begin{theoremEnd}[default,category=kat]{lem}
  \label{lem:gs}
  Let $L$ be a language of guarded strings. We have
  \begin{gather*}
    \phis_0 a_0 \cdots \phis_{n-1} a_{n-1} \phis_n\in\cl\hkat(L)\\
    \Leftrightarrow \qquad \forall~(\alpha_i)_{i\leq n},\quad (\forall i\leq n, \alpha_i\models\overline{\phis_i})\quad\Rightarrow\quad \alpha_0 a_0 \cdots \alpha_{n-1} a_{n-1} \alpha_n\in L
  \end{gather*}
\end{theoremEnd}
\begin{proofEnd}
  The proof for the left to right direction proceeds by induction on the closure---c.f.~\eqref{cind}.
  In the base case, where the word already belongs to $L$, and thus is already a guarded string, the $\phis_i$ must all be single atoms. Since $\alpha\models\beta$ iff $\alpha=\beta$ for all atoms $\alpha,\beta$, the condition is immediately satisfied.

  Otherwise, we have $\phis_0 a_0 \cdots \phis_{n-1} a_{n-1} \phis_n\in u\lang{e}v$ for $u,v\in (\Sigma+\termsba)^*$ and $e\leq f\in\hkat$, where guarded strings in $u\lang{f}v$ satisfy the property by induction.
  We proceed by a case distinction on the hypothesis in $e\leq f\in \hkat$.
  \begin{itemize}
  \item $e\leq f\in\hbool$: in this case $e=\phi$ and $f=\psi$ for
    some formulas $\phi,\psi$ such that $\BAproves \phi=\psi$.
    Therefore, we have $\alpha\models\phi$ iff $\alpha\models\psi$ for all atoms $\alpha$, so that the induction hypothesis on the unique word of $u\psi v$ is equivalent to the property we have to prove about the unique word of $u\phi v$.
  \item  $e=\phi\wedge \psi$ and $f=\phi\cdot \psi$:
    there must be $j$ such that $\phis_j=\phis (\phi\wedge\psi) \psis$ for $\phis,\psis\in\termsba^*$.
    The word of $u\phi\psi v$ has the same decomposition as $u(\phi\wedge\psi)v$, except that
    $\phis_j$ is replaced by $\phis\phi\psi\psis$. Since $\BAproves\overline{\phis \phi\psi\psis}=\overline{\phis(\phi\wedge\psi)\psis}$, the induction hypothesis immediately applies.
    \item The case for $e=\phi\cdot\psi$ and $f=\phi\wedge\psi$ is handled similarly.
    \item $e=\phi\vee \psi$ and $f=\phi+\psi$:
      there must be $j$ such that $\phis_j=\phis (\phi\vee\psi) \psis$ for $\phis,\psis\in\termsba^*$.
      We have $\lang{f}=\set{\phi,\psi}$ and the two words of $u\lang{f}v$ have the same decomposition as $u(\phi\vee\psi)v$, except that
      $\phis_j$ is replaced by $\phis \phi\psis$ or $\phis \psi\psis$.
      We use the induction hypothesis and the fact that $\alpha\models\phi\vee\psi$ iff $\alpha\models\phi$ or $\alpha\models\psi$.
    \item $e=\phi+\psi$ and $f=\phi\vee\psi$:
      there must be $j$ such that w.l.o.g. $\phis_j=\phis (\phi) \psis$ for $\phis,\psis\in\termsba^*$.
      We have $\lang{f}=\set{\phi\vee\psi}$ and the word in $u\lang{f}v$ has the same decomposition as $u(\phi)v$, except that
      $\phis_j$ is replaced by $\phis (\phi\vee\psi) \psis$.
      We use the induction hypothesis and the fact that if $\alpha\models\phi$ then $\alpha\models\phi\vee\psi$.
    \item $e=\bot$ and $f=0$: there must be $j$ such that $\bot$ belongs to $\phis_j$, so that there are no atoms $\alpha$ such that $\alpha\models\phis_j$: the condition is trivially satisfied.
    \item $e=0$ and $f=\bot$: in this case $u\lang{e}v$ is empty so this case is trivially satisfied.
    \item $e=\top$ and $f=1$: there is $j,\phis,\psis$ such that $\phis_j=\phis\top\psis$. The word of $uv$ decomposes like $u\top v$, except that $\phis_j$ is replaced by $\phis\psis$. Since $\BAproves \overline{\phis\psis}=\overline{\phis\top\psis}$, the induction hypothesis suffices to conclude.
    \item The case where $e=1$ and $f=\top$ is handled similarly.
  \end{itemize}
  We now prove the right-to-left implication.
  From the assumptions we obtain
  \begin{align*}
    \lang{
    \sum_
    {\alpha\models\overline{\phis_0}}
    \alpha} a_0 \ldots
    \lang{
    \sum_
    {\alpha\models\overline{\phis_{n-1}}}\alpha}
    a_{n-1}
    \lang{
    \sum_
    {\alpha\models\overline{\phis_n}}\alpha}\subseteq L
  \end{align*}
  We deduce using the glueing inequations for $\vee$ and $\bot$ that
  \begin{align*}
    \left(\bigvee_
    {\alpha\models\overline{\phis_0}}
    \alpha\right) a_0 \ldots
    \left(
    \bigvee_
    {\alpha\models\overline{\phis_{n-1}}}\alpha\right)
    a_{n-1}
    \left(
    \bigvee_
    {\alpha\models\overline{\phis_n}}\alpha\right)\in \cl\hkat(L)
  \end{align*}
  Since $\BAproves \phi=\bigvee_{\alpha\models\phi}\alpha$ for all $\phi$, the inequations from \hbool yield
  \begin{align*}
    \overline{\phis_0}
    a_0 \ldots
    \overline{\phis_{n-1}}
    a_{n-1}
    \overline{\phis_n}\in \cl\hkat(L)
  \end{align*}
  We conclude with the glueing inequations for $\wedge$ and $\top$, which give
  \begin{align*}
    \phis_0
    a_0 \ldots
    \phis_{n-1}
    a_{n-1}
    \phis_n\in \cl\hkat(L)\tag*\qedhere
  \end{align*}
\end{proofEnd}

\begin{textAtEnd}[allend, category=kat]
Now we turn to showing that the \hkat-closures of $\lang{e}$ and $\G(e)$ coincide (\cref{lem:kat:from:gs}), ultimately yielding \cref{cor:kat:gs}.
We first prove two lemmas about \hkat-closures of guarded string languages.
\end{textAtEnd}
\begin{theoremEnd}[allend, category=kat]{lem}
  \label{lem:kat:dot}
  For all guarded string languages $L,K$, $\cl\hkat(L\cdot K)\subseteq \cl\hkat(L\diamond K)$.
\end{theoremEnd}
\begin{proofEnd}
  We show that $L\cdot K\subseteq \cl\hkat(L\diamond K)$, which is sufficient. Take $w\in L\cdot K$. Hence, $w=xy$
  with $x\in L$ and $y\in K$. Because $x$ and
  $y$ are guarded strings, we know both of them begin and end in an atom. So we can write $x=x'\alpha$ and $y=\beta y'$ for some $\alpha,\beta\in \At$.
  \begin{itemize}
  \item If $\alpha=\beta$, then $x\alpha y\in L\diamond K$. In that
    case, we have $xy = x'\alpha \alpha y'$, but then, since $x$ and
    $y$ are guarded strings and $\alpha$ is the only atom
    s.t. $\alpha\models \alpha \wedge \alpha$, we get
    $xy \in \cl\hkat(L\diamond K)$ by \cref{lem:gs}.
  \item Otherwise, if $\alpha\neq\beta$, then \emph{a priori}
    $xy\not\in L\diamond K$. However, since there are no atoms
    $\gamma$ s.t. $\gamma \models \alpha \wedge \beta$, we have
    $xy \in \cl\hkat(L\diamond K)$ by \cref{lem:gs}.\qedhere
  \end{itemize}
\end{proofEnd}

\begin{textAtEnd}[allend, category=kat]
Given a guarded string language $L$, let us denote by $L^\diamond$ the language $L$ iterated w.r.t. the coalesced product $\diamond$: $L^\diamond=\bigcup_n L^{\diamond n}$, where $L^{\diamond 0}=\set{\alpha\mid\alpha\in\At}$ and $L^{\diamond n+1}=L\diamond L^{\diamond n}$.
\end{textAtEnd}
\begin{theoremEnd}[allend, category=kat]{lem}
  \label{lem:kat:str}
  For all guarded string languages $L$, $\cl\hkat(L^*)\subseteq \cl\hkat(L^\diamond)$.
\end{theoremEnd}
\begin{proofEnd}
  It suffices to prove
  $L^* \subseteq \cl\hkat(L^\diamond)$, which in turn
  follows once we prove $\varepsilon \in \cl\hkat(L^\diamond)$
  and $L\cdot \cl\hkat(L^\diamond) \subseteq \cl\hkat(L^\diamond)$.
  The former follows from the fact that $\proves[\hkat] 1\leq\sum_\alpha\alpha$, since $L^\diamond$ contains $L^{\diamond0}$ and thus all atoms $\alpha$.
  For the latter, we have
  \begin{align*}
    L \cdot \cl\hkat(L^\diamond)
    &\subseteq \cl\hkat(L\cdot L^\diamond) \tag{$\cl\hkat$ is contextual---\cref{lem:con}} \\
    &\subseteq \cl\hkat(L\diamond L^\diamond) \tag{\cref{lem:kat:dot}} \\
    &\subseteq \cl\hkat(L^\diamond)\tag*\qedhere
  \end{align*}
\end{proofEnd}

\begin{textAtEnd}[allend, category=kat]
We can finally prove \cref{lem:kat:from:gs}:
\end{textAtEnd}
Then we show that the \hkat-closures of $\lang{e}$ and $\G(e)$ coincide:
\begin{theoremEnd}[default,category=kat,no link to proof]{lem}
  \label{lem:kat:from:gs}
  For all KAT expressions $e$, $\sem\hkat e=\cl\hkat(\G(e))$.
\end{theoremEnd}
\begin{proofEnd}
  It suffices to prove that $\lang{e}\subseteq \cl\hkat(\G(e))$ and $\G(e)\subseteq \sem\hkat e$. We prove both results via induction on $e$.

  We start with $\lang{e}\subseteq \cl\hkat(\G(e))$.
  \begin{itemize}
  \item $0$: trivial as $\lang{0}=\emptyset$.
  \item $1$: $\lang{1} = \set\varepsilon$ and $\G(1)=\set{\alpha \mid \alpha\in\At}$,
    therefore $\varepsilon \in \cl\hkat(\G(1))$ by \cref{lem:gs}.

  \item $a \in \Sigma$: $\lang{a} = \set a$ and $\G(a)=\set{\alpha a \beta \mid \alpha,\beta\in\At}$,
    therefore $a \in \cl\hkat(\G(a))$ by \cref{lem:gs}.

  \item $\phi\in\termsba$: $\lang{\phi} = \set\phi$ and $\G(\phi)=\set{\alpha\mid \alpha\models\phi}$,
    hence $\phi \in \cl\hkat(\G(\phi))$ by \cref{lem:gs}.
  \item $e+f$: we have
    \begin{align*}
      \lang{e+f}=\lang{e}\cup\lang{f}
                 &\subseteq \cl\hkat(\G(e))\cup \cl\hkat(\G(f)) \tag{IH}\\
                 &\subseteq \cl\hkat(\G(e)\cup \G(f)) \tag{$\cl\hkat$ monotone}
                 = \cl\hkat(\G(e+f))
    \end{align*}
  \item $e\cdot f$: we have
    \begin{align*}
      \lang{e\cdot f}=\lang{e}\cdot\lang{f}
                      &\subseteq \cl\hkat(\G(e))\cdot \cl\hkat(\G(f)) \tag{IH}\\
                      &\subseteq \cl\hkat(\G(e)\cdot \G(f)) \tag{\cref{lem:con:clo}} \\
                      &\subseteq \cl\hkat(\G(e)\diamond\G(f)) \tag{\cref{lem:kat:dot}}
                      = \cl\hkat(\G(e\cdot f))
    \end{align*}
  \item $e^*$: we have
    \begin{align*}
      \lang{e^*}=\lang{e}^*
                 &\subseteq \cl\hkat(\G(e))^* \tag{IH}\\
                 &\subseteq \cl\hkat(\G(e)^*) \tag{\cref{lem:con:clo:str}} \\
                 &\subseteq \cl\hkat(\G(e)^\diamond) \tag{\cref{lem:kat:str}}
                 = \cl\hkat(\G(e^*))
    \end{align*}
  \end{itemize}
  Next we prove $\G(e)\subseteq \sem\hkat e$:
  \begin{itemize}
  \item $0$: trivial as $\G(0)=\emptyset$.
  \item $1$: $\G(1)=\set{\alpha \mid \alpha\in\At }$;
    for all $\alpha$, $\proves[\hkat] \alpha\leq 1$, whence $\alpha\in\cl\hkat{\set\varepsilon}=\sem\hkat1$.
  \item $a \in \Sigma$: $\G(a)=\set{\alpha a \beta \mid \alpha,\beta\in\At}$;
    like above, for all $\alpha,\beta$, $\alpha a \beta\in\cl\hkat{\set a}=\sem\hkat a$.
  \item $\phi\in\termsba$: $\G(\phi)=\set{\alpha\mid \alpha\models\phi}$,
    for all $\alpha$ such that $\alpha\models\phi$, $\proves[\hkat] \alpha\leq \phi$, whence $\alpha\in\cl\hkat{\set\phi}=\sem\hkat\phi$.
  \item $e+f$: we have
    \begin{align*}
      \G(e+f)&=\G(e)\cup\G(f) \\
                 &\subseteq \sem\hkat e\cup \sem\hkat f \tag{IH}\\
                 &\subseteq \cl\hkat(\lang e\cup \lang f) \tag{$\cl\hkat$ monotone} \\
                 &= \sem\hkat{e+f})
    \end{align*}
  \item $e\cdot f$: suppose $w\in\G(e\cdot f)$, i.e., $w=x\alpha y$
    for $x\alpha\in\G(e)$ and $\alpha y\in\G(f)$. Via the induction
    hypothesis we know that $x\alpha\in\sem\hkat e$ and
    $\alpha y\in\sem\hkat f$.  Hence
    $x\alpha\alpha y\in \sem\hkat e\cdot
    \sem\hkat f\subseteq \sem\hkat{e\cdot f}$ via \cref{lem:con:clo}.  As
    $\proves[\hkat] \alpha\leq \alpha \alpha$, we get that
    $x\alpha y \in \sem\hkat{e\cdot f}$.
  \item $e^*$: a similar argument works.\qedhere
  \end{itemize}
\end{proofEnd}

Let $\GS$ be the set of all guarded strings; we have:
\begin{textAtEnd}[allend,category=kat]
Recall that $\GS$ is the set of all guarded strings.
We deduce from \cref{lem:gs} that closing a guarded string language under \hkat does not add new guarded strings:
\begin{lem}\label{lem:noextraguard}
For all guarded string language $L$, $L=\cl\hkat(L)\cap \GS$.
\end{lem}
\begin{proof}
The left-to-right direction is trivial, as $L\subseteq \cl\hkat(L)$ and all strings in $L$ are guarded.
For the right-to-left direction,
given a guarded string $\alpha_0 a_0\dots a_{n-1} \alpha_n$ in $\cl\hkat(L)$, we simply use \cref{lem:gs} with the sequence of $\alpha_i$ themselves: we have $\alpha_i\models\alpha_i$ for all $i$, so that $\alpha_0 a_0\dots a_{n-1} \alpha_n$ actually belongs to $L$.
\end{proof}
\end{textAtEnd}
\begin{theoremEnd}[default,category=kat]{lem}
  \label{lem:gs:from:kat}
  For all KAT expressions $e$, $\G(e)=\sem\hkat e\cap\GS$.
\end{theoremEnd}
\begin{proofEnd}
  Immediate from \cref{lem:noextraguard} and \cref{lem:kat:from:gs}.
\end{proofEnd}

As an immediate consequence of \cref{lem:kat:from:gs} and \cref{lem:gs:from:kat}, we can finally relate
the guarded strings languages semantics to the \hkat-closed languages
one:
\begin{theoremEnd}[default,category=kat]{cor}
  \label{cor:kat:gs}
  Let $e,f\in\termskat$. We have $\G(e)=\G(f) \Leftrightarrow \sem\hkat e=\sem\hkat f$.
\end{theoremEnd}

\subsection{Completeness}
\label{sec:kat:complete}
To prove completeness of the \hkat-closed language model, we take the following steps:
\begin{enumerate}
\item We reduce the hypotheses in \hkat to a simpler set of axioms:
  by putting the Boolean expressions into normal forms via the atoms,
  we can get rid of the hypotheses in \hbool. We do not remove
  the hypotheses in \hglue directly: we transform them into the
  following hypotheses about atoms:
  \begin{align*}
    \hatom_0 &\eqdef \set{\alpha\cdot\beta \leq 0\mid \alpha,\beta\in\At,~\alpha\neq\beta}\\
    \hatom_1 &\eqdef \set{\alpha\leq 1\mid \alpha\in\At}\\
    \hatom_2 &\eqdef \set{1\leq \textstyle\sum_{\alpha\in\At}\alpha }
  \end{align*}
  We thus first show that \hkat reduces to $\hatom \eqdef \hatom_0 \cup \hatom_1 \cup \hatom_2$.
\item Then we use results from \cref{sec:tools} to construct a reduction from \hatom to the empty set, and thereby obtain completeness of $\KA_{\hkat}$.
\end{enumerate}

Let $r\colon \termska(\Sigma + \termsba) \to \termska(\Sigma+\At)$ be the homomorphism defined by
\begin{align*}
  r(x) &\eqdef
         \begin{cases}
           a & x = a \in \Sigma\\
           \sum_{\alpha \models \phi} \alpha & x = \phi \in \termsba
         \end{cases}
\end{align*}
We show below that $r$ yields a reduction from \hkat to \hatom, using \cref{prop:hred}.

\begin{lem}\label{lem:atomlaws}
  We have $\proves[\hatom] 1=\sum_{\alpha\in\At}\alpha$, and
  for all $\alpha$, $\proves[\hatom] \alpha\alpha=\alpha$.
\end{lem}
\begin{proof}
  The first equation comes from $\hatom_1$ and $\hatom_2$
  For the second one,
  \begin{align*}
    \proves[\hatom] \alpha  = \alpha \cdot 1 =\alpha \cdot \sum_{\beta\in\At}\beta = \sum_{\beta\in\At} \alpha\cdot \beta = \alpha\cdot\alpha
  \end{align*}
  where we use $\hatom_0$ in the last step.
\end{proof}

\begin{lem}\label{lem:atom:r:kat}
  For all $e\leq f\in \hkat$, we have that $\proves[\hatom] r(e)\leq r(f)$.
\end{lem}
\begin{proof}
  We have five families of equations to consider, and thus ten families of inequations. We derive  equalities directly, so that we only have to consider five cases:
  \begin{itemize}
  \item $\bot=0$: we derive
    \begin{align*}
      \proves r(\bot)=\sum_{\alpha \models \bot} \alpha =0=r(0)
    \end{align*}
  \item $\phi\vee\psi = \phi+\psi$: we derive
    \begin{align*}
        \proves
         r(\phi+\psi)
      = r(\phi)+r(\psi)
      = \sum_{\alpha\models\phi}\alpha+\sum_{\alpha\models\psi}\alpha 
      = \sum_{\alpha\models\phi\text{ or }\alpha\models\psi}\alpha 
      = \sum_{\alpha\models\phi\vee\psi}\alpha
       = r(\phi\vee\psi)
    \end{align*}
  \item $\top=1$: we derive
    \begin{align*}
      \proves[\hatom] r(\top)=\sum_{\alpha \models\top}\alpha=\sum_{\alpha\in\At}\alpha= 1=r(1)
    \end{align*}
    (using \cref{lem:atomlaws} for the last but one step)
  \item $\phi\wedge\psi = \phi\cdot \psi$: we derive
    \begin{align*}
      \proves[\hatom]
         r(\phi\cdot \psi)
       = r(\phi)\cdot r(\psi)
       = \sum_{\alpha\models\phi}\alpha \sum_{\beta \models\psi} \beta
      &= \sum_{\alpha\models\phi,\text{ } \beta\models\psi}\alpha\beta \\
      &= \sum_{\alpha\models\phi,\text{ } \alpha\models\psi}\alpha\alpha \tag{$\hatom_0$}\\
      &= \sum_{\alpha\models\phi,\text{ } \alpha\models\psi}\alpha \tag{\cref{lem:atomlaws}}\\
      &= \sum_{\alpha\models\phi\wedge\psi}\alpha
       = r(\phi\wedge  \psi)
    \end{align*}
  \item $\phi = \psi$, an instance of a Boolean algebra axiom: in this case, $\alpha\models\phi$ iff $\alpha\models\psi$ for all $\alpha$, so that $r(\phi)$ and $r(\psi)$ are identical.\qedhere
  \end{itemize}
\end{proof}

\begin{lem}
  \label{lem:kat:atom}
The homomorphism $r$ yields a reduction from \hkat to \hatom.
\end{lem}
\begin{proof}
  We use \cref{prop:hred}. We first need to show $\proves[\hkat]\hatom$: for
  $\alpha,\beta\in\At$ with $\alpha\neq\beta$, we have the following derivations in $\KA_\hkat$
  \begin{mathpar}
  \alpha\cdot \beta = \alpha\wedge\beta = \bot = 0
  \and
  \alpha \leq \top = 1
  \and
  1= \top = \bigvee_{\alpha\models\top} \alpha = \sum_{\alpha\in\At} \alpha
  \end{mathpar}
  Now for $a\in \Sigma+\At$, we have $a=r(a)$ (syntactically):
  if $a=\ltr a\in\Sigma$, then $r(\ltr a)=\ltr a$; if $a=\alpha\in\At$, then  $r(\alpha)=\sum_{\alpha\models\alpha} \alpha = \alpha$.
  Condition~\ref{hred:idem} is thus satisfied, and it suffices to verify condition~\ref{hred:proves:r} for $\phi\in\termsba$. In this case, we have
  we have $\proves[\hkat] r(\phi)=\sum_{\alpha \models\phi} \alpha =\bigvee_{\alpha\models\phi} \alpha = \phi$. The last condition~(\ref{hred:hyps}) was proven in \cref{lem:atom:r:kat}.
\end{proof}

We can finally conclude:
\begin{thm}\label{thm:kat}
  For all $e,f\in\termskat$, $\sem\hkat e=\sem\hkat f$ implies $\KAproves[\hkat] e=f$.
\end{thm}
\begin{proof}
  By completeness of KA (\cref{thm:ka}) and \cref{thm:red}, it suffices to show that
  \hkat reduces to the empty set.
  By \cref{lem:kat:atom}, \hkat reduces to \hatom, 
  which reduces to $\hatom_{1,2} \eqdef \hatom_1 \cup \hatom_2$ by \cref{lem:red:0}, which reduces to the empty set by~\cref{lem:reds}\ref{it:red:1S}. We conclude by \cref{lem:red:comp}.
\end{proof}

\section{More tools}
\label{sec:moretools}

Before dealing with other examples, we define a few more tools for combining reductions.
To alleviate notation in the sequel, we often omit the functional composition symbol $(\circ)$, thus writing $H_1H_2$ for the composition $H_1\circ H_2$, or $\clH n\dots \clH 0$ for $\clH n\circ\dots\circ \clH 0$.

\subsection{Combining more than two reductions}
\label{ssec:more:compo}

We start with a variant of \cref{lem:cup:2} to compose more than two reductions, in the case they all share a common target.
A similar lemma is formulated in the setting of bi-Kleene algebra~\cite[Lemma~4.49]{kappethesis}.
\begin{prop}
  \label{prop:cup}
  Let $H_0, \ldots, H_n, H'$ be sets of hypotheses over a common alphabet $\Sigma$.
  Let $H=\bigcup_{i \leq n} H_i$.
  If $H_i$ reduces to $H'$ for all $i$ and  $\cl H \subseteq\clH n \dots \clH0$,
  then $H$ reduces to $H'$.
\end{prop}
\begin{proof}
  The proof follows the same pattern as that of \cref{lem:cup:2}. Unfortunately, we cannot simply iterate that lemma as a black box: the assumption $\cl H =\clH n \dots \clH0$ does not provide enough information about closures associated to subsets of $H_0,\dots,H_n$.

  Let $r_0, \ldots, r_n$ be the reductions from $H_0, \ldots, H_n$ to $H'$.
  We show that $r_n \circ \ldots \circ r_0$ is a reduction from $H$ to $H'$.
  \begin{enumerate}
  \item Since $\proves[H_0] H'$, we have $\proves[H] H'$.
  \item By an easy induction on  $k\le n$, we obtain that for all $e$, $\proves[H] e = r_k \ldots r_0(e)$; whence the second requirement for $k=n$ holds.
  \item For the third requirement, note that the restriction to the codomain alphabet is void since there is a single alphabet (so that we simply have $\sem{H_i}e \subseteq \sem{H'}{r_i(e)}$ for all $i,e$).
    We prove by induction on $k\le n$ that for all $e$,  $\clH k \ldots \semH0e \subseteq \sem{H'}{r_k \ldots r_0(e)}$.
    \begin{itemize}
    \item for $k=0$, the statement holds by assumption about $r_0$: $\semH0e\subseteq\sem{H'}{r_0(e)}$;
    \item for $0<k<n$, we have
      \begin{align*}
        & \clH k \clH {k-1} \ldots \semH0e  \\
        \subseteq~& \clH k\sem{H'}{r_{k-1} \ldots r_0(e)} \tag{IH}\\
        \tag{$\proves[H_k] H'$ and \cref{lem:clo:incl}}
        \subseteq~& \clH k\sem{H_k}{r_{k-1} \ldots r_0(e)}\\
        \tag{$\clH k$ a closure}
        =~& \semH k{r_{k-1} \ldots r_0(e)}\\
        \tag{$r_k$ is a reduction}
        \subseteq~& \sem{H'}{r_k r_{k-1} \ldots r_0(e)}
      \end{align*}
    \end{itemize}
    Together with $\cl H\subseteq\clH n \dots \clH0$, the case $k=n$ finally gives $\sem H e\subseteq\sem{H'}{r_n \ldots r_0(e)}$, as required.
    \qedhere
  \end{enumerate}
\end{proof}

\subsection{Organising closures}
\label{ssec:more:gsos}

The previous proposition makes it possible to combine reductions for various sets of hypotheses, provided the closure associated to the union of those hypotheses can be presented as a sequential composition of the individual closures.

The two results below help to obtain such presentations more easily.
They are stronger (i.e., have weaker hypotheses) than similar results proposed in the past (e.g., \cite[Lemma 4.50]{kappethesis}, \cite[Lemma 3.10]{prw:ramics21:mkah})
As before, these results are best presented at a much higher level of generality; they are proved in \cref{app:closures}.

\begin{lem}
  \label{lem:gsos:dl}
  Let $H_0,\dots,H_n$ be sets of hypotheses; write $\Hlt j$ for $\bigcup_{i<j}H_i$.
  Suppose that for all $j\leq n$, we have $\clHlt j\clH j \subseteq \clH j\clHlt j$.
  Then $\clHlt{n+1}=\clH n\dots \clH 0$.
\end{lem}
\begin{proof}[Proof idea]
  By iterating \cref{lem:dl:2}---see~\cref{prop:gsos:dl:gen}.
\end{proof}

\begin{prop}
  \label{prop:gsos}
  Let $H_0,\dots,H_n$ be sets of hypotheses such that $H_i$ is affine except possibly for $H_n$;
  write $\Hlt j$ for $\bigcup_{i<j}H_i$ as before.

  \noindent
  If for all $j\leq n$ we have either
  $\begin{cases}
  (1)~\forall i<j,~H_i\,H_j\subseteq \clH j\,\Hlt j^=\,, \text{ or}\\
  (2)~\forall i<j,~H_i\,H_j\subseteq H_j^=\,\clHlt j\,,
  \end{cases}$
  then $\clHlt{n+1}=\clH n\dots \clH 0$.
\end{prop}
\begin{proof}[Proof idea]
  By combining \cref{lem:gsos:dl} and \cref{lem:gsos:2}---see~\cref{cor:gsos:gen}.
\end{proof}

\subsection{Proving partial commutations}
\label{ssec:more:overlaps}

The last proposition above is convenient as it makes it possible to analyse only compositions of single rewriting steps (i.e., functions of the shape $H_iH_j$). In the remainder of this section we provide additional tools to analyse such compositions (\cref{prop:overlaps}).

We first need to discuss an important property satisfied by all closures associated to sets of hypotheses.
\begin{defi}
  \label{def:con}
  A monotone function $f$ on languages is \emph{contextual} if for all words $u,v$ and languages $K$, we have $u\cdot f(K)\cdot v\subseteq f(u\cdot K\cdot v)$.
\end{defi}
\begin{lem}
  \label{lem:con}
  The identity function is contextual, the composition of contextual functions is contextual, the union of contextual functions is contextual. Given a set $H$ of hypotheses, the functions $H$ and $H^\star$ are contextual.
\end{lem}
\begin{proof}
  See \cref{app:con}.
\end{proof}

For closures, we actually have:
\begin{lem}
  \label{lem:con:clo}
  Let $H$ be a set of hypotheses.
  For all languages $L,K$ we have
  \begin{align*}
  \cl H(L)\cdot \cl H(K)\subseteq \cl H(L\cdot K)\,.
  \end{align*}
\end{lem}
\begin{proof}
  See \cref{lem:con:clo:app}.
\end{proof}
Note that such a property typically does not hold for single step functions $H$: take for instance $H=\set{a\leq b}$ and $L=K=\set{b}$, for which we have $H(L\cdot K)=H(\set{bb})=\set{ab,ba}$ while $H(L)\cdot H(K)=H(\set{b})\cdot H(\set{b})=\set{a}\cdot \set{a}=\set{aa}$.

Also note that the converse inequality does not hold in general: take for instance $H=\set{1\leq aa}$ and $L=K=\set{a}$, for which we have $\cl H(L\cdot K)=\cl H(\set{aa})=\set{1,aa}$ while $\cl H(L)\cdot \cl H(K)=\set{a}\cdot\set{a}=\set{aa}$.

The above property actually makes it possible to construct models of $\KA_H$~\cite{dkpp:fossacs19:kah}, and to prove soundness (\cref{thm:soundness}, cf. \cref{app:soundness}).

\medskip

Now we define a notion of \emph{overlap} between words; together with contextuality, this will help to analyse all potential interactions between two sets of hypotheses.

\begin{defi}
  An \emph{overlap} of two words $u,v$ is a tuple $\tuple{x,y,s,t}$ of words such that $xuy=svt$ and either:

  \noindent
  \begin{minipage}{.5\linewidth}
    \begin{itemize}
    \item $x,t$ are empty and $|y|<|v|$; or
    \item $y,s$ are empty and $|x|<|v|$; or
    \item $x,y$ are empty and $s,t$ are non-empty; or
    \item $s,t$ are empty and $x,y$ are non-empty.
    \end{itemize}
  \end{minipage}
  \begin{minipage}{.5\linewidth}
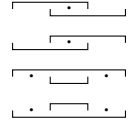
    \begin{center}
  \begin{tikzpicture}[xscale=.5,yscale=.45]
    \node at (1.5,3) {$\cdot$};
    \draw (0,3) -- +(0,+0.2)-| (2,3);
    \draw (1,3) -- +(0,-0.2)-| (3,3);
    \node at (1.5,2) {$\cdot$};
    \draw (0,2) -- +(0,-0.2)-| (2,2);
    \draw (1,2) -- +(0,+0.2)-| (3,2);
    \node at (0.5,1) {$\cdot$};
    \node at (2.5,1) {$\cdot$};
    \draw (0,1) -- +(0,+0.2)-| (3,1);
    \draw (1,1) -- +(0,-0.2)-| (2,1);
    \node at (0.5,0) {$\cdot$};
    \node at (2.5,0) {$\cdot$};
    \draw (0,0) -- +(0,-0.2)-| (3,0);
    \draw (1,0) -- +(0,+0.2)-| (2,0);
  \end{tikzpicture}    
\end{center}
  \end{minipage}
\end{defi}
The idea is that overlaps of $u,v$ correspond to the various ways $u$ and $v$ may overlap.
Those situations are depicted on the right of each item. In each case the upper segment denotes $u$ while the lower segment denotes $v$, and a dot indicates a non-empty word: we do forbid trivial overlaps (for instance, that the end of $u$ overlaps with the beginning of $v$ via the empty word does not count.)

\begin{fact}
  If $\tuple{x,y,s,t}$ is an overlap of $u,v$ then
 $\tuple{s,t,x,y}$ is an overlap of $v,u$. Moreover,
  \begin{itemize}
  \item There is no overlap of $1$ and a word of length smaller or equal to one.
  \item There is a single overlap of $1$ and a word $ab$ of size two: $\tuple{a,b,1,1}$.
  \item The only overlap between a letter and itself is $\tuple{1,1,1,1}$.
  \item There is no overlap of two distinct letters.
  \item There are at most two overlaps of a letter and a word of size two.
  \item There are at most three overlaps of words of size two.
  \end{itemize}
\end{fact}
As an example of the last item, $\tuple{1,c,a,1}$ is the only overlap of $ab,bc$ (provided $a,b,c$ are pairwise distinct letters); and $\tuple{1,1,1,1}$, $\tuple{a,1,1,a}$, and $\tuple{1,a,a,1}$ are the three overlaps between $aa$ and itself.

The key property of overlaps is the following lemma.
\begin{lem}
  \label{lem:overlaps}
  For all words $x,u,y,s,v,t$ such that $xuy=svt$, we have either:
  \begin{itemize}
  \item $|xu|\le|s|$, or
  \item $|sv|\le|x|$, or
  \item there are words $l,r$ such that $x=lx'$, $s=ls'$, $y=y'r$, and $t=t'r$ for some overlap
    $\tuple{x',y',s',t'}$ of $u,v$.
  \end{itemize}
\end{lem}

\begin{prop}
  \label{prop:overlaps}
  Let $H_i$ be a set of hypotheses whose RHS (right-hand side) are words.
  Let $H_j$ be a set of hypotheses whose LHS (left-hand side) are words.
  Let $g$ be a contextual function. 
  
  If for all hypotheses $e\leq u$ of $H_i$, all hypotheses $v\leq f$ of $H_j$, and
  all overlaps $\tuple{x,y,s,t}$ of $u,v$, we have $x\lang e y\subseteq g(s\lang f t)$, then $H_iH_j\subseteq H_jH_i\cup g$.
\end{prop}
\begin{proof}
  Suppose $w\in H_i(H_j(L))$.
  There is a hypothesis $e\leq u$ of $H_i$ such that $w=xu_0y$ with $u_0\in\lang e$ and $xuy\in H_j(L)$.
  In turn, there is a hypothesis $v\leq f$ of $H_j$ such that $xuy=svt$ with $s\lang f t\subseteq L$.
  According to \cref{lem:overlaps}, there are three cases to consider.
  \begin{itemize}
  \item Either $|xu|\le|s|$, in which case $s=xus'$ and $y=s'vt$ for some word $s'$.
    For all words $v_0\in\lang f$, we have $xus'v_0 t\in s\lang f t\subseteq L$, whence $xu_0 s'v_0 t\in x\lang e s'v_0 t\subseteq H_i(L)$; in other words, $xu_0 s'\lang f t\subseteq H_i(L)$.
    We deduce that $xu_0y=xu_0s'vt \in H_j(H_i(L))$.
  \item Or $|sv|\le|x|$, in which case we proceed symmetrically.
  \item Or $x=lx'$, $s=ls'$, $y=y'r$, and $t=t'r$ for some words $l,r$ and some overlap
    $\tuple{x',y',s',t'}$ of $u,v$. We have $x'\lang e y'\subseteq g(s'\lang f t')$ by assumption, we conclude by contextuality of $g$:
    \begin{align*}
      \tag*{\qedhere}
      xu_0y=lx'u_0y'r\in lx' \lang e y'r\subseteq l\cdot g(s'\lang f t')\cdot r\subseteq g( ls'\lang f t'r)=g( s\lang f t)\subseteq g(L)
    \end{align*}
  \end{itemize}
\end{proof}

Given two sets of hypotheses $H_i,H_j$, we say there is an overlap between $H_i$ and $H_j$ if there is any overlap between a RHS of $H_i$ and an LHS of $H_j$.

\begin{cor}
  \label{cor:nooverlap}
  Let $H_i$ be a set of hypotheses whose RHS are words.
  Let $H_j$ be a set of hypotheses whose LHS are words.
  If there are no overlaps between $H_i$ and $H_j$, then we have $H_iH_j\subseteq H_jH_i$.
\end{cor}

In the sequel we will see examples where we have to combine more than two sets of hypotheses, and where there are some overlaps. In those cases, we will use \cref{prop:overlaps} with a composite function $g$ obtained from the various sets of hypotheses at hand: such functions will always be contextual by \cref{lem:con}, and they will be contained in functions like $\clH j\,\Hlt j^=$ or $H_j^=\,\clHlt j$, so that \cref{prop:gsos} may subsequently be applied.

\begin{exa}
  A typical use of the above proposition is the following. Let $H_0=\set{ac\leq cc}$, $H_1=\set{aaa\leq ab}$, and $H_2=\set{bc\leq ccc}$ for three distinct letters $a,b,c$.
  To analyse the composition $H_1H_2$, it suffices by \cref{prop:overlaps} to consider the overlaps of $ab$ and $bc$. There is only one such overlap, $\tuple{1,c,a,1}$, so that we have only to show
  $1\lang e c\subseteq g(a\lang f 1)$, where $e=aaa$ is the LHS of $H_1$, $f=ccc$ is the RHS of $H_2$, and $g$ is a function to be defined. Simplifying both sides, we get to prove
  $aaac\in g(\set{accc})$. Since $aaac\leftsquigarrow_{H_0}aacc\leftsquigarrow_{H_0}accc$, we can take $g=H_0H_0$ (or $g=\clH0$). We depict the overall situation as follows:
  \closeoverlapii{H_1}{H_2}{H_0}{H_0}{aaac}{abc}{accc}{aacc}%
  Thus we conclude that $H_1H_2\subseteq H_2H_1\cup H_0H_0$, and we can weaken this bound to
  \begin{align*}
    H_1H_2\subseteq H_2H_1\cup H_0H_0\subseteq H_2H_1\cup \clH0 \subseteq H^=_2(H_1\cup \clH0)\subseteq H^=_2\clHlt2\,.
  \end{align*}
\end{exa}

\subsection{Universal hypotheses}
\label{ssec:univ}

We prove one last generic result before moving to examples, which makes it possible to simplify the presentation of certain sets of hypotheses which are universally quantified on terms to their restriction where the quantification is only on words. This is useful in \cref{sec:katf,sec:katc}.

We call a function $h$ on regular expressions \emph{monotone} if for all expressions $e,f$, $\lang e\subseteq \lang f$ implies $\lang{h(e)}\subseteq\lang{h(f)}$; we call it \emph{linear} if for all expressions $e$, $\lang{h(e)}=\bigcup_{w\in\lang e} \lang{h(w)}$.
\begin{prop}
  \label{prop:reg:ctx}
  Let $l,r$ be functions on regular expressions, with $l$ linear and $r$ monotone.
  Let $H\eqdef \set{l(e)\leq r(e) \mid e\in\termska(\Sigma)}$ and $H'\eqdef \set{l(w)\leq r(w) \mid w\in \Sigma^*}$.
  We have $H=H'$, as functions on languages.
\end{prop}
\begin{proof}
  We have $H'\subseteq H$ as sets and thus as functions, so that it suffices to prove the converse inclusion.
  Suppose $u\in H(L)$: $u\in x\lang{l(e)} y$ for some $x,e,y$ such that $x\lang{r(e)}y\subseteq L$.
  By linearity of $l$, $u\in x\lang{l(w)} y$ for some word $w\in\lang e$.
  By monotonicity of $r$, we have $x\lang{r(w)}y\subseteq x\lang{r(e)}y$, whence $x\lang{r(w)}y\subseteq L$, and thus $u\in H'(L)$.
\end{proof}

\section{Kleene Algebra with Observations}
\label{sec:kao}

A \emph{Kleene algebra with Observations (KAO)} is a Kleene algebra which also contains a Boolean algebra,
but the connection between the Boolean
algebra and the Kleene algebra is different than for KAT: instead of
having the axioms $\top=1$ and $\phi\wedge \psi=\phi\cdot \psi$ for all
$\phi,\psi\in\termsba$, we only have
$\phi\wedge \psi \leq \phi\cdot \psi$~\cite{kao-2019}. This system was
introduced to allow for concurrency and tests in a Kleene algebra
framework, because associating $\phi\cdot \psi$ and $\phi\wedge \psi$
in a concurrent setting is no longer appropriate: $\phi\wedge \psi$ is
one event, where we instantaneously test whether both $\phi$ and $\psi$ are
true, while $\phi\cdot \psi$ performs first the test $\phi$, and then
$\psi$, and possibly other things can happen between those tests in
another parallel thread. Hence, the behaviour of $\phi\wedge \psi$
should be included in $\phi\cdot \psi$, but they are no longer
equivalent.

Algebraically this constitutes a small change, and an ad-hoc completeness proof is in~\cite{kao-2019}. Here we show how to obtain completeness within our framework. We also show how to add back the natural axiom $\top=1$, which is not present in~\cite{kao-2019}, and thereby
emphasise the modular aspect of the approach.

Note that even if we add the axiom $\top=1$, in which case we have that $\phi\cdot \psi$ is below both $\psi$ and $\phi$, $\phi\cdot \psi$ and $\phi\wedge\psi$ do not collapse in this setting, because $\phi\cdot \psi$ need not be an element of the Boolean algebra. In particular, in contrast to KAT, we do not have $\alpha\cdot\beta=0$ when $\alpha,\beta$ are distinct atoms.

Similar to KAT, we add the additional axioms of KAO to KA as hypotheses. The additional axioms of KAO are the axioms of Boolean algebra and the axioms specifying the interaction between the two algebras. The KAO-terms are the same as the KAT-terms: regular expression over the alphabet $\Sigma+\termsba$.

\begin{defi}
  We define the set of hypotheses $\hkao\eqdef\hbool\cup\hgluekao$, where
  \begin{align*}
    \hgluekao &\eqdef
                \set{\phi\wedge\psi \leq \phi\cdot\psi,~
                \phi\vee\psi = \phi+\psi \mid \phi,\psi\in\termsba }
                \cup
                \set{\bot=0}
  \end{align*}
\end{defi}
We prove completeness with respect to the \hkao-closed interpretation $(\sem\hkao-)$. As shown below, this also implies completeness for the language model presented in~\cite{kao-2019}.
We take similar steps as for KAT:
\begin{enumerate}
\item Reduce \hkao to a simpler set of axioms, $\hcontr \eqdef \set{\alpha\leq \alpha\cdot\alpha \mid \alpha \in \At}$, where $\At = 2^{\Omega}$ is the same set of atoms as in \cref{sec:kat}.
\item Use results from \cref{sec:tools} to reduce \hcontr to the empty set.
\end{enumerate}
For the first step, we use the same homomorphism $r$ as for KAT.
\begin{lem}\label{lem:contr:r:kao}
  For all $e\leq f\in \hkao$, we have $\proves[\hcontr] r(e)\leq r(f)$.
\end{lem}
\begin{proof}
  Similar to the proof of \cref{lem:atom:r:kat}.  First observe that
  in that proof, we needed the hypotheses of \hkat only for the
  $\top=1$ case, which is not there, and for the
  $\phi\wedge\psi=\phi\cdot\psi$ case, which is now only an
  inequation, and which is dealt with as follows:
  \begin{align*}
    \proves[\hcontr]
       r(\phi\wedge  \psi)
    = \sum_{\alpha\models\phi\wedge \psi} \alpha
    &\leq \sum_{\alpha\models\phi\wedge \psi} \alpha \cdot  \alpha \tag{\hcontr}\\
    &\leq \sum_{\alpha\models\phi,\text{ }\beta\models\psi} \alpha \cdot  \beta
    =  \sum_{\alpha\models\phi} \alpha \cdot \sum_{\beta\models\psi} \beta
    = r(\phi)\cdot r(\psi)=r(\phi\cdot \psi)\tag*{\qedhere}
  \end{align*}
\end{proof}

\begin{lem}
  \label{lem:kao:contr}
  The homomorphism $r$ yields a reduction from \hkao to \hcontr.
\end{lem}
\begin{proof}
  Like for \cref{lem:kat:atom}, we use \cref{prop:hred}.
  We show $\proves[\hkao]\hcontr$: for
  $\alpha\in\At$, we have
  $
    \proves[\hkao] \alpha = \alpha \wedge \alpha \leq \alpha\cdot \alpha
  $.
  The first and second condition about $r$ are obtained like in the
  KAT case: the glueing equations for $\wedge$ were not necessary
  there. The third and last condition was proven in
  \cref{lem:contr:r:kao}.
\end{proof}

\begin{thm}\label{thm:kao}
  For all $e,f\in\termskat$, $\sem\hkao e=\sem\hkao f$ implies $\KAproves[\hkao] e=f$.
\end{thm}
\begin{proof}
  The set of hypotheses \hkao reduces to \hcontr (\cref{lem:kao:contr}), which reduces to $\emptyset$ by
 \cref{lem:reds}\ref{it:red:aw}, as both $\alpha$ and $\alpha\cdot\alpha$ are words and $\alpha$ is a word of length $1$. 
\end{proof}
Note that the semantics defined
in~\cite{kao-2019} actually corresponds to $\sem\hcontr{r(-)}$ rather than $\sem\hkao-$. These semantics are nonetheless equivalent, \hkao reducing to \hcontr via $r$ (the proof of \cref{thm:red} actually establishes that when $H$ reduces to $H'$ via $r$ and $\KA_{H'}$ is complete, we have
$\sem H e=\sem H f$ iff
$\sem{H'}{r(e)}=\sem{H'}{r(f)}$ for all expressions $e,f$).

\medskip

Because we set up KAO in a modular way, we can now easily extend it with the extra axiom $\top=1$ (KA with ``bounded observations''---KABO). Revisiting the proofs that $r$ is a reduction from \hkat to \hatom and from \hkao to $\hcontr$, we can see that $r$ is also a reduction from $\hkabo\eqdef\hkao\cup\set{\top=1}$ to $\hatom_{1,2}\cup\hcontr$. Therefore, it suffices to find a reduction from $\hatom_{1,2}\cup\hcontr$ to the empty set.
Thanks to \cref{prop:cup}, it suffices to decompose the closure $\cl{(\hatom_{1,2}\cup\hcontr)}$ into a sequential composition of closures for which we do have reductions.
We use \cref{prop:gsos} and~\cref{prop:overlaps} to obtain such a decomposition: by analysing the overlaps between the various components of $\hatom_{1,2}\cup\hcontr$, we can find how to order them sequentially.

\begin{lem}
  \label{lem:kao:com}
  We have
  \begin{enumerate}[(i)]
  \item $\hatom_1\circ \hcontr \subseteq \hcontr \circ \hatom_1$\,,
  \item $\hatom_1\circ \hatom_2 \subseteq \hatom_2 \circ \hatom_1$\,,
  \item $\hcontr\circ \hatom_2 \subseteq \hatom_2\circ\hcontr \cup \hcontr \circ \hcontr$\,.
  \end{enumerate}
\end{lem}
\begin{proof}
  For the first two items, we use \cref{cor:nooverlap}: there are no overlaps between $1$ and $\alpha$, and no overlaps between $1$ and itself.
  For the third item, we use \cref{prop:overlaps} with the function $g=\hcontr \circ \hcontr$: the only overlap between $\alpha\alpha$ and $1$ is $\tuple{1,1,\alpha,\alpha}$, so that it suffices to show
  that $1\alpha1\subseteq g(\alpha\lang{\sum_{\beta\in\At}\beta}\alpha)$.
  We chose $\beta=\alpha$ in the sum: $\alpha\alpha\alpha\in \alpha\lang{\sum_{\beta\in\At}\beta}\alpha$, and we check that $\alpha\leftsquigarrow_\hcontr\alpha\alpha\leftsquigarrow_\hcontr\alpha\alpha\alpha$ (so that $\alpha\in\hcontr(\hcontr(\alpha\lang{\sum_{\beta\in\At}\beta}\alpha)))$.
\end{proof}

\begin{lem}
  \label{lem:kao:com:end}
  We have $\cl{(\hatom_{1,2}\cup \hcontr)}=\cl{\hatom_2}\circ\cl\hcontr\circ\cl{\hatom_1}$.
\end{lem}
\begin{proof}
  We use \cref{prop:gsos} with $\hatom_1$, $\hcontr$, and $\hatom_2$. 
  The right-hand sides of $\hcontr$ and $\hatom_1$ are words so that their closures are affine (\cref{lem:hyp:add}). This is not the case for $\hatom_2$ which is fine because it is placed in last position. 
  We finally have to show the commutation properties:
  \begin{itemize}
  \item there is nothing to show for the first set $(\hatom_1)$;
  \item for the second set $(\hcontr)$, we only have to bound $\hatom_1\circ\hcontr$, and we may choose any of the two alternatives ($(1)~\cl{\hcontr}\hatom_1^=$ or $(2)~\hcontr^=\cl{\hatom_1}$) since \cref{lem:kao:com}(i) gives a bound ($\hcontr\circ\hatom_1$) which is below both of them;
  \item for the third set $(\hatom_2)$, we have to bound both $\hatom_1\circ\hatom_2$ and $\hcontr\circ\hatom_2$ by either $(1)~\cl{\hatom_2}(\hatom_1\cup\hcontr)^=$ or $(2)~\hatom_2^=\cl{(\hatom_1\cup\hcontr)}$.
    We have that $\hatom_1\circ\hatom_2$ is below both by \cref{lem:kao:com}(ii); in contrast $\hcontr\circ\hatom_2$ is only (known to be) below the latter, by \cref{lem:kao:com}(iii);
    therefore we chose the latter option. \qedhere
  \end{itemize}
\end{proof}
Unfortunately, we cannot directly apply \cref{prop:cup} with this decomposition, because we do not have a reduction for $\hatom_2$ alone, but only one for $\hatom_{1,2}$ (cf.\ \cref{rem:wrong:red}).
We can nevertheless conclude just by assembling the results collected so far:
\begin{thm}
  For all $e,f\in\termskat$, if $\sem\hkabo e=\sem\hkabo f$ then $\KAproves[\hkabo] e=f$\,.
\end{thm}
\begin{proof}
  From \cref{lem:kao:com:end} we deduce
  \begin{align*}
    \cl{(\hatom_{1,2}\cup \hcontr\cup \hatom_1)}
    = \cl{(\hatom_2\cup \hcontr\cup \hatom_1)}
    = \cl{\hatom_2}\circ\cl\hcontr\circ\cl{\hatom_1}
    \subseteq\cl{\hatom_{1,2}}\circ\cl\hcontr\circ\cl{\hatom_1}
  \end{align*}
  (This is actually an equality since the right hand-side is contained in the left hand-side.) 
  Thus we can apply \cref{prop:cup} to $\hatom_1$, $\hcontr$, and $\hatom_{1,2}$, for which we have reductions by \cref{lem:reds}.
\end{proof}

The commutation properties used to combine the three sets of hypotheses $\hcontr$, $\hatom_1$ and $\hatom_2$ in the above proof can be presented in a table which makes it easier to see that \cref{prop:gsos} indeed applies (or to experiment when looking for a proof). This is \cref{table:kabo}.

\begin{table}[t]
  \begin{align*}
    \begin{array}{crr@{\,}l|c|c|c}
      &&&
      &1&c&2 \\
      \hline
      (\hatom_1) & 1:&\alpha&\leq 1&-&.&. \\
      (\hcontr) & c:&\alpha&\leq \alpha\alpha&&-&cc \\
      (\hatom_2) & 2:&1&\leq \sum\alpha&&&- \\
      \hline
      &&&&&.&2^=({<}2)^\star\\
    \end{array}
  \end{align*}
  \caption{Summary of commutations for KABO.}
  \label{table:kabo}
\end{table}

There we use single letter shorthands for the three sets of hypotheses, which are recalled in the first column.
In an entry $(i,j)$, an expression $g$ means we have $ij\subseteq ji\cup g$ and a dot indicates there is no overlap between $i$ and $j$, so that $ij\subseteq ji$ by \cref{cor:nooverlap}.
The diagonal is filled with dashes only for the sake of readability.
The last line collects, for each column $j$, an upper bound we may deduce for $({<}j)j$; a dot indicates the strong bound $j^=({<}j)^=$ (in both cases, ${<}j$ denotes the union of all functions appearing strictly before $j$).

The idea is that it suffices to cover all entries strictly above the diagonal, in such a way that the collected bound for each column (but the first one) fits into one of the two alternatives provided by \cref{prop:gsos}.

Since dotted entries comply with both options, they pose no constraint: once they have been verified, they can be ignored. Other entries might be constraining: a column must comply to either (1) or (2) and cannot mix between them: for each column we have to make a common choice for all its rows.
In \cref{table:kabo} we see that entry $(c,2)$ forces us to use the second option for column 2: with the order $1<c<2$, the function $cc$ is below $2^=({<}2)^\star$ (not using $2^=$ and entering the starred expression twice), but not below $2^\star({<}2)^=$.

\section{KAT with a full element}
\label{sec:katf}

Zhang et al.~observed in~\cite{ZhangAG22} that extending KAT with a constant for the full relation makes it possible to model incorrectness logic~\cite{OHearn20}.
They left open the question of finding an axiomatisation which is complete w.r.t.\ relational models.
We obtained a partial answer in~\cite{pw:concur22:katop}: for plain Kleene algebras (without tests), we get a complete axiomatisation by adding the following two axioms, where $\full$ is the new constant.
\begin{align}
  \label{ax:T}\tag{T}
  x &\leq \full\\
  \label{ax:F}\tag{F}
  x &\leq x\cdot \full\cdot x
\end{align}
This completeness result was obtained in two steps: first we proved that with this new constant, the equational theory of relational models is characterised by languages closed under the above two axioms, and then we proved completeness of these axioms w.r.t.\ closed languages, by exhibiting a (non-trivial) reduction.
In this section we show how to use this reduction as a black box and combine it with our reductions for KAT, so as to obtain completeness of KAT with the above two axioms (KATF).

Note that we use the symbol $\full$ here, unlike in \cite{ZhangAG22,pw:concur22:katop} where the symbol $\top$ is used: this makes it possible to distinguish this constant from the greatest element of the Boolean algebra of tests, and to emphasise that it should be understood as a \emph{full} element rather than just a \emph{top} element: axiomatically, a top element should only satisfy \eqref{ax:T} (as was done in \cref{ex:katop}).

KATF terms ($\termskatf$) are regular expressions over the alphabet
$\Sigma_\hkatf\eqdef\Sigma\uplus\termsba\uplus\set\full$,
where $\termsba$ are the Boolean algebra terms as defined in \cref{sec:kat}.
From now on we will often use commas to denote unions of named sets of hypotheses.
As expected, we extend the set of hypotheses we used for KAT (\hkat) by setting
$\hkatf\eqdef\hkat,\hfull$, where
\begin{align*}
  \hfull\eqdef
  \set{e\leq \full\mid e\in\termskatf}\cup
  \set{e\leq e\cdot\full\cdot e\mid e\in\termskatf}
\end{align*}
The reduction $r$ from $\hkat$ to $\hatom$ we defined in \cref{sec:kat:complete} actually extends to a reduction  from $\hkatf$ to $\hatom,\hfull$.
\begin{lem}
  The homomorphism $r$ from \cref{sec:kat:complete} is a reduction from $\hkatf$ to $\hatom,\hfull$.
\end{lem}
\begin{proof}
  It suffices to adapt the proof of \cref{lem:kat:atom}.
  In turn, this essentially means extending \ref{lem:atom:r:kat} to deal with the two new axioms:
  \begin{itemize}
  \item for $e\leq\full$, we have $\proves[\hfull] r(e)\leq \full = r(\full)$;
  \item for $e\leq e\cdot\full\cdot e$, we have $\proves[\hfull] r(e)\leq r(e)\cdot\full\cdot r(e)=r(e\cdot\full\cdot e)$.
  \end{itemize}
  (In both cases using the fact that $r$ is a homomorphism defined such that $r(\full)=\full$: $r$ sees $\full$ as an arbitrary, non-interpreted, letter.)
\end{proof}
Let us now rephrase a central result of~\cite{pw:concur22:katop}:
\begin{thmC}[{\cite{pw:concur22:katop}}]
  \label{thm:kaf}
  There is a reduction from $\hfull$ to the empty set.
\end{thmC}
\begin{proof}
  This is not stated explicitly in that paper, but this is an implicit step in the proof of Theorem~4.16, which follows from Propositions 3.4, 4.7, and 4.14.
\end{proof}

As a consequence, since we have reductions for $\hatom_0$, $\hatom_1$ and $\hatom_{1,2}$ to the empty set, we can follow the same strategy as before: it suffices to find how to organise the closures $\hatom_0$, $\hatom_1$, $\hatom_{1,2}$ and $\hfull$, so that we can use \cref{prop:cup} and \cref{prop:gsos}.

Since they are quantified over expressions, the inequations in $\hfull$ are not easy to work with.
Fortunately, \cref{prop:reg:ctx} gives us that $\cl\hfull=\cl{(\htop,\hful)}$, where
\begin{align*}
  \htop\eqdef\set{w\leq \full\mid w\in\Sigma_\hkatf^*}&&
  \hful\eqdef\set{w\leq w\cdot\full\cdot w\mid w\in\Sigma_\hkatf^*}
\end{align*}
As a consequence, we can work with these simpler sets.
From now on we simply write 0, 1 and 2 for $\hatom_0$, $\hatom_1$ and $\hatom_2$, and we use the following ordering:
\begin{align*}
  0<1<\htop<\hful<2
\end{align*}
Accordingly, we study overlaps to build \cref{table:katf}.
We use the same conventions as for \cref{table:kabo}, with the addition that we use double dotted entries to denote that \cref{lem:dl:0} applies (in such a cell $(i,j)$, we thus have $ij\subseteq i$, which is not constraining.)
\begin{table}[t]
  \begin{align*}
    \begin{array}{lr@{\,}l|c|c|c|c|c}
      &&
      &0&1&\htop&\hful&2 \\
      \hline
      0:&\alpha\beta&\leq 0~(\alpha\neq\beta)&-&..&..&..&.. \\
      1:&\alpha&\leq 1&&-&\htop&\hful11&. \\
      \htop:&w&\leq \full&&&-&\hful\htop\htop&. \\
      \hful:&w&\leq w\full w&&&&-&2\hful1 \\
      2:&1&\leq \sum\alpha&&&&&- \\
      \hline
      &&&&.&.&\hful({<}\hful)^\star&2({<}2)^\star\\
    \end{array}
  \end{align*}
  \caption{Summary of commutations for KATF.}
  \label{table:katf}
\end{table}
More formally, we prove:
\begin{lem}
  \label{lem:katf:com}
  We have the following inclusions of functions:
  \begin{enumerate}[(i)]
  \item\label{it:kaft0s} $0s \subseteq 0$ for $s\in\set{1,\htop,\hful,2}$
  \item\label{it:kaft1t} $1\htop \subseteq \htop1\cup \htop$
  \item\label{it:kaft1f} $1\hful \subseteq \hful1\cup \hful11$
  \item\label{it:kafttf} $\htop\hful \subseteq \hful\htop\cup \hful\htop\htop$
  \item\label{it:kaft12} $12 \subseteq 21$
  \item\label{it:kaftt2} $\htop2 \subseteq 2\htop$
  \item\label{it:kaftf2} $\hful2 \subseteq 2\hful\cup 2\hful1$
  \end{enumerate}
\end{lem}
\begin{proof}
  As before, we use \cref{prop:overlaps}, except for (i) which is just \cref{lem:dl:0}.
  \begin{enumerate}[(i)]
    \setcounter{enumi}{1}
  \item $1\htop \subseteq \htop1\cup \htop$:
    overlaps are of the form $\tuple{u,v,1,1}$ which we solve as follows: %
    \closeoverlapi1\htop\htop{u\alpha v}{uv}{\full}
  \item $1\hful \subseteq \hful1\cup \hful11$:
    overlaps are of the form $\tuple{u,v,1,1}$ which we solve as follows: %
    \closeoverlapiii1\hful\hful11{u\alpha v}{uv}{uv\full uv}%
    {u\alpha v\full u\alpha v}{u\alpha v\full u v}
  \item $\htop\hful \subseteq \hful\htop\cup \hful\htop\htop$:
    overlaps are of the form $\tuple{u,w,1,1}$ which we solve as follows: %
    \closeoverlapiii\htop\hful\hful\htop\htop{uvw}{u\full w}{u\full w\full u\full w}%
    {uvw\full uvw}{uvw\full u\full w}
  \item $12 \subseteq 21$: there are no overlaps.
  \item $\htop2 \subseteq 2\htop$: there are no overlaps.
  \item $\hful2 \subseteq 2\hful\cup 2\hful1$: there are two kinds of overlaps:
    \begin{itemize}
    \item $\tuple{1,1,u,v\full uv}$ which we intuitively solve as follows: %
      \closeoverlapiii\hful22\hful1{uv}{uv\full uv}{u(\sum\alpha)v\full uv}%
      {u\beta v}{u\beta v\full u\beta v}%
      Here the diagram is a bit sloppy: there is an implicit universal quantification on the atom $\beta$ in the second line. Formally, we have to show $uv\in 2\hful1(u(\sum\alpha)v\full uv)$.
      To this end, we prove that for all $\beta$, $u\beta v\in \hful1(u(\sum\alpha)v\full uv)$.
      We fix such a $\beta$ and we prove
      $u\beta v\full u\beta v\in 1(u(\sum\alpha)v\full uv)$.
      This follows from
      $u\beta v\full uv\in u(\sum\alpha)v\full uv$, which holds by choosing $\alpha=\beta$ in the sum.
    \item $\tuple{1,1,uv\full u,v}$, which are handled symmetrically.
      \qedhere
    \end{itemize}
  \end{enumerate}
\end{proof}

Thanks to these partial commutation properties, we obtain:
\begin{lem}
  \label{lem:katf:fullcommutations}
  We have $\cl{(0,1,\hfull,2)}=\cl2\cl\hful\cl\htop\cl1\cl0=\cl{(1,2)}\cl\hfull\cl1\cl0$.
\end{lem}
\begin{proof}
  We have $\hfull=\htop,\hful$ as functions by \cref{prop:reg:ctx}. Then we obtain:
  \begin{align*}
     \cl{(0,1,\htop,\hful,2)}
    =\cl2\cl\hful\cl\htop\cl1\cl0
    \subseteq\cl{(1,2)}\cl{(\htop,\hful)}\cl1\cl0
    \subseteq\cl{(0,1,\htop,\hful,2)}
  \end{align*}
  The equality is a direct application of \cref{prop:gsos}: all requirements are provided in \cref{table:katf} (using alternative (2) for each column). The subsequent inclusions hold by general closure properties.
\end{proof}

\begin{thm}
  \label{thm:katf}
  The set of hypotheses $\hkatf$ reduces to the empty set; $\KA_\hkatf$ is complete and decidable.
\end{thm}
\noindent
(Note that we must apply \cref{prop:cup} to $0$, $1$, $\hfull$, and $(1,2)$ rather than to $0,1,\htop,\hful,2$: while there is a reduction to the empty set for \htop alone---cf. \cref{ex:katop}, we do not know if there is such a reduction for \hful alone; similarly for $1$ and $2$.)

\section{KAT with converse}
\label{sec:katc}

Another important extension of Kleene algebra is Kleene algebra with converse (KA$^\circ$), where a unary operation of \emph{converse} is added $(\cdot^\circ)$, to represent transposition of relations, or language reversal.
One obtains a complete axiomatisation w.r.t.\ language models fairly easily: it suffices to state that converse is an involution that distribute over sums and products, reversing the arguments of the latter:
\begin{align}
  \tag{I}
  \label{ax:I}
  x^{\circ\circ}=x && (x+y)^\circ=x^\circ+y^\circ && (x\cdot y)^\circ=y^\circ\cdot x^\circ
\end{align}
(That $0^\circ=0$, $1^\circ=1$, and $x^{{*}{\circ}}=x^{{\circ}{*}}$ follow.)
Like with the constant $\full$, dealing with relational models requires more work because converse in relational models satisfies more laws than in language models. In particular, it satisfies the law
\begin{align}
  \tag{C}
  \label{ax:C}
  x \leq x\cdot x^\circ\cdot x
\end{align}
Indeed, if a relation $x$ contains a pair $\tuple{i,j}$, then this pair also belongs to the relation $x\cdot x^\circ\cdot x$, by going from $i$ to $j$, then back to $i$ via $x^\circ$, and then to $j$ again.

Bloom, Bernátsky, Ésik, and Stefanescu have shown that adding this axiom actually yields a complete axiomatisation w.r.t.\ relational models~\cite{BES95,EB95}.
We show in this section that we can integrate these axioms for converse together with the axioms for KAT, and prove completeness of KAT with converse (KAT$^\circ$).

\medskip

First of all, we need to get rid of the converse operation: our framework deals only with plain regular expressions. To this end, we use the laws~\eqref{ax:I} and their consequences mentioned below as left-to-right rewrite rules: this makes it possible to normalise every regular expression with converse into a regular expression on a duplicated alphabet.

To this end, let us setup some conventions on sets of shape $2X\eqdef X+X$ for some set $X$:
\begin{itemize}
\item we let bold letters range over their elements;
\item we write $\rho$ for the codiagonal surjection from $2X$ to $X$;
\item we write $\inl,\inr$ for the two natural injections of $X$ into $2X$;
\item for $x\in X$, we simply write $x$ for $\inl(x)$; in contrast, we write $x^\bullet$ for $\inr(x)$;
\item we define an involution $\cdot^\circ$ on $2X$ by letting $x^\circ\eqdef x^\bullet$ and $x^{\bullet\circ}\eqdef x$ for all $x\in X$;
\end{itemize}

We extend the above involution $\cdot^\circ$ on $2X$ to an involution on regular expressions on the alphabet $2X$: $e^\circ$ is obtained from $e$ by applying $\cdot^\circ$ on all letter leaves, and swapping the arguments of all products. (E.g., $(a\cdot b^*+c^\bullet)^\circ=b^{{\bullet}{*}}\cdot a^\bullet+c$.).
This involution restricts to an involution on words over $2X$. For instance, we have $(ab^\bullet c)^\circ=c^\bullet ba^\bullet$.

Bloom, Bernátsky, Ésik, and Stefanescu first reduce the problem to a problem on (converse-free) regular expressions over the duplicated alphabet, and then they (implicitly) provide a reduction from~\eqref{ax:C} to the empty set (axioms~\eqref{ax:I} being dealt with by normalisation).
To do so, we need to translate this axiom into a set of hypotheses over the duplicated alphabet:
\begin{align*}
  \hcnv_X \eqdef \set{e \leq e\cdot e^\circ\cdot e \mid e \in \termska(2X)}
\end{align*}

\begin{thmC}[{\cite{EB95}}]
  \label{thm:kac}
  For all alphabets $X$, there is a reduction from $\hcnv_X$ to the empty set.
\end{thmC}

Now let us consider the combination of tests (from KAT) and converse, to obtain KAT$^\circ$.
There we also need to specify the behaviour of converse on tests $(\phi)$: we need the following law, which, unlike $1^\circ=1$ and $0^\circ=0$, is not derivable from~\eqref{ax:I}.
\begin{align}
  \tag{I'}
  \label{ax:II}
  \phi^\circ = \phi
\end{align}
Concerning terms, we move to regular expressions over the duplicated alphabet $2(\Sigma+\termsba)$.
According to the above axiom, we define:
\begin{align*}
  \hi_X &\eqdef \set{\pphi^\circ \leq \pphi \mid \pphi \in 2X}\\
  \hkatc &\eqdef \inl(\hkat) \cup \hcnv_{\Sigma+\termsba} \cup \hi_\termsba
\end{align*}
(Where $\pphi^\circ$ and $\pphi$ in the definition of $\hi_X$ are seen as one letter words over the alphabet $2(\Sigma+X)$, and where $\inl(\hkat)$ denotes the injection via $\inl$ of all inequations in $\hkat$ (cf. \cref{sec:kat}), which are inequations between terms in $\termska(\Sigma+\termsba)$, into inequations between terms in $\termska(2(\Sigma+\termsba))$.

Recall the reduction $r\colon \termska(\Sigma + \termsba) \to \termska(\Sigma+\At)$ from $\hkat$ to $\hatom$ we defined in \cref{sec:kat:complete},
and let $s\colon \termska(2(\Sigma + \termsba)) \to \termska(2(\Sigma+\At))$ be the homomorphism defined by
\begin{align*}
  \begin{cases}
    s(\mathbf a)\eqdef \mathbf a & \mathbf a \in 2\Sigma\\
    s(\mathbf\phi)\eqdef s(\phi^\bullet)\eqdef r(\phi)+r(\phi)^\circ & \phi \in \termsba
  \end{cases}
\end{align*}
Following \cref{sec:kat:complete}, we define the following sets of inequations over the alphabet $2(\Sigma+\At)$:
\begin{itemize}
\item $0\eqdef\set{\alpha\beta\leq0\mid\alpha,\beta\in\At,~\alpha\neq\beta}$
\item $1\eqdef\set{\aalpha\leq1\mid\aalpha\in2\At}$
\item $2\eqdef\set{1\leq \sum_{\alpha\in\At}\alpha}$
\end{itemize}
Note that $0=\inl(\hatom_0)$ and $2=\inl(\hatom_2)$: these sets only deal with atoms on the left; in contrast, 1 is extended to deal with atoms on the right: $1=\inl(\hatom_1),\inr(\hatom_1)$. We also set $1'\eqdef\inl(\hatom_1)$, so that $1',2$ corresponds to $\hatom_{1,2}$.

\begin{lem}
  The homomorphism $s$ is a reduction from $\hkatc$ to $0,1,2,\hcnv_{\Sigma+\At},\hi_\At$.
\end{lem}
\begin{proof}
  We adapt the proof of \cref{lem:kat:atom}.
  First observe that $\proves[\hi_\At] s(\phi)=r(\phi)$ for all $\phi\in\termsba$.
  This makes it possible to reuse \cref{lem:atom:r:kat}, which we only need to extend to deal with the two new axioms:
  \begin{itemize}
  \item for $e\leq e\cdot e^\circ\cdot e$,
    we first prove by induction on $e$ that $\proves[] s(e)^\circ=s(e^\circ)$, using that $s$ is a homomorphism, and commutativity of sum for the base case of a formula $\pphi\in2\termsba$;
    then we deduce $\proves[\hcnv_{\Sigma+\At},\hi_\At] s(e)\leq s(e)\cdot s(e)^\circ\cdot s(e)=s(e)\cdot s(e^\circ)\cdot s(e)=s(e\cdot e^\circ \cdot e)$.
  \item for $\pphi^\circ\leq\pphi$ with $\pphi\in 2\termsba$, we have $s(\pphi^\circ)=s(\pphi)$, syntactically.
  \end{itemize}
  For $\mathbf a\in 2\Sigma$, we still have $\mathbf a=s(\mathbf a)$ syntactically.
  We no longer have such a strong property for $\aalpha\in 2\At$, for which  $s(\aalpha)=r(\alpha)+r(\alpha)^\circ=\alpha+\alpha^\bullet$ with $\alpha=\rho(\aalpha)$; but then $\proves \aalpha\leq s(\aalpha)$, and Condition 3 is satisfied.
  It remains to check Condition 2 for $\pphi\in 2\termsba$; thanks to the first observation, we have $\proves[\hi_\At] s(\pphi)=r(\phi)$ with $\phi=\rho(\pphi)$, and we can use the fact that $r$ satisfies Condition 2: $\proves[\hkat]\phi=r(\phi)$, so that $\proves[\hkatc]\pphi=\phi=r(\phi)=s(\pphi)$.
\end{proof}
We have reductions to the empty set for each set in $0,1,(1',2),\hcnv_{\Sigma+\At},\hi_\At$: for
 $\hcnv_{\Sigma+\At}$ this is \cref{thm:kac}, and the rest follows by \cref{lem:reds}.
Therefore, it suffices to find how to organise their closures.
Like for KATF, we first simplify $\hcnv_{\Sigma+\At}$ into the following set, via \cref{prop:reg:ctx}:
\begin{align*}
  \hc &\eqdef \set{w \leq w\cdot w^\circ\cdot w \mid w \in (2(\Sigma+\termsba))^*}
\end{align*}
We also write $\hi$ for $\hi_\At$, and we use the following ordering:
\begin{align*}
  0<1<\hi<\hc<2
\end{align*}
Accordingly, we study overlaps to build \cref{table:katc}.
\begin{table}[t]
  \begin{align*}
    \begin{array}{lr@{\,}l|c|c|c|c|c}
      &&
      &0&1&\hi&\hc&2 \\
      \hline
      0:&\alpha\beta&\leq 0~(\alpha\neq\beta)&-&..&..&..&.. \\
      1:&\aalpha&\leq 1&&-&.&c111&. \\
      \hi:&\aalpha^\circ&\leq \aalpha&&&-&\hc\hi\hi\hi&. \\
      \hc:&w&\leq ww^\circ w&&&&-&2\hc\hi^=11 \\
      2:&1&\leq \sum\alpha&&&&&- \\
      \hline
      &&&&.&.&\hc({<}\hc)^\star&2({<}2)^\star\\
    \end{array}
  \end{align*}
  \caption{Summary of commutations for KAT$^\circ$.}
  \label{table:katc}
\end{table}
More formally, we prove:
\begin{lem}
  \label{lem:katc:com}
  We have the following inclusions of functions:
  \begin{enumerate}[(i)]
  \item\label{it:kafc1i} $1\hi \subseteq \hi1$
  \item\label{it:kafc1c} $1\hc \subseteq \hc1\cup \hc111$
  \item\label{it:kafcic} $\hi\hc \subseteq \hc\hi\cup \hc\hi\hi\hi$
  \item\label{it:kafc12} $12 \subseteq 21$
  \item\label{it:kafci2} $\hi2 \subseteq 2\hi$
  \item\label{it:kafcc2} $\hc2 \subseteq 2\hc\cup 2\hc\hi^=11$
  \end{enumerate}
\end{lem}
\begin{proof}
  As before, we use \cref{prop:overlaps}.
  \begin{enumerate}[(i)]
  \item there are no overlaps.
  \item $1\hc \subseteq \hc1\cup \hc111$:
    overlaps are of the form $\tuple{u,v,1,1}$ which we solve as follows: %
    \closeoverlapiiii1\hc\hc111{u\aalpha v}{uv}{uvv^\circ u^\circ uv}%
    {u\aalpha vv^\circ \aalpha^\circ u^\circ u\aalpha v}%
    {u\aalpha vv^\circ \aalpha^\circ u^\circ u v}%
    {u\aalpha vv^\circ u^\circ u v}
    (Note that this case forces us to define 1 so that it is able to produce both kinds of atoms.)
  \item $\hi\hc \subseteq \hc\hi\cup \hc\hi\hi\hi$:
    overlaps are of the form $\tuple{u,v,1,1}$ which we solve as follows: %
    \closeoverlapiiii\hi\hc\hc\hi\hi\hi{u\aalpha^\circ v}{u\aalpha v}%
    {u\aalpha v v^\circ\aalpha^\circ u^\circ u\aalpha v}%
    {u\aalpha^\circ v v^\circ\aalpha u^\circ u\aalpha^\circ v}%
    {u\aalpha^\circ v v^\circ\aalpha u^\circ u\aalpha v}%
    {u\aalpha^\circ v v^\circ\aalpha^\circ u^\circ u\aalpha v}%
  \item $12 \subseteq 21$: there are no overlaps.
  \item $\hi2 \subseteq 2\hi$: there are no overlaps.
  \item $\hc2 \subseteq 2\hc\cup 2\hc\hi^=11$: there are three kinds of overlaps:
    \begin{itemize}
    \item $\tuple{1,1,u,vv^\circ u^\circ uv}$ which we intuitively solve as follows: %
      \closeoverlapiiii\hc22\hc11{uv}{uvv^\circ u^\circ uv}{u(\sum\alpha)vv^\circ u^\circ uv}%
      {u\beta v}%
      {u\beta v v^\circ\beta^\bullet u^\circ u\beta v}%
      {u\beta v v^\circ u^\circ u \beta v}%
      (Using the same implicit universal quantification on $\beta$ in the bottom line, like in the proof of \cref{lem:katf:com}---restricted to atoms on the left.)
    \item $\tuple{1,1,uvv^\circ u^\circ u,v}$, which is handled symmetrically.
    \item $\tuple{1,1,uvv^\circ,u^\circ uv}$, for which we have
      \closeoverlapiiii\hc22\hc\hi{11}{uv}{uvv^\circ u^\circ uv}{uvv^\circ (\sum\alpha)u^\circ uv}%
      {u\beta v}%
      {u\beta v v^\circ\beta^\bullet u^\circ u\beta v}%
      {u\beta v v^\circ \beta u^\circ u\beta v}%
      (Again with an implicit universal quantification on left atom $\beta$ in the bottom line.)
      Note that we need one more step than in the previous cases, to turn the left $\beta$ produced from the sum into a $\beta^\bullet$.
      \qedhere
    \end{itemize}
  \end{enumerate}
\end{proof}

Thanks to these partial commutation properties, we obtain:
\begin{lem}
  \label{lem:katc:fullcommutations}
  We have $\cl{(0,1,\hi,\hc,2)}=\cl2\cl\hc\cl\hi\cl1\cl0=\cl{(1',2)}\cl\hc\cl\hi\cl1\cl0$.
\end{lem}
\begin{proof}
  The first equality is a direct application of \cref{prop:gsos}: all requirements are provided in \cref{table:katc} (using alternative (2) for all columns). The second one follows by general closure properties.
\end{proof}

Recall that $(1',2)$ corresponds to $\hatom_{1,2}$, for which we have a reduction to the empty set (like for $0$, $1$, $\hi$ and $\hc$). Therefore we can conclude: 
\begin{thm}
  \label{thm:katc}
  $\hkatc$ reduces to the empty set; $\KA_\hkatc$ is complete and decidable.
\end{thm}

\section{Kleene algebra with positive tests}
\label{sec:kapt}

In KAT, tests are assumed to form a Boolean algebra. This is sometimes
too strong; for instance, this prevents using the Coq tactic for
KAT~\cite{pous:itp13:ra} in situations where tests are only
intuitionistic propositions.
Here we study the structure obtained by assuming that they only form a distributive
lattice. A \emph{Kleene algebra with positive tests (KAPT)} is a
Kleene algebra $K$ containing a lattice $L$ such that the meet of $L$
coincides with the product of $K$, the join of $L$ coincides with the
sum of $K$, and all elements of $L$ are below the multiplicative identity of $K$.
(We discuss the variant where we have a bounded lattice at the
end, see \cref{rem:kaptt}).  Since the product distributes over sums in
$K$, $L$ must be a distributive lattice.  Also note that
there might be elements of $K$ below $1$ that do not belong to $L$.

As before, we fix two finite sets $\Sigma$ and $\Omega$ of primitive
actions and primitive tests. We consider regular expressions over
the alphabet $\Sigma+\termsdl$, where $\termsdl$ is the set of lattice
expressions over $\Omega$: expressions built from elements of $\Omega$
and two binary connectives $\vee$, $\wedge$.

We write $\hdl$ for the set of all instances of distributive lattice
axioms over $\termsdl$~\cite{DaveyPriestley90}, and we set
$\hkapt\eqdef\hdl\cup\hgluedl$ where
\begin{align*}
  \hgluedl &\eqdef \set{\phi\wedge  \psi = \phi\cdot \psi \text{, }\phi\vee  \psi = \phi+\psi\mid \phi,\psi\in\termsdl }
  \cup \set{\phi\leq 1\mid \phi\in\termsdl}
\end{align*}

Like for Boolean algebras, the free distributive lattice over $\Omega$ is finite and can be described easily. An \emph{atom} $\alpha$ is a non-empty subset of $\Omega$, and we write $\At$ for the set of such atoms as before. However, while an atom $\set{a,b}$ of Boolean algebra was implicitly interpreted as the term $a\wedge b\wedge \lnot c$ (when $\Omega=\set{a,b,c}$), the same atom in the context of distributive lattices is implicitly interpreted as the term $a\wedge b$---there are no negative literals in distributive lattices.
Again similarly to the case of Boolean algebras, the key property for atoms in distributive lattices is the following: for all atoms $\alpha$ and formulas $\phi$, we have
\begin{align*}
  \alpha\models\phi ~\Leftrightarrow~ \DLproves\alpha\leq\phi \qquad\text{and}\qquad
  \DLproves\phi = \bigvee_{\alpha\models \phi}\alpha
\end{align*}
Like for KAT, such a property makes it possible to reduce $\hkapt$ to the following set of equations on the alphabet $\Sigma+\At$.
\begin{align*}
  \hatomdl \eqdef \set{\alpha\cdot \beta = \alpha\cup\beta\mid \alpha,\beta\in\At }
  \cup \set{\alpha\leq 1\mid \alpha\in\At}
\end{align*}
(Note that in the right-hand side of the first equation, $\alpha\cup\beta$ is a single atom, whose implicit interpretation is $\alpha\wedge\beta$.)

\begin{lem}
  \label{lem:kapt:atomdl}
  There is a reduction from $\hkapt$ to $\hatomdl$, witnessed by the homomorphism
  $r \colon \termska(\Sigma+\termsdl) \to \termska(\Sigma+\At)$ defined by
  \begin{align*}
    r(x) &=
           \begin{cases}
             a & x = a \in \Sigma\\
             \sum_{\alpha \models \phi} \alpha & x = \phi \in \termsdl
           \end{cases}
  \end{align*}
\end{lem}

As a consequence, in order to get decidability and completeness for KAPT (i.e., $\hkapt$), it suffices to reduce $\hatomdl$ to the empty set. Let us number the three kinds of inequations that appear in this set:
\begin{align*}
1&{\eqdef}\set{\alpha{\cup}\beta \leq \alpha{\cdot}\beta \mid \alpha,\beta\in\At}&
2&{\eqdef}\set{\alpha{\cdot}\beta \leq \alpha{\cup}\beta \mid \alpha,\beta\in\At}&
4&{\eqdef}\set{\alpha \leq 1 \mid \alpha\in\At}
\end{align*}
We number the third set with 4 by anticipation: we will need the number 3 for another set of hypotheses later.
\cref{lem:reds}\ref{it:red:aw} gives reductions to the empty set for 1 and 4, but so far we have no reduction for 2.
We actually do not know if there is a reduction from 2 to the empty set. Instead, we establish a reduction from 2 together with 4 to 4 alone.
\begin{lem}
  \label{lem:kapt:24:4}
  There is a reduction from 2,4 to 4, witnessed by the homomorphism
  $r \colon \termska(\Sigma+\At)\to \termska(\Sigma+\At)$ defined by
  \begin{align*}
    r(x) &=
           \begin{cases}
             a & x = a \in \Sigma\\
             \sum
                   \set{\alpha_1\cdot\ldots\cdot\alpha_n\mid \alpha=\bigcup_{i\leq n}\alpha_i,~i\neq j \Rightarrow \alpha_i\neq \alpha_j} & x = \alpha \in \At
           \end{cases}
  \end{align*}
\end{lem}
\begin{proof}
  We use \cref{prop:hred}. The first condition is trivially satisfied since $2\cup4$ contains $4$.
  For a letter $\ltr a\in \Sigma$, $r(\ltr a)=\ltr a$ so that second and third conditions are trivial for such letters, and we need to prove them only for atoms $\alpha\in\At$.
  $\proves \alpha\leq r(\alpha)$ follows by using the singleton sequence $\alpha$ which is a term in the sum $r(\alpha)$. For the other condition, it thus suffices to show $\proves[2,4] r(\alpha)\leq \alpha$, i.e., $\proves[2,4] \alpha_1\dots\alpha_n\leq \alpha$ for all sequences $\alpha_1,\dots,\alpha_n$ of pairwise distinct atoms whose union is $\alpha$. This follows by $n-1$ successive applications of inequations in 2. It remains to check the last condition of \cref{prop:hred}; we consider the two kinds of equations separately:
  \begin{itemize}
  \item $\alpha\cdot\beta\leq\alpha\cup\beta$: we have to derive
    $\proves[4]r(\alpha)\cdot r(\beta)\leq r(\alpha\cup\beta)$. By distributivity, this amounts to proving
    $\alpha_1\dots\alpha_n\beta_1\dots\beta_m\leq r(\alpha\cup\beta)$ for all sequences of pairwise distinct atoms $\alpha_1,\dots,\alpha_n$ and $\beta_1,\dots,\beta_m$ whose unions are $\alpha$ and $\beta$, respectively. The sequence $\alpha_1,\dots,\alpha_n,\beta_1,\dots,\beta_m$ almost yields a term in $r(\alpha\cup\beta)$: its union is $\alpha\cup\beta$, but it may contain duplicate entries.
    We simply remove such duplicates using inequations in 4 ($\gamma\leq 1$).
  \item $\alpha\leq 1$: we have to show $\proves[4]r(\alpha)\leq 1$. This follows by repeated applications of inequations in 4: all terms of the sum $r(\alpha)$ are below the identity.\qedhere
  \end{itemize}
\end{proof}
\noindent (Note that the above reduction requires 4 in its target, and cannot be extended directly into a reduction from 1,2,4 to 4: $r(\alpha\cup\beta)\leq r(\alpha\cdot\beta)$ cannot be proved from 4---take $\alpha=\set a$, $\beta=\set b$, then $ba$ is a term in $r(\alpha\cup\beta)$ which is not provably below $ab=r(\alpha\cdot\beta)$.)

Composed with the existing reduction from 4 to the empty set (\cref{lem:reds}\ref{it:red:aw}), we thus have a reduction from 2,4 to the empty set.
It remains to combine this reduction to the one from 1 to the empty set (\cref{lem:reds}\ref{it:red:aw} again).
To this end, we would like to use \cref{prop:cup}, which simply requires us to prove that the closure $\cl\hatomdl=\cl{(1,2,4)}$ is equal either to $\cl1\cl{(2,4)}$ or to $\cl{(2,4)}\cl1$.
Unfortunately, this is not the case.
To see this, suppose we have two atomic tests $a$ and $b$.
For the first option, consider the singleton language $\set{ab}$ (a word consisting of two atoms); we have $ba\in\cl{(1,2,4)}\set{ab}$ (because $(a\wedge b)\in\cl1\set{ab})$, and then using $\cl2$) but $ba\not\in\cl1\cl{(2,4)}\set{ab}$.
For the second option, consider the singleton language $\set{a}$; we have $(a\wedge b)\in\cl{(1,2,4)}\set{a}$, because $ab\in\cl4\set{a}$, but $(a\wedge b)\not\in\cl{(2,4)}\cl1\set{a}$ because $\cl1\set{a}$ is just $\set{a}$, and $\cl{(2,4)}$ does not make it possible to forge conjunctions.

In order to circumvent this difficulty, we use a fourth family of equations:
\begin{align*}
3&\eqdef\set{\alpha\cup\beta \leq \alpha\mid\alpha,\beta\in\At}
\end{align*}
These axioms are immediate consequences of 1 and 4. Therefore, 1,2,4 reduces to 1,2,3,4.
Moreover they consist of `letter-letter' inequations, which are covered by \cref{lem:reds}\ref{it:red:aw}: 3 reduces to the empty set.
We shall further prove that $\cl{(1,2,3,4)}=\cl3\cl{(2,4)}\cl1$, so that \cref{prop:cup} applies to obtain a reduction from 1,2,3,4 to the empty set.

Let us recall the five sets of hypotheses defined so far:
\begin{itemize}
\item $1=\set{\alpha\cup\beta \leq \alpha\cdot\beta \mid \alpha,\beta\in\At}$,
\item $2=\set{\alpha\cdot\beta \leq \alpha\cup\beta \mid \alpha,\beta\in\At}$,
\item $3=\set{\alpha\cup\beta \leq \alpha\mid\alpha,\beta\in\At}$,
\item $4=\set{\alpha \leq 1 \mid \alpha\in\At}$.
\end{itemize}
We now prove the following partial commutations, which we summarise in \cref{table:kapt}.
In addition to the conventions used for the previous tables, we mark with triple dots those entries $(i,j)$ where there are overlaps but that we nevertheless have $ij\subseteq j^=i^=$ (a bound which is not constraining when we try to bound each column, nor if we try to reorder the lines and columns).
\begin{table}[t]
  \begin{align*}
    \begin{array}{lr@{\,}l|c|c|c|c||c}
      &&&1&2&3&4&5 \\
      \hline
      1:&\alpha\cup\beta&\leq \alpha\beta&-&...&...&3&11 \\
      2:&\alpha\beta&\leq \alpha\cup\beta&&-&332&44&. \\
      3:&\alpha\cup\beta&\leq \alpha&&&-&...&. \\
      4:&\alpha&\leq 1&&&.&-&. \\
      5:&1&\leq \emptyset&&&&&- \\
      \hline
      &&&&.&3^\star({<}3)^=&4^\star({<}4)^=&5^=({<}5)^\star\\
    \end{array}
  \end{align*}
  \caption{Summary of commutations for KAPT (with 5 for KABPT).}
  \label{table:kapt}
\end{table}

\begin{lem}
  \label{lem:kapt:commutations}
  We have the following inclusions of functions:
  \begin{enumerate}[(i)]
  \item\label{it:kapt12} $12 \subseteq 21\cup\id$
  \item\label{it:kapt13} $13 \subseteq 31$
  \item\label{it:kapt23} $23 \subseteq 32\cup332$
  \item\label{it:kapt14} $14 \subseteq 41\cup3$
  \item\label{it:kapt24} $24 \subseteq 42\cup44$
  \item\label{it:kapt34} $34 \subseteq 43\cup4$
  \item\label{it:kapt43} $43 \subseteq 34$
  \end{enumerate}
\end{lem}
\begin{proof}
  We use \cref{prop:overlaps} so that it suffices to analyse overlaps.
  \begin{enumerate}[(i)]
  \item $12 \subseteq 21\cup\id$: we must consider the overlaps of $\alpha\beta$ and $\gamma\delta$:
    \begin{itemize}
    \item when $\alpha\beta=\gamma\delta$, we have the full overlap $\tuple{1,1,1,1}$, and we check that
      \closeoverlap12{\alpha\cup\beta}{\alpha\beta=\gamma\delta}{\gamma\cup\delta}
    \item when $\beta=\gamma$, we have the overlap $\tuple{1,\delta,\alpha,1}$, and we check that
      \closeoverlapii1221{(\alpha\cup\beta)\delta}{\alpha\beta\delta=\alpha\gamma\delta}{\alpha(\gamma\cup\delta)}%
      {\alpha\cup\beta\cup\delta=\alpha\cup\gamma\cup\delta}
    \item when $\alpha=\delta$, we have the overlap $\tuple{\gamma,1,1,\beta}$, which is handled symmetrically.
    \end{itemize}
  \item $13 \subseteq 31$: we must consider the overlaps of $\alpha\beta$ and $\gamma\cup\delta$.
    \begin{itemize}
    \item when $\beta=\gamma\cup\delta$ we have the overlap $\tuple{1,1,\alpha,1}$, for which we have
      \closeoverlapii1331{\alpha\cup\gamma\cup\delta}{\alpha(\gamma\cup\delta)}{\alpha\gamma}%
      {\alpha\cup\gamma}
    \item when $\alpha=\gamma\cup\delta$ we have the overlap $\tuple{1,1,1,\beta}$, which is handled symmetrically.
    \end{itemize}
  \item $23 \subseteq 32\cup332$: we must consider the overlaps of $\alpha\cup\beta$ and $\gamma\cup\delta$:
    there is only the full overlap, when $\alpha\cup\beta=\gamma\cup\delta$.
    We have
    \closeoverlapiii23332{\alpha\beta}{\alpha\cup\beta=\gamma\cup\delta}{\gamma}%
    {\alpha(\beta\cap\gamma)}{(\alpha\cap\gamma)(\beta\cap\gamma)}
    (observing that $(\alpha\cap\gamma)\cup(\beta\cap\gamma)=\gamma$ for the first step, since $\alpha\cup\beta=\gamma\cup\delta$)
  \item $14 \subseteq 41\cup3$: we must consider the overlaps of $\alpha\beta$ and $\gamma$.
    \begin{itemize}
    \item if $\gamma=\beta$ then we have the overlap $\tuple{1,1,\alpha,1}$ for which we have
      \closeoverlapi143{\alpha\cup\beta}{\alpha\beta}{\alpha}
    \item the overlap when $\gamma=\alpha$ is handled symmetrically
    \end{itemize}
  \item $24 \subseteq 42\cup44$: 
    there is only the full overlap, for which we have
    \closeoverlapii2444{\alpha\beta}{\alpha\cup\beta}{1}{\alpha}
  \item $34 \subseteq 43\cup4$: 
    there is only the full overlap, for which we have
    \closeoverlapi344{\alpha\cup\beta}{\alpha}{1}
  \item $43 \subseteq 34$: there are no overlaps between $1$ and $\alpha\cup\beta$.
    \qedhere
  \end{enumerate}
\end{proof}

Thanks to these partial commutation properties, we obtain:
\begin{lem}
  \label{lem:kapt:fullcommutations}
  We have $\cl{(1,2,3,4)}=\cl4\cl3\cl2\cl1=\cl3\cl{(2,4)}\cl1$.
\end{lem}
\begin{proof}
  We have
  \begin{align*}
    \cl{(1,2,3,4)}=\cl4\cl3\cl2\cl1
                  \subseteq\cl3\cl4\cl2\cl1
                  \subseteq\cl3\cl{(2,4)}\cl1
                  \subseteq\cl{(1,2,3,4)}
  \end{align*}
  The first equality is a direct application of \cref{prop:gsos}: all requirements are provided in \cref{table:kapt} (using alternative (1) for each column).
  The subsequent inclusion comes from \cref{lem:gsos:2} applied to 4 and 3, thanks to \cref{lem:kapt:fullcommutations}\ref{it:kapt43}.
  The remaining inclusions follow from basic properties of closures.
\end{proof}

\begin{rem}
  We actually have $\cl{(3,4)}=\cl3\cl4=\cl4\cl3$: \cref{lem:gsos:2} can be applied in both directions with 3 and 4, thanks to items \ref{it:kapt34} and \ref{it:kapt43} in \cref{lem:kapt:fullcommutations}.
  This can be read directly on \cref{table:kapt}: when restricted to lines and columns 3,4, we only get dotted entries.

  However, we cannot place 4 before 3: the occurrence of 3 in entry (1,4) requires 3 to appear before 4 in order to validate column 4 when we apply \cref{prop:gsos} to 1,2,3,4.

  We did proceed differently in the conference version of this article~\cite{prw:ramics21:mkah}, where 3 and 4 were swapped and where we did not have \cref{prop:gsos}. This required us to iterate and combine the partial commutation properties from \cref{lem:kapt:fullcommutations} manually, in a non-trivial way. We prefer the present proof, which is slightly more automatic.
\end{rem}

\begin{thm}
  \label{thm:kapt}
  $\hkapt$ reduces to the empty set; $\KA_\hkapt$ is complete and decidable.
\end{thm}
\begin{proof}
  $\hkapt$ reduces to $\hatomdl$ by \cref{lem:kapt:atomdl}, which in turn reduces to $1,2,3,4$ by \cref{cor:red:id}. We see the latter as being composed of three sets of hypotheses, $1$, $3$, and $2,4$. All three of them reduce to the empty set: the first two by \cref{lem:reds}\ref{it:red:aw}, and the third one by \cref{lem:kapt:24:4} and \cref{lem:reds}\ref{it:red:aw} again. These three reductions can be composed together by \cref{prop:cup} and \cref{lem:kapt:fullcommutations}.
\end{proof}

\begin{rem}
  \label{rem:kaptt}
  The case of Kleene algebras containing a \emph{bounded} distributive lattice, with extremal elements $\bot$ and $\top$ coinciding with $0$ and $1$, may be obtained as follows.
  Allow the empty atom $\emptyset$ in $\At$ (interpreted as $\top$), and add the inequation $5\eqdef\set{1\leq\emptyset}$ to $\hatomdl$.
  \cref{lem:kapt:atomdl} extends easily, and we have a reduction from 5 to the empty set (\cref{lem:reds}\ref{it:red:aw}). Therefore it suffices to find how to combine $\cl5$ with the other closures. We do so in the lemma below, so that we can conclude that the equational theory of Kleene algebras with a bounded distributive lattice is complete and decidable.
  \qedhere
\end{rem}

\begin{lem}
  \label{lem:kapttcom}
  We have $\cl{(1,2,3,4,5)}=\cl5\cl4\cl3\cl2\cl1=\cl5\cl{(1,2,3,4)}$.
\end{lem}
\begin{proof}
  It suffices to prove the first equality: the second one follows from \cref{lem:kapt:fullcommutations}.
  To this end, it suffices to show that $5$ somehow commutes over the four other functions.
  We do so by completing \cref{table:kapt} and applying \cref{prop:gsos}.
  Since the LHS of 5 is the empty word, there are very few overlaps to consider: the only one if for bounding 15, where $\alpha\beta$ and the empty word overlap via $\tuple{1,1,\alpha,\beta}$.
  In that case, we have %
  \closeoverlapii1511{\alpha\cup\beta}{\alpha\beta}{\alpha\emptyset\beta}%
  {\alpha\beta=\alpha(\emptyset\cup\beta)}%
  whence the 11 entry at position $(1,5)$ in \cref{table:kapt}, thanks to \cref{prop:overlaps}.
  It is now easy to check that \cref{prop:gsos} applies to 1,2,3,4,5, using alternative (2) for the fifth column.
\end{proof}

\section{NetKAT}
\label{sec:netkat}

NetKAT is a framework for analysing network programs~\cite{AndersonFGJKSW14}. It is a variant of KAT, adding an explicit alphabet
of variable assignments (of the form $x \leftarrow n$) and tests of the form $x=n$, where $n$ ranges over a finite domain of values.
Here we work with \emph{reduced} NetKAT (equivalent to NetKAT), where tests and assignments are replaced with more general complete assignments and tests. These complete tests are the atoms of a Boolean algebra.
NetKAT in its reduced form is also treated as one of the motivating examples of Mamouras and Kozen's theory of Kleene algebra
with equations~\cite{KozenM14}.

Accordingly, we fix a finite set $A$ and we work with the alphabet $A+P+\set\dup$ where $P\eqdef A$.
We let $\alpha,\beta$ range over $A$ and we call them \emph{atoms}.
We let $p,q$ range over $P$ and we call them \emph{assignments}.
Given an atom $\alpha$, we write $p_\alpha$ for the corresponding assignment in $P$.
Given an assignment $p$, we write $\alpha_p$ for the corresponding atom in $A$.

Following~\cite[Figure~6]{AndersonFGJKSW14}, we define $\hnetkat$ as the following collection of equations:
\begin{mathpar}
  pq=q \and
  p_\alpha \alpha = p_\alpha \and
  \alpha p_\alpha = \alpha \and
  \alpha\dup=\dup\alpha \and
  \textstyle\sum_{\alpha\in A}\alpha = 1\and
  \alpha\beta=0~(\alpha\neq\beta) 
\end{mathpar}
These equations are equivalent to $\hnetkat'\eqdef0,1,2,3,4,5,6,7$, where
\begin{itemize}
\item $0\eqdef\set{\alpha\beta,\alpha\dup\beta,p_\alpha\beta\leq0 \mid \alpha,\beta\in A,~\alpha\neq\beta}$
\item $1\eqdef\set{\alpha p_\alpha\leq 1\mid \alpha\in A}$
\item $2\eqdef\set{pq \leq q\mid p,q\in P}$
\item $3\eqdef\set{\alpha \leq 1\mid\alpha\in A}$
\item $4\eqdef\set{p_\alpha\leq p_\alpha \alpha\mid\alpha\in A}$
\item $5\eqdef\set{q\leq pq\mid p,q\in P}$
\item $6\eqdef\set{\alpha\leq \alpha p_\alpha\mid\alpha\in A}$
\item $7\eqdef\set{1\leq \sum_{\alpha\in A}\alpha}$
\end{itemize}
Note that 3 are 4 are redundant: they follows respectively from 1,6 and 0,7. We include them on purpose: they help getting appropriate partial commutations.

\begin{lem}
  We have $\proves[\hnetkat']\hnetkat$ and $\proves[\hnetkat]\hnetkat'$.
\end{lem}
\begin{proof}
  First observe that:
  \begin{itemize}
  \item $\proves[\hnetkat']pq=q$ using 2,5;
  \item $\proves[\hnetkat']p_\alpha\alpha=p_\alpha$, using 3,4;
  \item $\proves[\hnetkat']\alpha p_\alpha\leq \alpha p_\alpha\alpha \leq \alpha\leq \alpha p_\alpha$, using 4,1,6;
  \item $\proves[\hnetkat']\alpha\dup\leq\alpha\dup\sum_\beta \beta=\sum_\beta \alpha\dup\beta=\alpha\dup\alpha\leq\dup\alpha$, using 7,0,3, and symmetrically, $\proves[\hnetkat']\dup\alpha\leq\alpha\dup$.
  \end{itemize}
  Thus we have $\proves[\hnetkat']\hnetkat$.
  Conversely, $2,4,5,6,7$ already belong to $\hnetkat$ and $1,3$ are almost immediate consequences (using the sum axiom of $\hnetkat$). For $0$, for all $\alpha\neq\beta$, we have:
  $\proves[\hnetkat]\alpha\dup\beta=\dup\alpha\beta=\dup0=0$ and
  $\proves[\hnetkat]p_\alpha\beta=p_\alpha\alpha\beta=p_\alpha0=0$.
\end{proof}
As a consequence, $\hnetkat$ and $\hnetkat'$ reduce to each other via \cref{cor:red:id}.

Let us now study overlaps between the hypotheses of $\hnetkat'$ and build \cref{table:netkat}.
\begin{table}[t]
  \begin{align*}
    \begin{array}{lr@{\,}l|c|c|c|c|c|c|c|c}
      &&
      &0&1&2&3&4&5&6&7 \\
      \hline
      0:&\alpha\beta,\alpha\dup\!&\beta,p_\alpha\beta\leq 0~~(\alpha \neq \beta)&-&..&..&..&..&..&..&..\\
      1:&\alpha p_\alpha&\leq 1&&-&322&.&.&.&.&. \\
      2:&pq &\leq q&&&-&.&...&...&.&. \\
      3:&\alpha &\leq 1&&&&-&.&.&.&. \\
      4:&p_\alpha&\leq p_\alpha \alpha&&&&&-&...&54&44 \\
      5:&q&\leq pq&&&&&&-&.&54 \\
      6:&\alpha&\leq \alpha p_\alpha&&&&&&&-&66441 \\
      7:&1&\leq \sum\alpha&&&&&&&&-\\
      \hline
      &&&&&\times&.&.&.&6^=({<}6)^\star&7^=({<}7)^\star\\
    \end{array}
  \end{align*}
  \caption{Summary of commutations for NetKAT.}
  \label{table:netkat}
\end{table}
As before for KATF and KAT$^\circ$, the entries in the first line are all trivial by \cref{lem:dl:0}.
For the other lines, the cases where there are overlaps are obtained as follows.
\begin{proof}[Proof of non-trivial entries in \cref{table:netkat}]\mbox{}\\
  \begin{itemize}
  \item 12: the only overlap is $\tuple{p,q,1,1}$, for which we have
    \closeoverlapiii12322{p\alpha p_\alpha q}{pq}{q}%
    {p p_\alpha q}{p_\alpha q}
  \item 24: we only have the full overlap, for which we have
    \closeoverlapii2442{qp_\alpha}{p_\alpha}{p_\alpha\alpha}%
    {qp_\alpha\alpha}
  \item 25: we only have the full overlap, for which we have
    \closeoverlapii2552{rp}p{qp}%
    {rqp}
  \item 45: the only overlap is $\tuple{1,1,1,\alpha}$, for which we have
    \closeoverlapii4554{p_\alpha}{p_\alpha\alpha}{qp_\alpha\alpha}%
    {qp_\alpha}
  \item 46: the only overlap is $\tuple{1,1,p_\alpha,1}$, for which we have
    \closeoverlapii4654{p_\alpha}{p_\alpha\alpha}{p_\alpha\alpha p_\alpha}%
    {p_\alpha p_\alpha}
  \item 47: the only overlap is $\tuple{1,1,p_\alpha,\alpha}$, for which we have
    \closeoverlapii4744{p_\alpha}{p_\alpha\alpha}{p_\alpha(\sum\beta)\alpha}%
    {p_\alpha\alpha}
    (by choosing $\beta=\alpha$ in the sum to get $p_\alpha\alpha\alpha$)
  \item 57: the only overlap is $\tuple{1,1,p,q}$, for which we have
    \closeoverlapii5754q{pq}{p(\sum\beta)q}%
    {pq}%
    (by choosing $\beta=\alpha_p$ in the sum to get $p\alpha_pq$)
  \item 67: the only overlap is $\tuple{1,1,\alpha,p_\alpha}$, for which we have
    \closeoverlapiii67{66}{44}1\alpha{\alpha p_\alpha}{\alpha (\sum\beta) p_\alpha}%
    {\alpha p_\alpha p_\alpha}{\alpha p_\alpha\alpha \alpha p_\alpha}%
    (by choosing $\beta=\alpha$ in the sum to get $\alpha\alpha p_\alpha$)
    \qedhere
  \end{itemize}
\end{proof}

At this point, we face several difficulties:
\begin{enumerate}[(a)]
\item we do not have a reduction from 1 to the empty set: these hypotheses are not covered by \cref{lem:reds} (in fact such a reduction cannot exist: $\cl1$ does not preserve regularity);
\item we do not have a reduction from 2 to the empty set, yet: while single hypotheses from 2 fit \cref{lem:reds}\ref{it:red:ea}/\ref{it:red:aa}, we need to combine the corresponding reductions to get a reduction for 2. (In contrast, note that sets 3,4,5,6 do not pose such a problem since \cref{lem:reds}\ref{it:red:aw} already deals with sets of inequations of the form letter below word.)
\item the entry 322 at position $(1,2)$ in \cref{table:netkat} prevents us from using \cref{prop:gsos} directly: 322 is not below $2^=({<}2)^\star$, nor below $2^\star({<}2)^=$;
\end{enumerate}

For (a), we will enrich 1 with other hypotheses (a subset of 2, and 3) in such a way that we recover a set of hypotheses whose closure preserve regularity, and construct a reduction via \cref{lem:basic:e1}. We do so in \cref{ssec:netkat1}.
For (b), we will decompose 2 into smaller sets and use \cref{prop:cup} to compose the basic reductions provided by \cref{lem:reds}. We do so in \cref{ssec:netkat2}.
For (c), we will deal with $\cl{(1,2,3)}$ `manually' (\cref{lem:netkat:123}), in such a way that we can use \cref{prop:gsos} on \cref{table:netkat} by grouping (1,2,3) into a single set of hypotheses. We do so in \cref{ssec:netkat:org}, where we eventually assemble all the results to get completeness of NetKAT.

\subsection{A reduction including \texorpdfstring{$\alpha p_\alpha\leq 1$}{α p\_α ≤ 1}}
\label{ssec:netkat1}

Let us first provide a generic reduction, independently from NetKAT.
We locally reuse notations 1,2,3 in the lemma below on purpose, just in the context of that lemma and its proof.
\begin{lem}
  \label{lem:netkat:123:single}
  Let $a,b$ be distinct letters and pose
  \begin{itemize}
  \item $1\eqdef\set{ab\leq 1}$
  \item $2\eqdef\set{bb\leq b}$
  \item $3\eqdef\set{a\leq 1}$
  \end{itemize}
  We have that 1,2,3 reduces to 2,3, which reduces to the empty set.
\end{lem}
\begin{proof}
  For the first reduction, we apply \cref{lem:basic:e1} to $f=ab$ and $H'=2,3$, using $f'=a(a+b)^*$.
  For the two conditions about $f'$:
  \begin{enumerate}
  \item We have $\proves[2] b=b^+$ and
    $\proves[3] (a+b)^*\leq (1+b)^*=b^*\leq(a+b)^*$, from which we
    deduce $\proves[1,2,3]f'=a(a+b)^*=ab^*=a+ab^+=a+ab\leq 1+1=1$.
  \item $H_1\lang{f'}$ consists of words from $\lang {f'}$ where an occurrence of $ab$ has been inserted; such words always start with an $a$ so that they belong to $\lang{f'}$.
  \end{enumerate}
  It remains to show the preliminary closure condition.
  Like for entries (1,2) and (1,3) in \cref{table:netkat} analysing overlaps gives
 $12\subseteq21\cup 322$ and
 $13\subseteq31$,
  whence $1(2,3)\subseteq 21\cup 322\cup 31\subseteq (2,3)^\star 1^=$ (where we use additivity of $1$ in the first step) and thus $(1,2,3)^\star=(2,3)^\star 1^\star$ by \cref{lem:gsos:2}(1), as required. We have thus proved that 1,2,3 reduces to 2,3.

  For the second reduction, we have reductions to the empty set for 2 and 3 separately, by \cref{lem:reds}\ref{it:red:aa}/\ref{it:red:aw}. We moreover have $23\subseteq32$: there are no overlaps. We conclude by \cref{lem:cup:2,lem:gsos:2}.
\end{proof}

For $i=1,2,3$, write $i_\alpha$ for $i$ from the previous lemmas with $a=\alpha$ and $b=p_\alpha$.
Write $\hs_\alpha$ for $1_\alpha,2_\alpha,3_\alpha$, and $\hs$ for the union of all $\hs_\alpha$.
The set of hypotheses $\hs$ almost corresponds to 1,2,3: it just misses $qp\leq p$ for all $p\neq q$.

Now recall that 0 consists of the $\hnetkat'$ hypotheses of the shape $e\leq 0$.
\begin{lem}
  \label{lem:netkat:s0}
  $0,\hs$ reduces to the empty set.
\end{lem}
\begin{proof}
  We have individual reductions to the empty set for 0 (by \cref{lem:red:0}) and for each $\hs_\alpha$ (by \cref{lem:netkat:123:single}). By \cref{prop:cup,prop:gsos}, it suffices to check partial commutations. We have:
  \begin{itemize}
  \item for all $\alpha$, $0s_\alpha\subseteq \hs_\alpha^=0$ by \cref{lem:dl:0};
  \item for all $\alpha\neq\beta$, $\hs_\alpha \hs_\beta\subseteq \hs_\beta \hs_\alpha\cup 0$; indeed we have
    \begin{itemize}
    \item one overlap between $1_\alpha$ and $1_\beta$: $\tuple{\beta,p_\beta,1,1}$, which we can solve as follows:
      \closeoverlapi{1_\alpha}{1_\beta}0{\beta \alpha p_\alpha p_\beta}{\beta p_\beta}1
    \item one overlap between $1_\alpha$ and $2_\beta$: $\tuple{p_\beta,p_\beta,1,1}$, which we can solve as follows:
      \closeoverlapi{1_\alpha}{2_\beta}0{p_\beta \alpha p_\alpha p_\beta}{p_\beta p_\beta}{p_\beta}
    \item one overlap between $3_\alpha$ and $1_\beta$: $\tuple{\beta,p_\beta,1,1}$, which we can solve as follows:
      \closeoverlapi{3_\alpha}{1_\beta}0{\beta \alpha p_\beta}{\beta p_\beta}1
    \item one overlap between $3_\alpha$ and $2_\beta$: $\tuple{p_\beta,p_\beta,1,1}$, which we can solve as follows:
      \closeoverlapi{3_\alpha}{2_\beta}0{p_\beta \alpha p_\beta}{p_\beta p_\beta}{p_\beta}
    \item no overlap between the remaining pairs.
    \end{itemize}
  \end{itemize}
  We should thus place 0 first, followed by the $\hs_\alpha$ in any order.
  This concludes the proof.
\end{proof}

\subsection{A single reduction for \texorpdfstring{$\set{qp\leq p\mid p,q\in P}$}{\{qp≤ p| p,q∈ P\}}}
\label{ssec:netkat2}

Let us decompose the set of hypotheses 2 into several subsets: for $p\in P$, set
\begin{itemize}
\item $2_p\eqdef \set{qp\leq p \mid q \in P}$
\item $2_{=p}\eqdef \set{pp\leq p}$
\item $2_{\neq p}\eqdef \set{qp\leq p \mid q \in P, q\neq p}$
\end{itemize}
We have $2_p=2_{=p},2_{\neq p}$ and 2 is the union of all $2_p$s. Moreover, we have reductions to the empty set for each $2_{=p}$ by \cref{lem:reds}\ref{it:red:aa}, and for each $2_{\neq p}$ by \cref{lem:reds}\ref{it:red:ea} (by grouping the hypotheses in those latter sets into a single inequation $(\sum_{p\neq q} q)p\leq p$). We combine those elementary reductions in two steps.

\begin{lem}\label{lem:netkat2p}
  For all $p\in P$, $2_p$ reduces to the empty set.
\end{lem}
\begin{proof}
  There is a single overlap between $2_{=p}$ and $2_{\neq p}$, which we solve as follows:
  \closeoverlapii{2_{=p}}{2_{\neq p}}{2_{\neq p}}{2_{=p}}{qpp}{qp}p{qp}
  Thus we have $2_{=p}2_{\neq p}\subseteq 2_{\neq p}2_{=p}$, and we conclude with \cref{lem:cup:2,lem:gsos:2}.
\end{proof}
(Alternatively, it is also easy to build directly a homomorphic reduction from $2_p$ to the empty set by adapting the proof of \cref{lem:reds}\ref{it:red:ea}.)

\begin{lem}\label{lem:netkat2}
  $2$ reduces to the empty set.
\end{lem}
\begin{proof}
  For $p\neq q$, there is a single overlap between $2_q$ and $2_p$, which we solve as follows:
  \closeoverlapii{2_q}{2_p}{2_p}{2_p}{rqp}{qp}p{rp}
  Thus we have $2_q2_p\subseteq 2_p2_q\cup 2_p2_p\subseteq 2_p^\star2_q^=$, and we conclude with \cref{prop:cup,prop:gsos} by using any ordering of the $2_p$s.
\end{proof}

\subsection{Organising the closures}
\label{ssec:netkat:org}

We proceed in two steps to organise the closures.
\begin{lem}
  \label{lem:netkat:123}
  We have $\cl{(1,2,3)}=\cl3\cl2\cl1$.
\end{lem}
\begin{proof}
  We follow the same path as in the proof of \cref{lem:netkat:123:single}:
  from \cref{table:netkat} we first deduce
  $1(2,3)=12\cup 13\subseteq 332\cup31 \subseteq (2,3)^\star 1^=$ (since 1 is linear by \cref{lem:hyp:add} for the first step), whence
  $(1,2,3)^\star=(2,3)^\star 1^\star$ by \cref{lem:gsos:2}(1).
  Then we have $(2,3)^\star=3^\star 2^\star$ by \cref{lem:gsos:2}, since $23\subseteq 32$.
  The announced equality follows.
\end{proof}

\begin{lem}
  \label{lem:netkat:com}
  We have $\cl{\hnetkat'}=\cl{(0,\hs,2,3,4,5,6,7)}=\cl{(3,7)}\cl6\cl5\cl4\cl3\cl2\cl{(0,\hs)}$.
\end{lem}
\begin{proof}
  We have
  \newlength{\abovedisplayskipsave}
  \setlength{\abovedisplayskipsave}{\abovedisplayskip}
  \setlength{\abovedisplayskip}{0pt}
  \vspace{-\baselineskip}
  \begin{align*}
     \cl{\hnetkat'}
   &= \cl{(0,1,2,3,4,5,6,7)}\\
   \tag{by \cref{prop:gsos}}
   &= \cl7\cl6\cl5\cl4\cl{(1,2,3)}\cl0\\
   \tag{by \cref{lem:netkat:123}}
   &= \cl7\cl6\cl5\cl4\cl3\cl2\cl1\cl0\\
   &\subseteq \cl{(3,7)}\cl6\cl5\cl4\cl3\cl2\cl{(0,\hs)}\\
   &\subseteq \cl{(0,\hs,2,3,4,5,6,7)} %
   = \cl{\hnetkat'}
  \end{align*}
  \setlength{\abovedisplayskip}{\abovedisplayskipsave}
  The first equality is by definition, the last one comes from the equality of the underlying sets (since $1,2,3 = \hs,2$). The last two inclusions follow from basic closure properties and $0,1\subseteq 0,\hs$. The application of \cref{prop:gsos} is justified by \cref{table:netkat}, where we consider $1,2,3$ as a single set of hypotheses. This amounts to merging the corresponding lines and columns in the table, and we observe that all subsequent columns satisfy the requirements (choosing the second alternative for columns 6,7).
\end{proof}

At this point we can easily conclude.
\begin{thm}
  \label{thm:netkat}
  $\hnetkat$ reduces to the empty set; $\KA_\hnetkat$ is complete and decidable.
\end{thm}
\begin{proof}
  $\hnetkat$ reduces to $\hnetkat'=(0,\hs,2,3,4,5,6,7)$, and we have reductions to the empty set for $0,\hs$ by \cref{lem:netkat:s0}, for 2 by \cref{lem:netkat2}, and for $3$, $4$, $5$, $6$, and $(3,7)$ by \cref{lem:reds}. Given \cref{lem:netkat:com}, \cref{prop:cup} applies.
\end{proof}

\section{Related work}
\label{sec:related}

There is a range of papers on completeness and decidability of Kleene algebra together with specific forms of hypotheses, starting with~\cite{cohen94:ka:hypotheses}.
The general case of Kleene algebra with hypotheses, and reductions to prove completeness, has been studied recently in~\cite{KappeB0WZ20,dkpp:fossacs19:kah,KozenM14}.
The current paper combines and extends these results, and thereby aims to provide a comprehensive overview and a showcase of how to apply these techniques to concrete case studies (KAT, KAO, NetKAT and the new theories KATF, KAT$^\circ$ and KAPT).
Below, we discuss each of these recent works in more detail.

Kozen and Mamouras~\cite{KozenM14} consider restricted forms of hypotheses in terms of rewriting systems, and provide reductions for equations of the form $1 = w$ and $a = w$ (cf. Lemma~\ref{lem:reds}). Their general results cover completeness results which instantiate to KAT and NetKAT. In fact, the assumptions made in their technical development are tailored towards these cases; for instance, their assumption $\alpha \beta \leq \bot$ (in Assumption 2) would have to be dropped to consider KAPT. The current paper focuses on generality and how to construct reductions in a modular way.

Doumane et al.~\cite{dkpp:fossacs19:kah} define the language interpretation of regular expressions 
for an arbitrary set of hypotheses via the notion of $H$-closure (\cref{def:clo}), and study (un)decidability of the (in)equational theory of $\KA_H$, i.e., 
the problem of checking $\KAproves[H] e \leq f$, for various types of hypotheses $H$. 
In particular,
they construct a reduction for hypotheses of the form $1 \leq \sum_{a \in S} a$ (cf. Lemma~\ref{lem:reds}).
A first step towards modularity may also be found in~\cite[Proposition~3]{dkpp:fossacs19:kah}.

Kapp\'{e} et al.~\cite{KappeB0WZ20} study hypotheses on top of bi-Kleene algebra, where the canonical interpretation is based on pomset languages, and ultimately prove completeness of \emph{concurrent Kleene algebra with observations}; many of the results there apply to the word case as well. We follow this paper for the basic definitions and results about reductions, with a small change in the actual definition of a reduction (Remark~\ref{rem:red:ckah}). Compositionality in the sense of \cref{ssec:compo} is treated in Kapp\'{e}'s PhD thesis~\cite{kappethesis}. In the current paper we systematically investigate tools for combining reductions,
based on lattice theory, proposing a number of new techniques (e.g.~\cref{lem:gsos:2} and \cref{ssec:more:gsos,ssec:more:overlaps,ssec:univ}).
Further, we highlight the word case in this paper (as opposed to the pomset languages in concurrent Kleene algebra), by showcasing several examples.

\section{Conclusions and future work}
\label{sec:future}

We presented a general toolbox for proving completeness of Kleene algebra with hypotheses.
While our examples demonstrate the rather wide applicability of our techniques, there are natural extensions that we have not covered here. For instance KAT+B!~\cite{kat:bbang} and KAO with a full element and/or converse seem to fit our framework. We hope we have provided enough tools and examples such that these theories can be investigated easily and in a principled way in the future. 

Of course, there are Kleene algebra extensions that we cannot cover, a priori. For instance, for \emph{action algebras}~\cite{Pratt90} and \emph{action lattices}~\cite{K92a}, it is not clear how to interpret the new operations as letters with additional structure. 
Finding reference models and proving completeness for such theories remains an important challenge.

There are also variations of Kleene algebra where some of the axioms weakened or removed, such as left-handed Kleene algebra~\cite{KozenS12,ddp:lpar18:lefthanded} (where one of the star-induction rules is removed) or Kleene algebra with abnormal termination~\cite{Mamouras17:abnormal} (where the axiom $x0=0$ is no longer valid). In the former case, the equational theory remains the same, so that except for some of the basic reductions (e.g., \cref{lem:reds}\ref{it:red:ea}), most of the framework we developed here can be reused. In the latter case instead, it is not clear how to proceed since the remaining axioms are no longer complete w.r.t. the language interpretation.

\medskip

The general theory proposed here results in completeness with respect to a canonical
language model, defined via language closure. In several instances however, there are other
reference models for which we would like to obtain completeness. For instance, in the case of KAT, 
we provided a separate argument to relate the canonical model defined via closure to the model
of guarded strings, which is a standard model of KAT. For NetKAT, there is a similar model of guarded
strings, which we have not considered in our treatment. It would be interesting to identify a 
common pattern, and to try and develop generic tools that help transporting our results
to such reference models.

Another direction of future work is that of decidability and complexity. Whenever our reduction technique yields a proof of completeness, and the reduction itself is computable (as it is in all examples), this immediately entails a decidability result, via decidability of KA (\cref{thm:dec}). However, in general this is far from an efficient procedure: computing the reduction itself can lead to a blow-up. Developing tools that help to obtain more efficient decision procedures is left for future work. This problem may well be related to the problem discussed in the previous paragraph, on identifying suitable reference models.

\paragraph*{Acknowledgements}

We would like to thank the reviewers for all their comments, as well as Pierre Goutagny for suggesting an optimisation of our definition of reduction, and spotting an error in the previous version of this article: the need for affine functions to cover the case of constant functions.

\bibliographystyle{alphaurl}
\bibliography{pous,refs}

\appendix

\section{On least closures in complete lattices}
\label{app:closures}

For a set of hypotheses $H$, the $H$-closure $\cl H$ is formally a closure operator: for all languages $L$
we have $L \subseteq \cl H(L)$ and $\cl H \cl H (L) \subseteq \cl H(L)$. 
In this appendix, we use basic lattice theory~\cite{DaveyPriestley90} to establish various properties of such closure operators.

\medskip

Let $\tuple{X,{\leq},\bigvee}$ be a complete lattice. We write $x+y$ for binary joins $\bigvee\set{x,y}$.
We write $1$ for the identity function on $X$. Given two functions $f,g$, we write $fg$ for their composition: $fg(x)=f(g(x))$, and $f\leq g$ when $\forall x, f(x)\leq g(x)$. The finite iterations $f^i$ of a function $f$ are defined by induction on $i\in\NN$: $f^0=1$ and $f^{i+1}=ff^i$.
Functions on $X$, ordered pointwise as above, form a complete lattice where suprema are also computed pointwise.

\medskip

A function $f$ is \emph{monotone} if $\forall x,y, x\leq y \Rightarrow f(x)\leq f(y)$.
By Knaster-Tarski's theorem, every monotone function $f$ admits a \emph{least (pre)fixpoint} $\mu f$, which satisfies $f(\mu f)\leq \mu f$ and the \emph{induction principle}: $\forall x, f(x)\leq x \Rightarrow \mu f\leq x$.

\subsection{Least closures}

A \emph{closure} is a monotone function $c$ such that $1\leq c$ and $cc \leq c$.

\medskip

Given a monotone function $s$ and an element $x\in X$, we write $\cl s(x)$ for the least (pre)fixpoint of $s$ above $x$: $\cl s(x)\eqdef \mu(\lambda y.x+s(y))$. This definition gives, for all $x,y$,
\begin{align}
  \label{cpref}
  x+s(s^\star(x))\leq s^\star(x)\\
  \label{cind}
  x+s(y)\leq y \Rightarrow s^\star(x)\leq y
\end{align}
Equivalently we have that for all functions $g,h$,
\begin{align}
  \label{cpref:fun}
  \tag{\ref{cpref}'}
  1+ss^\star\leq s^\star\\
  \label{cind:fun}
  \tag{\ref{cind}'}
  g+sh\leq h \Rightarrow s^\star g\leq h
\end{align}

\begin{lem}
  \label{lem:mu:mon}
  If $f,g$ are two monotone functions such that $f\leq g$, then $\mu f\leq \mu g$.
\end{lem}
\begin{proof}
  Under the assumption, every pre-fixpoint of $g$ is a pre-fixpoint of $f$.
\end{proof}

\begin{prop}
  \label{prop:least:closure}
  For every monotone function $s$, $s^\star$ is the least closure above $s$.
\end{prop}
\begin{proof}
  We first prove that $s^\star$ is a closure:
  \begin{itemize}
  \item monotonicity follows from Lemma~\ref{lem:mu:mon};
  \item $1\leq s^\star$ follows from~\eqref{cpref:fun};
  \item $s^\star s^\star\leq s^\star$ follows from $s^\star+ss^\star\leq s^\star$ by~\eqref{cind:fun}, which holds by~\eqref{cpref:fun};
  \end{itemize}
  Moreover, $s^\star$ is above $s$: we have $s\leq ss^\star\leq s^\star$ by $1 \leq s^\star$, monotonicity of $s$ and \eqref{cpref:fun}.

  Now, if $c$ is a closure above $s$, then $s^\star\leq c$ follows from $1+sc\leq c$ by~\eqref{cind:fun}, which holds thanks to the assumption $s\leq c$ and the fact that $c$ is a closure.
\end{proof}

\begin{lem}
  \label{lem:clot:mon}
  Let $s,t$ be two monotone functions. If $s\leq t$ then $s^\star\leq t^\star$.
\end{lem}
\begin{proof}
  Direct consequence of Lemma~\ref{lem:mu:mon}.
\end{proof}

\begin{prop}
  \label{prop:gsos:dl:2:gen}
  If $s,t$ are monotone functions such that $s^\star t^\star\leq t^\star s^\star$, then $(s+t)^\star=t^\star s^\star$.
\end{prop}
\begin{proof}
  Set $c=t^\star s^\star$.
  First $c$ is a closure: we have $1\leq t^\star\leq t^\star s^\star$ and $t^\star s^\star t^\star s^\star \leq t^\star t^\star s^\star s^\star\leq t^\star s^\star$.
  Since $c$ is above both $s$ and $t$, we get $(s+t)^\star\leq c$ by Proposition~\ref{prop:least:closure}.
  Finally, $c=t^\star s^\star\leq (s+t)^\star(s+t)^\star\leq (s+t)^\star$ using Lemma~\ref{lem:clot:mon} twice for the first inequality.
\end{proof}

\begin{prop}
  \label{prop:gsos:dl:gen}
  Let $s_1,\dots,s_n$ be monotone functions;
  write $c_j$ for $s_j^\star$, and $c_{<j}$ for $\paren{\bigvee_{1\leq i<j}s_i}^\star$.
  Suppose that for all $1<j\leq n$, we have $c_{<j}c_j \leq c_jc_{<j}$.
  Then $c_{<n+1}=c_n\dots c_1$.
\end{prop}
(Note that for $n=2$, the statement amounts to Proposition~\ref{prop:gsos:dl:2:gen}.)
\begin{proof}
  We prove by induction on $j$ that for all $1\leq j\leq n$, $c_{<j+1}=c_j\dots c_1$.
  The case $j=1$ is trivial. For the inductive case, suppose $1<j\leq n$ and
  let $s=\bigvee_{1\leq i<j} s_i$. 
  By assumption, we have $s^\star s_j^\star=c_{<j}\,c_j\leq c_j\,c_{<j}=s_j^\star s^\star$.
  Therefore, we deduce
  \begin{align*}
    \tag{by definition}
    c_{<j+1} & = (s+s_j)^\star \\
    \tag{by Proposition~\ref{prop:gsos:dl:2:gen}}
             & = s_j^\star s^\star \\
    \tag{by definition}
             & = c_j\,c_{<j} \\
    \tag*{(by induction)\qedhere}
             & = c_j\,c_{j-1}\dots c_1
  \end{align*}
\end{proof}

\begin{prop}
  \label{prop:clo:iter:l}
  Let $s,f$ be monotone functions, and $c$ a closure. We have:
  \begin{align*}
    sf\leq fc \Rightarrow s^\star f\leq fc
  \end{align*}
\end{prop}
\begin{proof}
  By \eqref{cind:fun}, it suffices to show $f+sfc\leq fc$.
  We have $f\leq fc$ by monotonicity of $f$ and $1\leq c$.
  We have $sfc\leq fcc\leq fc$ by assumption, monotonicity of $f$, and $cc\leq c$.
\end{proof}

\medskip
A function $f$ admits another function $f^\sharp$ as \emph{upper adjoint} if $\forall x,y,~f(x)\leq y \Leftrightarrow x \leq f^\sharp(y)$. Equivalently, for all functions $g,h$,
\begin{align}
  \label{eq:adj}
  fg \leq h \Leftrightarrow g \leq f^\sharp h
\end{align}
In such a situation, $f$ and $f^\sharp$ are monotone, and we have
\begin{align}
  \label{eq:adj:unit}
  ff^\sharp \leq 1 \leq f^\sharp f
\end{align}
A function $f$ is \emph{linear} if it preserves all joins: $\forall Y\subseteq X, f(\bigvee_{y\in Y} y) = \bigvee_{y\in Y} f(y)$; it is \emph{affine} if it preserves all non-empty joins (i.e.,\ $Y\neq\emptyset$ in the previous formula).

Linear functions are affine; affine functions are monotone; the set of linear (resp. affine) functions is closed under composition and arbitrary joins; all constant functions are affine, and only the least one is linear. Moreover, we have the following characterisations:
\begin{lem}\label{lem:aff:lin}\hfill
  \begin{enumerate}[(i)]
  \item A function is linear if and only if it admits an upper adjoint.
  \item A function is affine if and only if it is the join of a constant and a linear function.
  \end{enumerate}
\end{lem}
\begin{proof}
  The first item is~\cite[Proposition~7.34]{DaveyPriestley90}.
  For the direct implication in the second one, let $f$ be an affine function. Set
  $a=f(\bot)$, and define
  $f'(x)=\begin{cases}
    \bot\text{ if }x=\bot\\
    f(x)\text{ otherwise }
  \end{cases}$. \\
  We have $f(x)=a+f'(x)$ for all $x$, and $f'$ is linear. The converse implication is immediate.
\end{proof}

We obtain a symmetrical version of \cref{prop:clo:iter:l} by restricting to affine functions.
\begin{prop}
  \label{prop:clo:iter:r}
  Let $s$ be a monotone function, $f$ an affine function, and $c$ a closure.
  We have:
  \begin{align*}
    fs\leq cf \Rightarrow fs^\star\leq cf
  \end{align*}
\end{prop}
\begin{proof}
  We first prove the statement when $f$ is linear.
  In such a case $f$ has an upper adjoint $f^\sharp$, and we have:
  \begin{align}
    \tag{by assumption}
    fsf^\sharp cf
    &\leq cff^\sharp cf \\
    \tag{by~\eqref{eq:adj:unit}}
    &\leq ccf \\
    \tag{$c$ a closure}
    &\leq cf
  \end{align}
  By~\eqref{eq:adj}, we deduce $sf^\sharp cf\leq f^\sharp cf$.
  We also have $1\leq f^\sharp f \leq f^\sharp cf$ by~\eqref{eq:adj:unit} and $1\leq c$.
  Therefore we get $s^\star\leq f^\sharp cf$ by~\eqref{cind:fun}, which means $fs^\star\leq cf$ by~\eqref{eq:adj}.

  Now, if $f$ is only affine, then $f=x\mapsto a+f'(x)$ for some constant $a$ and some linear function $f'$,  by \cref{lem:aff:lin}.
  The assumption $fs\leq cf$ at bottom gives us $a\leq c(a)$.
  Define $c'(x)=c(a +x)$. This function is a closure: we have $1\leq c\leq c'$ and
  $c'(c'(x))=c(a+c(a+x))\leq c(c(a)+c(a+x))\leq c(c(a+x))\leq c(a+x)=c'(x)$.
  Also observe that $cf=c'f'$.
  We have $f's\leq fs\leq cf=c'f'$ so that we can apply the linear version of the lemma to $s$, $f'$ and $c'$, to deduce $f's^\star\leq c'f'=cf$. Combined with $a\leq c(a)$, it follows that $fs^\star\leq cf$.
\end{proof}

\begin{rem}
  \label{rem:hetero}
  \cref{prop:clo:iter:l,prop:clo:iter:r} easily generalise to the case where $s$ and $c$ are functions on two distinct lattices $X$ and $Y$, and $f$ is a function between those lattices (from $Y$ to $X$ for \cref{prop:clo:iter:l}, and from $X$ to $Y$ for \cref{prop:clo:iter:r}.
  We actually use this more general form of the latter proposition in the proof of \cref{prop:hred}.
\end{rem}

\begin{lem}
  \label{lem:clo:add}
  If $s$ is affine, so is $s^\star$.
\end{lem}
\begin{proof}
  As above, we first prove statement for linear functions $s$, using the characterisation of linear functions as those admitting an upper adjoint (\cref{lem:aff:lin}).
  
  When given a monotone function $g$, let $g^\circ(y)$ be the greatest \mbox{(post-)}fixpoint of $\lambda x.y\wedge g(x)$. This construction is dual to the construction of the least closure above a monotone function, and yields the ``largest coclosure'' below $g$.

  Assume that $s$ admits $t$ as upper adjoint; we show that $s^\star$ admits $t^\circ$ as upper adjoint.
  Suppose $s^\star(x) \leq y$. Since $s(s^\star(x))\leq s^\star(x)$, the adjoint property~\eqref{eq:adj} gives $s^\star(x)\leq t(s^\star(x))$. Therefore, $s^\star(x)$ is a post-fixpoint of $\lambda x.y\wedge t(x)$, so that $s^\star(x)\leq t^\circ(y)$. Since $x\leq s^\star(x)$, we deduce $x\leq t^\circ(y)$ by transitivity.
  The converse implication ($\forall x,y,~x\leq t^\circ(y)\Rightarrow s^\star(x)\leq y$) holds by duality.

  Now, if $s$ is only affine, then $s=a+f$ for some constant $a$ and some linear function $f$,  by \cref{lem:aff:lin}. We prove that $s^\star$ is affine by showing that $s^\star=f^\star(a)+f^\star$.
  \begin{itemize}
  \item We use~\eqref{cind:fun} to show $s^\star\leq f^\star(a)+f^\star$.
    Indeed, we have $1\leq f^\star\leq f^\star(a)+f^\star$, and $s(f^\star(a)+f^\star)=a+f(f^\star(a))+ff^\star\leq f^\star(a)+f^\star$.
  \item For the converse inequation, 
    we have $f^\star(a)\leq s^\star(\bot)$ by~\eqref{cind}, since $a\leq s(\bot)\leq s^\star(\bot)$ and $f(s^\star(\bot))\leq s(s^\star(\bot))\leq s^\star(\bot)$; and $f^\star\leq s^\star$ follows from $f\leq s$.
    \qedhere
  \end{itemize}
\end{proof}

Given a function $f$, we write $f^=$ for the function $f+1$ (i.e., $f^=(x)=f(x)+x$).
\begin{lem}
  \label{lem:clo:cup1}
  For every monotone function $s$, we have $(s^=)^\star=s^\star$.
\end{lem}
\begin{proof}
  Consequence of Proposition~\ref{prop:gsos:dl:2:gen}, since $1^\star=1$.
\end{proof}

\begin{prop}
  \label{prop:clo:iter}
  Let $s$ be affine and $t$ be monotone.
  \begin{enumerate}
  \item If $st\leq t^\star s^=$ then $s^\star t^\star\leq t^\star s^\star$.
  \item If $st\leq t^= s^\star$ then $s^\star t^\star\leq t^\star s^\star$.
  \end{enumerate}
\end{prop}
\begin{proof}
  We first prove the statements where we remove the $\cdot^=$ in the assumptions:
  \begin{enumerate}
  \item assuming $st\leq t^\star s$, we deduce $st^\star\leq t^\star s$ by Proposition~\ref{prop:clo:iter:r} ($s$ being affine, and $t^\star$ being a closure), whence $st^\star\leq t^\star s^\star$ since $s\leq s^\star$, and finally $s^\star t^\star\leq t^\star s^\star$ by Proposition~\ref{prop:clo:iter:l} ($s^\star$ being a closure).
  \item assuming $st\leq t s^\star$, we deduce $s^\star t\leq t s^\star$ by Proposition~\ref{prop:clo:iter:l} ($s^\star$ being a closure), whence $s^\star t\leq t^\star s^\star$ since $t\leq t^\star$, and finally $s^\star t^\star\leq t^\star s^\star$ by Proposition~\ref{prop:clo:iter:r} ($s^\star$ being affine by Lemma~\ref{lem:clo:add}, and $t^\star$ being a closure).
  \end{enumerate}
  Then we add the $\cdot^=$ back in the assumptions:
  \begin{enumerate}
  \item assume we have $st\leq t^\star s^=$; $s^=$ is affine, and we have
    $s^=t=st+t\leq t^\star s^=+t=t^\star s^=$. We can thus apply the first result above to $s^=$ and $t$, so as to obtain $(s^=)^\star t^\star\leq t^\star (s^=)^\star$.
  \item assume we have $st\leq t^= s^\star$; we have
    $st^==st+s$ since $s$ is affine, and thus $st^==st+s\leq t^= s^\star+s=t^= s^\star$. We can thus apply the second result above to $s$ and $t^=$, so as to obtain $s^\star (t^=)^\star\leq (t^=)^\star s^\star$.
  \end{enumerate}
  In both cases, we conclude by Lemma~\ref{lem:clo:cup1}.
\end{proof}

\begin{cor}
  \label{cor:gsos:gen}
  Let $s_1,\dots,s_{n-1}$ be affine functions, and $s_n$ a monotone function.
  Write $c_j$ for $s_j^\star$, $s_{<j}$ for $\bigvee_{1\leq i<j}s_i$, and $c_{<j}$ for $s_{<j}^\star$, like in Proposition~\ref{prop:gsos:dl:gen}.

  \noindent
  If for all $1<j\leq n$ we have either
  $\begin{cases}
  (1)~\forall i<j,~s_i\,s_j\leq c_j\,s_{<j}^=\,, \text{ or}\\
  (2)~\forall i<j,~s_i\,s_j\leq s_j^=\,c_{<j}\,.
  \end{cases}$
  then $c_{<n+1}=c_n\dots c_1$.
\end{cor}
\begin{proof}
  For all $1<j\leq n$, we have that $s_{<j}$ is affine,
  and we deduce by summing the hypotheses that either $s_{<j}s_j\leq c_j\,s_{<j}^=$ or $s_{<j}s_j\leq s_j^=\,c_{<j}$; in both cases, Proposition~\ref{prop:clo:iter} applies so that we deduce
  $c_{<j}c_j\leq c_jc_{<j}$. We conclude with Proposition~\ref{prop:gsos:dl:gen}.
\end{proof}

\subsection{Contextual functions}
\label{app:con}

Let us assume that $X$ actually is a \emph{quantale}: a complete lattice which comes with a monoid $\tuple{X,\cdot,1}$ whose multiplication $(\cdot)$ distributes over all joins.

\begin{defi}
  \label{def:con:gen}
  A monotone function $f$ on $X$ is \emph{contextual} if for all $x,y,z\in X$, we have $x\cdot f(y)\cdot z\leq f(x\cdot y\cdot z)$.
\end{defi}

\begin{lem}
  \label{lem:con:clo:gen}
  Let $c$ be a closure. The function $c$ is contextual iff for all $x,y\in X$, we have $c(x)\cdot c(y)\leq c(x\cdot y)$.
\end{lem}
\begin{proof}
  If $c$ is a contextual closure, then $c(x)\cdot c(y)\leq c(c(x)\cdot y)\leq c(c(x\cdot y))=c(x\cdot y)$ using contextuality twice.
  Conversely, we have $x\cdot c(y)\cdot z\leq c(x)\cdot c(y)\cdot c(z)\leq c(x\cdot y\cdot z)$ using extensivity (i.e., $1 \leq c$) twice, and then twice the assumption.
\end{proof}

\begin{fact}
  Contextual functions are closed under composition and arbitrary joins, and contain the identity function.
\end{fact}

\begin{lem}
  \label{lem:clo:con}
  If $s$ is contextual then so is $s^\star$.
\end{lem}
\begin{proof}
  For $x,z\in X$, let $f_{x,z}: y\mapsto x\cdot y \cdot z$.
  A function $s$ is contextual iff for all $x,z$, $f_{x,z}s\leq sf_{x,z}$.
  The functions $f_{x,z}$ are linear, so that we may apply Proposition~\ref{prop:clo:iter:r} to deduce that when $s$ is contextual, we have $f_{x,z}s^\star\leq s^\star f_{x,z}$ for all $x,z$, i.e., $s^\star$ is contextual.
\end{proof}

Now consider the quantale of languages.

\begin{lem}
  \label{lem:con:iff}
  A monotone function $f$ on languages is contextual iff for all words $l,r$ and languages $K$, we have $l\cdot f(K)\cdot r\subseteq f(l\cdot K\cdot r)$.
\end{lem}
\begin{proof}
  The left-to-right implication is trivial: take singleton languages $\set{l}$ and $\set{r}$ for $x$ and $z$.
  For the converse implication, suppose that $u\in L\cdot f(K)\cdot R$: $u=lxr$ for words $l,x,r$ respectively in $L,f(K)$ and $R$. We have in particular $lxr\in l\cdot f(K)\cdot r$. We deduce $lxr\in f(l\cdot K\cdot r)$ by assumption, whence $lxr\in f(L\cdot K\cdot R)$ by monotonicity.
\end{proof}

\begin{lem}
  \label{lem:hyp:con}
  Let $H$ be a set of hypotheses. The functions $H$ and $H^\star$ are contextual.
\end{lem}
\begin{proof}
  It suffices to prove that $H$ is contextual by \cref{lem:clo:con}, which is easy.
\end{proof}

\begin{lem}
  \label{lem:con:clo:app}
  Let $H$ be a set of hypotheses.
  For all languages $L,K$ we have
  \begin{align*}
  \cl H(L)\cdot \cl H(K)\subseteq \cl H(L\cdot K)\,.
  \end{align*}
\end{lem}
\begin{proof}
  By \cref{lem:hyp:con,lem:con:clo:gen}.
\end{proof}

\begin{lem}
  \label{lem:con:clo:str}
  Let $H$ be a set of hypotheses.
  For all languages $L$ we have
  \begin{align*}
    (\cl H(L))^*\subseteq\cl H(L^*)\,.
  \end{align*}
\end{lem}
\begin{proof}
  It suffices to show $1+\cl H(L)\cdot\cl H(L^*)\subseteq \cl H(L^*)$, which follows from \cref{lem:con:clo}.
\end{proof}

\section{Direct soundness proof}
\label{app:soundness}

We give in this appendix a direct proof of soundness (\cref{thm:soundness}, \cite[Theorem~2]{dkpp:fossacs19:kah}).
It relies on the following lemma from~\cite{dkpp:fossacs19:kah}.

\begin{lemC}[{\cite[Lemma~2]{dkpp:fossacs19:kah}}]
  \label{lem:presoundness}
  Let $H$ be a set of hypotheses.
  For all languages $L,K$, we have
  \begin{enumerate}
  \item $\cl H(L+K)=\cl H(\cl H(L)+\cl H(K))$\,,
  \item $\cl H(L\cdot K)=\cl H(\cl H (L)\cdot\cl H(K))$\,,
  \item $\cl H(L^*)=\cl H(\cl H(L)^*)$\,.
  \end{enumerate}
\end{lemC}
\begin{proof}
  Via basic closure properties, \cref{lem:con:clo:app}, and \cref{lem:con:clo:str}.
\end{proof}
Before proving soundness, let us also recall the following equivalence, which we use implicitly in the proof below: for all languages $L,K$,
\begin{align*}
  \cl H(L+K)=\cl H(K) \quad\text{iff}\quad L\subseteq \cl H(K)
\end{align*}
(This is useful to deal with inequations, since $e\leq f$ is defined as $e+f=f$.)

\printProofs[snd]

\section{Guarded string interpretation of \texorpdfstring{$\sem\hkat-$}{kat*〚−〛}}
\label{app:kat:gs}
\printProofs[kat]

\end{document}